\documentclass[letterpaper,12pt]{article}
\usepackage{diagbox}
\usepackage{latexsym,amssymb,amsmath, bm}
\usepackage{mathtools}
\usepackage{amsthm}
\usepackage{lipsum,graphicx,multicol}
\usepackage{setspace,color}
\usepackage[sort&compress]{natbib}
\usepackage{tcolorbox}
\usepackage{subcaption}
\usepackage{float}
\usepackage[normalem]{ulem}
\usepackage{authblk}
\usepackage{leftidx}
\usepackage{dsfont}
\usepackage{geometry}
\usepackage{bbm}
\usepackage[title,titletoc]{appendix}
\usepackage{actuarialangle}
\usepackage{lscape}
\usepackage{afterpage}
\usepackage{relsize}
\usepackage{scalerel}
\usepackage[colorlinks=true,allcolors=blue, linkcolor=red]{hyperref}
\allowdisplaybreaks

\usepackage{lineno}

\usepackage{accents}

\newcommand\medtilde[1]{\vstretch{0.8}{\hstretch{0.8}{\widetilde{#1}}}}

\newcommand{\medhat}[1]{{
  \ThisStyle{%
    \ooalign{%
      $\SavedStyle#1$\cr
      \hidewidth\kern 0.5pt \raisebox{-0.75pt}{$\SavedStyle\hstretch{0.75}{\widehat{\hstretch{1.333}{\phantom{\SavedStyle #1}}}}$}\kern -0.5pt \hidewidth\cr
    }%
  }%
  \kern 0.75pt 
}}

\def \one {\bm{1}}

\def \eqd {\buildrel d \over =}

\newcommand{\Fbar}{\,\overline{\kern-0.175em F\kern-0em}}

\def \gh {\medhat{g}}
\def \th {\medhat{t}}
\def	 \phih {\widehat{\Phi}}
\def \psih {\widehat{\Psi}}

\def \hh {\boldsymbol{h}}

\def \i {\mathrm{i}}

\def \dom {\mathrm{dom}}

\def \D {\mathfrak{S}_{u^{\ast}}}

\def \NN {\mathbb{N}}

\def \PP {\mathbb{P}}

\def \RR {\mathbb{R}}

\def \T {\mathfrak{T}_{u^{\ast}}}

\def \XX {\boldsymbol{X}}

\def\VaR{{\mathrm{VaR}}}
\def\CTE{{\mathrm{CTE}}}
\def\e{{\mathbb E}}

\def\Cov{{\mathrm{Cov}}}
\def\Var{{\mathrm{Var}}}

\numberwithin{equation}{section} 
\newtheorem{theorem}{Theorem}[section]
\newtheorem{definition}[theorem]{Definition}

\newtheorem{example}[theorem]{Example}
\newtheorem{remark}[theorem]{Remark}
\newtheorem{proposition}[theorem]{Proposition}
\newtheorem{lemma}[theorem]{Lemma}
\newtheorem{corollary}[theorem]{Corollary}


\definecolor{darkread}{rgb}{0.7, 0, 0}

\begin{document}
\onehalfspacing
\title{\Large\textbf{{A tale of two allocations: Risk capital contributions versus risk contributions in the tail}}}
\author[]{\small Nawaf Mohammed\thanks{\url{nawaf.mohammed.ac@gmail.com}}\ }
\author[]{\small Edward Furman} 
\affil[]{\footnotesize Department of Mathematics and Statistics, York University, Toronto, ON M3J 1P3 \\ RISC Foundation, 5000 Yonge Street, Suite 1901, Toronto, ON M2N 7E9, Canada.}
\date{}
\maketitle
\vspace{-1cm}
\begin{abstract}
\onehalfspacing
We compare two natural proportional notions of a risk component’s contribution to the aggregate tail risk of a collection of risks: the fraction of aggregate tail risk capital allocated to the component under Conditional Tail Expectation (CTE), and the component's expected realized share of aggregate risk under Geometric Tail Expectation (GTE). The resulting proportional allocations generally differ. For arbitrary random vectors—allowing atoms, signed risks, and any tail domain—we establish when the allocations agree, determine their ordering when they do not, and characterize their asymptotic separation.

Both proportional allocations are weighted averages of the conditional risk share of a component given the aggregate: the proportional CTE allocation weights tail scenarios by severity, whereas the proportional GTE allocation weights them uniformly. Their difference is therefore a normalized tail covariance. The proportional CTE and GTE allocations agree throughout a tail exactly when the conditional risk share is constant there; under a mild unimodality condition, dominance is characterized by its monotonicity. Among independent exponential dispersion models with heterogeneous natural parameters, exact agreement is possible only for the scaled Poisson family; under comonotonicity, it is equivalent to proportional quantile functions. 

In the extreme tail, the limiting relationship between the proportional CTE and GTE allocations is governed by the ratio of Expected Shortfall to Value-at-Risk: boundedness ensures common limits, convergence to one forces the allocations to merge, and divergence can cause their limits to separate.
\medskip

\noindent{{\em Key words and phrases}: risk capital allocation; Conditional Tail Expectation; Geometric Tail Expectation; exponential dispersion models; comonotonicity; tail asymptotics.}

\medskip
\noindent{{\em JEL Classification}: C46; G22.}
\end{abstract}

\newpage
\section{Introduction}
\label{sec:introduction}

The allocation of aggregate risk capital among the constituents of a collection of risks is an important problem in actuarial science and quantitative risk management, underpinning performance measurement, pricing, and capital budgeting. Given risks $X_1,\ldots,X_n$, indexed by $N=\{1,\ldots,n\}$, with aggregate risk $S=\sum_{i\in N}X_i$ and an aggregate risk measure $\varrho$, one seeks contributions $\kappa_i$, $i\in N$, satisfying the full-allocation constraint $\sum_{i\in N}\kappa_i=\varrho(S)$. The Euler, or gradient, allocation principle \citep{Denault2001,Tasche2004} sets $\kappa_i=\partial\varrho\bigl(\sum_{j\in N}h_jX_j\bigr)/\partial h_i\big|_{\hh=\one}$, thereby charging each component its marginal contribution to the aggregate risk measure. When $\varrho$ is positively homogeneous of degree one, Euler's theorem on homogeneous functions guarantees full allocation, and the allocation principle admits both game-theoretic \citep{Denault2001} and axiomatic \citep{Kalkbrener2005} foundations. Beyond the Euler paradigm, prominent alternatives include the distorted allocations of \citet{TsanakasBarnett2003}, the weighted allocations of \citet{Furman2008b}, and the optimal allocations of \citet{Dhaene2012}.

Two risk measures are central to our investigation. The first is the Conditional Tail Expectation, $\CTE_u(S)=\e[S\mid S>q(u)]$, $u\in[0,1)$, where $q(u)$ is the Value-at-Risk of $S$ at level $u$. For continuous $S$ the CTE coincides with the Expected Shortfall \citep{Acerbi.2002}, on whose quantile representation our general framework is built so as to accommodate atoms. The CTE risk measure is prominent in regulatory capital measurement \citep{Embrechts2014}, and, under standard regularity conditions, its Euler contributions are $\e[X_i\mid S>q(u)]$, $i\in N$. Since these contributions sum to $\CTE_u(S)$, the canonical proportional Euler allocation of the CTE is
\begin{equation}
\label{eq:intro_Phi}
  \Phi_{i,u}
  =\frac{\e[X_i\mid S>q(u)]}
             {\e[S\mid S>q(u)]}.
\end{equation}
The second risk measure is the \emph{Geometric Tail Expectation} (GTE), defined whenever $S>0$ almost surely on the tail event $\{S>q(u)\}$ by
\[
  \mathrm{GTE}_u(S)
  =\exp\bigl(\e[\log S\mid S>q(u)]\bigr).
\]
A tail quasi-linear mean in the sense of \citet{BaurleShushi2020} with logarithmic generator, the GTE is positively homogeneous and, by Jensen's inequality, never exceeds the CTE; it gauges the typical order of magnitude of tail outcomes through their geometric, rather than arithmetic, mean. Its Euler contributions are $\mathrm{GTE}_u(S)\,\e[X_i/S\mid S>q(u)]$, $i\in N$, yielding the proportional allocation \citep{Mohammed2021}
\begin{equation}
\label{eq:intro_Psi}
  \Psi_{i,u}
  =\e\left[\frac{X_i}{S}\,\bigg|\,S>q(u)\right].
\end{equation}
Both \eqref{eq:intro_Phi} and \eqref{eq:intro_Psi} sum to one over $i\in N$, take values in $[0,1]$ for non-negative risks, and coincide trivially when $n=1$. For $n\geq 2$, however, they are generically distinct, and the comparison between them is the central subject of this paper.

Although both \eqref{eq:intro_Phi} and \eqref{eq:intro_Psi} are proportional allocations, they reflect different notions of contribution. Following \citet{Furman2020}, we distinguish between conventional and compositional notions of proportional allocation. Under the conventional approach, the risk capital allocated to a component is normalized by the total risk capital; over the tail, this produces $\Phi_{i,u}$, the $i$-th component's contribution to total risk capital. Under the compositional approach, the fundamental stochastic quantity is $R_i=X_i/S$, the $i$-th component's realized share of the aggregate risk $S$; averaging this share over the tail produces $\Psi_{i,u}$, the component's expected risk share. The latter perspective treats proportional risk contributions as compositional data on the simplex \citep{Aitchison1986}, as in \citet{Belles-Sampera2016} and, in a dynamic forecasting setting, \citet{Boonen2019}. Equivalently, $\Phi_{i,u}$ is a ratio of tail expectations, whereas $\Psi_{i,u}$ is a tail expectation of a ratio. The two generally differ because normalization and tail averaging do not commute. Writing $X_i=S\,R_i$, this non-commutativity is quantified, for continuous $S$ with $S>0$ on the tail, by
\begin{equation}
\label{eq:intro_cov}
  \Phi_{i,u}-\Psi_{i,u}
  =\frac{\Cov\bigl(S,R_i\mid S>q(u)\bigr)}{\e[S\mid S>q(u)]}.
\end{equation}
In the language of weighted allocation principles \citep{Furman2008b}, both allocations are weighted tail averages $\e[R_i\,w(S)]/\e[w(S)]$ of the realized share, with the severity weight $w(s)=s\,\mathds{1}_{\{s>q(u)\}}$ for $\Phi_{i,u}$ and the flat weight $w(s)=\mathds{1}_{\{s>q(u)\}}$ for $\Psi_{i,u}$. The CTE allocation therefore exceeds the GTE allocation exactly when the aggregate risk and the component's share are positively correlated in the tail, that is, loosely speaking, when the component drives the tail. Moreover, by the tower property, \eqref{eq:intro_cov} depends on $R_i$ only through the conditional risk share $g_i(s)=\e[R_i\mid S=s]$, which is a central tool in our analysis.

\citet{Bauer2016} interpret \eqref{eq:intro_Psi} as a profit-maximizing allocation. Building on this interpretation, \citet{Mohammed2021} establish conditions under which the CTE allocation \eqref{eq:intro_Phi} coincides with \eqref{eq:intro_Psi} for every $u\in[0,1)$, thereby clarifying the incentive structure embedded in CTE-based capital requirements. \citet{Owada2025} advocate the GTE allocation as a computationally tractable alternative and show that, for non-negative, continuously distributed risks with certain tail models, $\Phi_{i,u}$ and $\Psi_{i,u}$ are \emph{asymptotically equivalent} as $u\to 1^-$; see also \citet{Asimit2011} for the extreme-value asymptotics of CTE-based allocations.

The existing comparisons, however, operate under four principal restrictions. First, risks are assumed non-negative and continuously distributed, which excludes discrete risks as well as risks taking both signs. Second, results are established globally, that is, for all $u\in[0,1)$; this requires the integrability of $R_i$ over the entire sample space, which fails for every vector of independent risks with continuous and strictly positive densities on $\RR$ (Proposition~\ref{prop:R_not_L1}). Third, attention is restricted to the \emph{equality} of the two allocations; neither the direction of the inequality between them nor the conditions for its reversal have been studied. Fourth, asymptotic equivalence is established under specific tail models, leaving unexplored the structural mechanism that governs the joint limiting behavior of the two allocations. The present paper removes each of these restrictions through the contributions summarized below.

\subsection{Main contributions}
\label{subsec:Main contributions}

We work with arbitrary random vectors under minimal tail-integrability conditions, defining both allocations through the quantile function of $S$ so as to accommodate aggregates with atoms and of arbitrary sign (Definition~\ref{def:allocations_paradigms}). All results are established on a tail domain $[u^\ast,1)$ for an arbitrary threshold $u^\ast\in[0,1)$. This localizes the comparison to any tail level of practical interest and weakens the integrability requirements substantially: whenever $q(u^\ast)>0$, tail integrability of the components alone suffices. The global setting is recovered as the special case $u^\ast=0$. Within this framework, we show that $\Phi_{i,u}=\Psi_{i,u}$ for all $u\in[u^\ast,1)$ if and only if $g_i$ is almost surely constant on the tail (Theorem~\ref{thm:tail_allocations_coincide}), and that, whenever $g_i$ is quasi-concave (quasi-convex) on the tail, $\Phi_{i,u}\geq(\leq)\,\Psi_{i,u}$ for all $u\in[u^\ast,1)$ if and only if $g_i$ is non-decreasing (non-increasing) there (Theorem~\ref{thm:allocations_dominance_tail}). Sufficiency rests on Chebyshev's integral inequality \citep{Hardy1952}, whereas necessity requires a finer argument exploiting the unimodality of $g_i$. For instance, if $X_i\sim N(\mu_i,\sigma_i^2)$, $i\in N$, are independent, then $\e[|R_i|]=\infty$ and the global GTE allocation is not even well defined; yet on any tail with $q(u^\ast)>0$, Gaussian regression yields $g_i(s)=\beta_i+(\mu_i-\beta_i\mu_S)/s$, where $\beta_i=\sigma_i^2/\sigma_S^2$, $\mu_S=\sum_{j\in N}\mu_j$ and $\sigma_S^2=\sum_{j\in N}\sigma_j^2$. Hence $\Phi_{i,u}\geq\Psi_{i,u}$ for all $u\in[u^\ast,1)$ if and only if $\mu_i/\sigma_i^2\leq\mu_S/\sigma_S^2$, with equality throughout if and only if the two ratios coincide (Example~\ref{ex:Normal_tail}): components whose mean-to-variance ratio falls below that of the aggregate
risk drive the tail and are charged more heavily by the severity-weighted CTE
allocation.

We then specialize these results to two model classes at opposite ends of the dependence spectrum. For independent risks following additive Exponential Dispersion Models (EDMs) \citep{Jorgensen1987} with a common cumulant function, the two allocations coincide at all levels whenever the natural parameters are equal, whereas under heterogeneous natural parameters the scaled Poisson family is the unique model for which equality can hold (Theorem~\ref{thm:EDM_coincide}); dominance criteria, expressed through ratios of cumulant derivatives, follow from a saddle-point approximation of $g_i$ \citep{Daniels1954} that is exact in the Gaussian case (Proposition~\ref{prop:EDM_gi_monotonicity}). For comonotonic risks \citep{Dhaene2002}, the exact representation $g_i(s)=q_{X_i}(F_S(s))/s$, $s\neq 0$ (Proposition~\ref{prop:comonotone_gi}), reduces equality and dominance, respectively, to the constancy and the monotonicity of the quantile ratio $q_{X_i}/q_S$ on $(u^\ast,1)$.

Finally, we develop a general theory of the limiting behavior of the two allocations as $u\to 1^-$. Here $\Psi_{i,u}$ is a Ces\`{a}ro-type, uniform tail average of the conditional share along the quantile curve, whereas $\Phi_{i,u}$ is a severity-weighted average which, by an integration-by-parts identity, is itself a convex mixture of $\{\Psi_{i,v}:v\in[u,1)\}$. Consequently, convergence of $\Psi_{i,u}$ always entails convergence of
$\Phi_{i,u}$ to the same limit, an Abelian statement. The converse is
Tauberian in nature \citep{Bingham1987}: it holds whenever the ratio of
Expected Shortfall to Value-at-Risk remains bounded, a condition satisfied
under every standard tail model. It may fail otherwise but is restored under
monotonicity of $g_i$ (Theorem~\ref{thm:allocation_limits} and
Corollary~\ref{cor:tauberian}). When $g_i$ itself converges at the upper endpoint of the support of $S$, both allocations inherit its limit (Corollary~\ref{cor:limit_identification}); this identifies, for instance, the concentration of comonotonic Pareto allocations on the heaviest-tailed components (Example~\ref{ex:comonotone_Pareto}).

Section~\ref{sec:setup} introduces the general framework and the foundational equality and dominance results. Sections~\ref{sec:Independent_EDMs} and~\ref{sec:comonotone_dependence} specialize to independent additive EDMs and to comonotonic risks, respectively. Section~\ref{sec:allocation_convergence} develops the asymptotic theory, and Section~\ref{sec:conclusions} concludes. All proofs are collected in Appendix~\ref{app:proofs}.
\section{Set-up and Preliminary Results}
\label{sec:setup}

This section develops the framework on which the remainder of the paper rests. Its guiding observation is that, along the quantile curve of the aggregate risk, both allocations become functionals of a single deterministic function, the conditional risk share evaluated at the quantiles of $S$, and that their comparison is governed by the tail covariance between this function and the quantile function itself.

Fix $n\in\NN$, $n\ge 2$, and let $N=\{1,\ldots,n\}$ denote the index set. Let $\XX=(X_1,\ldots,X_n)$ be an arbitrary random vector with aggregate $S=\sum_{i\in N}X_i$ and distribution function $F_S(s)=\PP(S\le s)$. For each $i\in N$, the realized share of the $i$-th component in the aggregate risk pool is
\[
R_i=\frac{X_i}{S}\mathds{1}_{\{S\neq 0\}}+c_i\mathds{1}_{\{S=0\}},
\]
where the constants $c_i\in\RR$, subject to $\sum_{i\in N}c_i=1$, prescribe how a vanishing aggregate is apportioned; they matter only when $\PP(S=0)>0$. By construction, $\sum_{i\in N}R_i=1$ holds identically.

For $u\in(0,1)$, the quantile function, or Value-at-Risk, of the aggregate $S$ is
\[
q(u)=\VaR_u[S]
:=\inf\{s\in\RR:F_S(s)\ge u\},
\]
and at the endpoints we set $q(0)=s_{\min}:=\inf\{s\in\RR:F_S(s)>0\}$ and $q(1)=s_{\max}:=\sup\{s\in\RR:F_S(s)<1\}$, the lower and upper endpoints of the support of $S$, with values in $[-\infty,\infty]$. We repeatedly use the following standard facts \citep{Dhaene2002}: $q$ is non-decreasing and left-continuous on $(0,1)$; for all $v\in(0,1)$ and $s\in\RR$,
\begin{equation}
\label{eq:quantile_galois}
q(v)\le s\iff v\le F_S(s);
\end{equation}
and $q(V)\eqd S$ whenever $V\sim\mathcal{U}(0,1)$.

\begin{definition}
\label{def:Tail_event_SDomain}
Take any $u^\ast\in[0,1)$ and let $s^{\ast}:=q(u^\ast)$ denote its corresponding quantile. The effective tail domain $\D$ of $u^\ast$ is defined as
\begin{equation*}
\label{eq:D_effective_domain}
 \D:=
  \begin{cases}
    (s^{\ast},s_{\max}] & \text{if }F_S(s^{\ast})=u^{\ast},\\[4pt]
    [s^{\ast},s_{\max}] & \text{if }F_S(s^{\ast})>u^{\ast}.
  \end{cases}
\end{equation*}
If $s_{\min}=-\infty$ or $s_{\max}=\infty$, then the intervals are taken open at those points.
\end{definition}

The two cases distinguish whether the upper tail of probability $1-u^\ast$ is the event $\{S>s^\ast\}$ or must also absorb part of an atom of $S$ at $s^\ast$. In either case, $\PP(S\in\D)\ge1-u^\ast>0$ and, by \eqref{eq:quantile_galois}, $q(v)\in\D$ for every $v\in(u^\ast,1)$, since $q(v)=s^\ast$ forces $u^\ast<v\le F_S(s^\ast)$, which is possible only in the second case. To describe how the quantile curve over $(u^\ast,1)$ represents the law of $S$ on $\D$, we introduce, for $u\in[u^\ast,1)$, the \emph{randomized tail weight}
\begin{equation}
\label{eq:tail_weight}
\pi_u(s):=\mathds{1}_{\{s>q(u)\}}+\theta_u\,\mathds{1}_{\{s=q(u)\}},
\qquad
\theta_u:=\frac{F_S\bigl(q(u)\bigr)-u}{\PP\bigl(S=q(u)\bigr)},
\end{equation}
with the convention $\theta_u:=0$ when $\PP(S=q(u))=0$. Thus $\pi_u(s)$ is the fraction of the scenarios with aggregate outcome $s$ that belong to the upper $(1-u)$-tail; an atom at the Value-at-Risk is split exactly as in the quantile representation of the Expected Shortfall \citep{Acerbi.2002}.

\begin{lemma}
\label{lem:tail_transfer}
Fix $u^\ast\in[0,1)$ with tail domain $\D$. For every $u\in[u^\ast,1)$, one has $\theta_u\in[0,1]$, $\pi_u$ vanishes $\PP_S$-almost everywhere outside $\D$, and $\e[\pi_u(S)]=1-u$. Moreover, for every Borel function $\phi:\RR\to\RR$:
\begin{enumerate}
\item[\textup{(a)}] if $\phi\ge 0$ or $\e\bigl[|\phi(S)|\,\mathds{1}_{\{S\in\D\}}\bigr]<\infty$, then
\begin{equation}
\label{eq:tail_transfer}
\int_u^1\phi\bigl(q(v)\bigr)\,dv=\e\bigl[\phi(S)\,\pi_u(S)\bigr];
\end{equation}
\item[\textup{(b)}] $\phi\circ q\in L^1(u^\ast,1)$ if and only if $\e\bigl[|\phi(S)|\,\mathds{1}_{\{S\in\D\}}\bigr]<\infty$, and $\phi\circ q=0$ almost everywhere on $(u^\ast,1)$ if and only if $\phi(S)=0$ almost surely on $\{S\in\D\}$;
\item[\textup{(c)}] $\phi\circ q$ is non-decreasing (non-increasing) on $(u^\ast,1)\setminus L$ for some Lebesgue-null set $L$ if and only if $\phi$ is $\PP_S$-almost surely non-decreasing (non-increasing) on $\D$, that is, $\phi$ is non-decreasing (non-increasing) on $\D\setminus\mathcal{N}$ for some $\PP_S$-null set $\mathcal{N}$.
\end{enumerate}
\end{lemma}

Throughout the paper, a threshold $u^\ast\in[0,1)$ with tail domain $\D$ is fixed, and we impose, for all $j\in N$, the \emph{tail-integrability conditions}
\begin{equation}
\label{eq:standing_integrability}
\e\bigl[|X_j|\,\mathds{1}_{\{S\in \D\}}\bigr]<\infty
\qquad\text{and}\qquad
\e\bigl[|R_j|\,\mathds{1}_{\{S\in\D\}}\bigr]<\infty,
\end{equation}
together with the \emph{consistency condition}
\begin{equation}
\label{eq:standing_consistency}
\e\bigl[X_j\,\mathds{1}_{\{S=0\}}\bigr]=0\qquad\text{whenever } 0\in\D.
\end{equation}
Since $R_i$ need not be integrable over the whole sample space (see Proposition~\ref{prop:R_not_L1}), we define the conditional risk share on the tail only: $g_i:\RR\to\RR$ denotes a Borel function such that
\[
g_i(S)\,\mathds{1}_{\{S\in\D\}}=\e\bigl[R_i\,\mathds{1}_{\{S\in\D\}}\bigm| S\bigr]\qquad\text{almost surely}.
\]
equivalently, $g_i(s)=\e[R_i\mid S=s]$ for $\PP_S$-almost every $s\in\D$. Throughout, $g_i$ denotes this function and $g_i(S)$ the random variable obtained by evaluating it at the aggregate. Its values off $\D$ play no role, it is unique $\PP_S$-almost everywhere on $\D$, and $\sum_{j\in N}g_j=1$ $\PP_S$-almost everywhere on $\D$.

Every statement about the shape of $g_i$ below is a statement about this function on the tail domain, up to a $\PP_S$-null set. We say that $g_i$ is \emph{$\PP_S$-almost surely constant}, \emph{non-decreasing} or \emph{non-increasing} on $\D$ if there is a Borel set $\mathcal{N}$ with $\PP(S\in\mathcal{N})=0$ such that the property in question holds for the restriction of $g_i$ to $\D\setminus\mathcal{N}$, that is, for every pair $s\le t$ in $\D\setminus\mathcal{N}$; no assertion is made at a single prescribed point such as $s^\ast$, nor outside $\D$. By Lemma~\ref{lem:tail_transfer}(b)--(c) these are equivalent to the corresponding properties of $v\mapsto g_i(q(v))$ on $(u^\ast,1)$, up to a Lebesgue-null set.

The second condition in \eqref{eq:standing_integrability} is implied by the first whenever $s^\ast>0$, since then $S\ge s^\ast$ on $\{S\in\D\}$ and $\e[|R_i|\mathds{1}_{\{S\in \D\}}]\le\e[|X_i|\mathds{1}_{\{S\in \D\}}]/s^\ast$. The consistency condition \eqref{eq:standing_consistency} guarantees that the conditional share and the conditional expectation of the component are linked by the aggregate in the natural way. On $\{S\in\D\}$ we may pull the factor $S$ out of the conditional expectation, since $R_iS\,\mathds{1}_{\{S\in\D\}}=X_i\mathds{1}_{\{S\neq0\}}\mathds{1}_{\{S\in\D\}}$ is integrable; this gives
\[
S\,g_i(S)=\e[R_iS\mid S]=\e\bigl[X_i\mathds{1}_{\{S\neq 0\}}\bigm| S\bigr]=\e[X_i\mid S]-\mathds{1}_{\{S=0\}}\,\frac{\e\bigl[X_i\mathds{1}_{\{S=0\}}\bigr]}{\PP(S=0)},
\]
where the last term is present only if $0\in\D$ and $\PP(S=0)>0$. Consequently,
\begin{equation}
\label{eq:consistency_identity}
\e[X_i\mid S]=S\,g_i(S)\qquad\text{almost surely on }\{S\in\D\}.
\end{equation}
Condition \eqref{eq:standing_consistency} is void when $0\notin\D$, holds trivially when $\PP(S=0)=0$ or when all components are non-negative, and, since $\sum_{j\in N}\e[X_j\mathds{1}_{\{S=0\}}]=0$ always, merely rules out a zero aggregate concealing offsetting component expectations.

Set $h(v):=g_i\bigl(q(v)\bigr)$ for $v\in(u^\ast,1)$; by Lemma~\ref{lem:tail_transfer}(b), $h$ does not depend, up to a Lebesgue-null set, on the version of $g_i$. The tail-integrability conditions translate into the integrability of the three functions on which all subsequent arguments rest:
\begin{equation}
\label{eq:L1_triplet}
q,\ h,\ q\,h\ \in L^1(u^\ast,1).
\end{equation}
Indeed, by Lemma~\ref{lem:tail_transfer}(b) it suffices to observe that $\e[|S|\mathds{1}_{\{S\in\D\}}]\le\sum_{j\in N}\e[|X_j|\mathds{1}_{\{S\in\D\}}]<\infty$, that $|g_i(S)|\le\e[|R_i|\mid S]$ on $\{S\in\D\}$, and that $|S\,g_i(S)|=|\e[X_i\mathds{1}_{\{S\neq0\}}\mid S]|\le\e[|X_i|\mid S]$ on $\{S\in\D\}$.

\begin{remark}
\label{rem:regular_version_gi}
The mapping $s\mapsto g_i(s)$ is determined only up to a $\PP_S$-null set, and is thus arbitrary on gaps of the support of $S$ within $\D$. By Lemma~\ref{lem:tail_transfer}, tail constancy, tail monotonicity and tail quasi-concavity of $g_i$ are insensitive to such modifications once formulated $\PP_S$-almost surely, and they transfer to $h$. For pointwise statements, such as limits at the endpoints of $\D$ in Section~\ref{sec:allocation_convergence}, we fix a regular version: whenever $g_i$ has one of these properties $\PP_S$-almost surely on $\D$, we select a version having it for every $s\in\D$. Such a version always exists. Constancy is handled by the constant version. If $g_i$ is non-decreasing on $\D\setminus\mathcal{N}$ with $\PP_S(\mathcal{N})=0$, a non-decreasing version on $\D$ is obtained by monotone extension from $\D\setminus\mathcal{N}$, for instance $s\mapsto\sup\{g_i(t):t\in\D\setminus\mathcal{N},\,t\le s\}$, with the obvious modification at points of $\D$ lying below $\D\setminus\mathcal{N}$. Quasi-concavity is handled by the same extension on either side of a mode.
\end{remark}

We can now define the two allocation principles in full generality. Write $Q(u):=\int_u^1q(v)\,dv=(1-u)\,\mathrm{ES}_u(S)$, where $\mathrm{ES}_u(S)$ denotes the Expected Shortfall of $S$ at level $u$.

\begin{definition}
\label{def:allocations_paradigms}
For $i\in N$ and $u\in[u^\ast,1)$ such that $Q(u)\neq0$, define
\[
\Phi_{i,u}=\frac{\frac{1}{1-u}\int_{u}^1 q(v)\,g_i\bigl(q(v)\bigr)\,dv}{\frac{1}{1-u}\int_{u}^1 q(v)\,dv}\qquad \text{and}\qquad \Psi_{i,u}=\frac{1}{1-u}\int_{u}^1 g_i\bigl(q(v)\bigr)\,dv.
\]
The limit case $u\to 1^{-}$ is treated in Section~\ref{sec:allocation_convergence}.
\end{definition}

By \eqref{eq:L1_triplet}, both quantities are well defined and finite. Since $q$ is non-decreasing, $u\mapsto\mathrm{ES}_u(S)$ is non-decreasing; in particular, $\mathrm{ES}_{u^\ast}(S)>0$ ensures $Q(u)>0$ for every $u\in[u^\ast,1)$. The next proposition identifies the probabilistic content of Definition~\ref{def:allocations_paradigms}.

\begin{proposition}
\label{prop:allocation_representation}
For every $u\in[u^\ast,1)$ with $Q(u)\neq 0$,
\begin{equation}
\label{eq:allocation_representation}
\Phi_{i,u}=\frac{\e\bigl[X_i\,\pi_u(S)\bigr]}{\e\bigl[S\,\pi_u(S)\bigr]}
\qquad\text{and}\qquad
\Psi_{i,u}=\frac{\e\bigl[R_i\,\pi_u(S)\bigr]}{\e\bigl[\pi_u(S)\bigr]},
\end{equation}
where $\e[\pi_u(S)]=1-u$ and $\e[S\,\pi_u(S)]=\e[S\mathds{1}_{\{S>q(u)\}}]+\bigl(F_S(q(u))-u\bigr)q(u)=(1-u)\,\mathrm{ES}_u(S)$. In particular, $\sum_{i\in N}\Phi_{i,u}=\sum_{i\in N}\Psi_{i,u}=1$. If $F_S(q(u))=u$, which holds for all $u$ when $S$ is continuous, then $\pi_u(S)=\mathds{1}_{\{S>q(u)\}}$ almost surely, and the two allocations reduce to \eqref{eq:intro_Phi} and \eqref{eq:intro_Psi}, respectively.
\end{proposition}

Proposition~\ref{prop:allocation_representation} shows that Definition~\ref{def:allocations_paradigms} is the canonical extension of \eqref{eq:intro_Phi} and \eqref{eq:intro_Psi} to arbitrary aggregates. The numerator $\e[X_i\,\pi_u(S)]/(1-u)$ of $\Phi_{i,u}$ is the contribution of $X_i$ to $\mathrm{ES}_u(S)$ \citep{Tasche2002,Kalkbrener2005}, with the atom of $S$ at its Value-at-Risk apportioned in the proportion $\theta_u$, so that $\Phi_{i,u}$ is the proportional Euler allocation of the Expected Shortfall, whereas $\Psi_{i,u}$ is the expected realized share over the same randomized tail. Both are weighted averages $\e[R_i\,w(S)]/\e[w(S)]$ of the realized share, with the severity weight $w(s)=s\,\pi_u(s)$ and the flat weight $w(s)=\pi_u(s)$, respectively. Note that $\Phi_{i,u}$ never depends on the constants $c_j$, whereas $\Psi_{i,u}$ does whenever $\pi_u(0)\,\PP(S=0)>0$. The discrepancy between the two weighting schemes is captured by the following identity.

\begin{lemma}
\label{lem:cov_representation}
Let $V\sim\mathcal{U}(0,1)$, and for $u\in[u^\ast,1)$ write $\e_u[\cdot]:=\e[\cdot\mid V>u]$ and $\Cov_u(\cdot,\cdot):=\Cov(\cdot,\cdot\mid V>u)$. Then, for every $u\in[u^\ast,1)$ with $Q(u)\neq0$,
\begin{equation}
\label{eq:cov_representation}
\Phi_{i,u}-\Psi_{i,u}=\frac{\Cov_u\bigl(q(V),h(V)\bigr)}{\e_u\bigl[q(V)\bigr]}.
\end{equation}
\end{lemma}

When $S$ is continuous, the pair $(q(V),h(V))$ given $V>u$ has the law of $(S,g_i(S))$ given $S>q(u)$, and \eqref{eq:cov_representation} reduces, by the tower property, to identity \eqref{eq:intro_cov} of the Introduction. The comparison of $\Phi_{i,u}$ and $\Psi_{i,u}$ is therefore the study of the sign of a tail covariance between the non-decreasing function $q$ and the function $h$. We first examine the extreme case in which this covariance vanishes identically.

\begin{theorem}
\label{thm:tail_allocations_coincide}
Fix $u^\ast\in[0,1)$, let $\D$ be its effective tail domain, and assume $Q(u)\neq0$ for all $u\in[u^\ast,1)$. Then
$\Phi_{i,u}=\Psi_{i,u}$ for all $u\in[u^\ast,1)$ if and only if
there exists $r_{i}\in\RR$ such that
\begin{equation}
\label{eq:tail_constancy}
  g_{i}(S)=r_{i}\quad \textup{almost surely\ on }\D.
\end{equation}
In this case, $\Phi_{i,u}=\Psi_{i,u}=r_i$ for all $u\in[u^\ast,1)$.
\end{theorem}

The proof reveals a dichotomy: wherever the aggregate retains dispersion beyond its Value-at-Risk ($\Delta>0$), equality forces the conditional share to be flat, while on a degenerate stretch of the tail ($\Delta=0$) equality is automatic and, by continuity, inherits the preceding level. Only \eqref{eq:L1_triplet} was used, so the argument applies verbatim to any $h$ with $h,\,qh\in L^1(u^\ast,1)$ and any threshold $u'\in[u^\ast,1)$ in place of $u^\ast$. For $u^\ast=0$, Theorem~\ref{thm:tail_allocations_coincide} covers the global setting studied by \citet{Mohammed2021}; see Corollary~\ref{cor:global_version_0}. Since condition \eqref{eq:tail_constancy} concerns an unobservable conditional expectation, it is useful to recast it through observable moment identities, which we do next.

\begin{corollary}
\label{cor:tail_weighted}
Retain the assumptions of Theorem~\ref{thm:tail_allocations_coincide}. A Borel function $w:\RR\to\RR$ is called \emph{admissible} if
\begin{enumerate}
  \item[(i)] $w(s)=0$ for all $s\notin\D$,
  \item[(ii)] $\e[|w(S)|]<\infty$ and $\e[|R_{i}\,w(S)|]<\infty$, and
  \item[(iii)] $\e[w(S)]\neq 0$.
\end{enumerate}
Then $g_i(S)=r_i$ almost surely on $\D$ if and only if
\begin{equation}
\label{eq:tail_weighted}
  r_{i}=\frac{\e[R_{i}\,w(S)]}{\e[w(S)]}\qquad\text{for every admissible }w.
\end{equation}
In this case, if $0\in\D$ and $\PP(S=0)>0$, then necessarily $c_i=r_i$, and for every admissible $w$ with $\e[w(S)\mathds{1}_{\{S\neq0\}}]\neq0$,
\begin{equation}
\label{eq:ci_formula}
  r_{i}=c_{i}
  =\frac{\e\!\left[\frac{X_{i}}{S}\,w(S)\,\mathds{1}_{\{S\neq0\}}\right]}
        {\e\bigl[w(S)\,\mathds{1}_{\{S\neq0\}}\bigr]}.
\end{equation}
Otherwise, \eqref{eq:tail_weighted} does not involve $c_i$, which remains free subject to $\sum_{j\in N}c_j=1$.
\end{corollary}

Under tail constancy, a whole family of allocation principles thus collapses to the single value $r_i$. By \eqref{eq:allocation_representation}, the flat weight $w=\pi_u$ returns $\Psi_{i,u}$ and the severity weight $w(s)=s\,\pi_u(s)$ returns $\Phi_{i,u}$, for every $u\in[u^\ast,1)$. The covariance weight $w(s)=s(s-\mu)\mathds{1}_{\{s\in\D\}}$, $\mu:=\e[S\mid S\in\D]$, is admissible provided $\e[X_j^2\mathds{1}_{\{S\in\D\}}]<\infty$ for all $j\in N$ and $S$ is non-degenerate on $\{S\in\D\}$; since $R_iS=X_i\mathds{1}_{\{S\neq0\}}$, condition \eqref{eq:standing_consistency} gives $\e[R_i\,w(S)]=\e[X_i(S-\mu)\mathds{1}_{\{S\in\D\}}]$, whence
\[
  r_{i}=\frac{\Cov(X_{i},S\mid S\in\D)}
              {\Var(S\mid S\in\D)},
\]
the tail covariance (beta) allocation. Thus the expected share, the proportional Expected Shortfall allocation and the tail covariance allocation all coincide under tail constancy. The factor $s$ in the covariance weight, as opposed to $(s-\mu)^2$, is what converts the share $R_i$ back into the component $X_i$.

Tail constancy can equivalently be expressed through a transform identity, which is particularly convenient for models specified through characteristic or cumulant generating functions, as in Section~\ref{sec:Independent_EDMs}.

\begin{proposition}\label{prop:tail_allocations_coincide_phi}
Fix $u^{\ast}\in[0,1)$ and let $\D$ be its effective tail domain. Then
$g_i(S)=r_i$ almost surely on $\D\setminus\{0\}$ if and only if
\begin{equation}
\label{eq:tail_cf_condition}
\e\bigl[X_i\,e^{\i tS}\,\mid {S\in\D}\bigr]
=r_i\,\e\bigl[S\,e^{\i tS}\,\mid {S\in\D}\bigr]
\qquad\text{for all }t\in\RR.
\end{equation}
Consequently, under the assumptions of Theorem~\ref{thm:tail_allocations_coincide}, $\Phi_{i,u}=\Psi_{i,u}$ for all $u\in[u^\ast,1)$ if and only if \eqref{eq:tail_cf_condition} holds and, in addition, $c_i=r_i$ when $0\in\D$ and $\PP(S=0)>0$.
\end{proposition}

If $g_j(S)=r_j$ on $\D\setminus\{0\}$ for all $j\in N$, then $\sum_{j\in N}r_j=1$, so that the choice $c_j=r_j$ is always compatible with the normalization of the constants. Specializing to $u^\ast=0$ yields the global statements used in Section~\ref{sec:Independent_EDMs}.

\begin{corollary}
\label{cor:global_version_0}
Suppose that $\e[|X_j|]<\infty$ and $\e[|R_j|]<\infty$ for all $j\in N$, that \eqref{eq:standing_consistency} holds with $u^\ast=0$, and that $\e[S]>0$. Then the tail domain $\mathfrak{S}_{0}$ carries the full mass of $S$, all conditional expectations over $\mathfrak{S}_{0}$ reduce to their unconditional counterparts, and:
\begin{enumerate}
\item[\textup{(i)}] $\Phi_{i,u}=\Psi_{i,u}$ for all $u\in[0,1)$ if and only if $g_i(S)=r_i$ almost surely for some $r_i\in\RR$, with $c_i=r_i$ when $\PP(S=0)>0$. Apart from the latter proviso, this is equivalent to
\begin{equation}
\label{eq:global_cf}
    \e\bigl[X_i\,e^{\i t S}\bigr]
    = r_i\,\e\bigl[S\,e^{\i t S}\bigr]
    \qquad\text{for all }t\in\RR.
\end{equation}
\item[\textup{(ii)}] Suppose, in addition, that the components of $\XX$ are independent, let $\varphi$ denote characteristic functions, let $I_0$ be the largest open interval containing the origin on which $\varphi_S$ does not vanish, and let $\operatorname{Log}$ denote the distinguished (continuous) logarithm. Then \eqref{eq:global_cf} implies
\begin{equation}
\label{eq:global_log_cf}
    \operatorname{Log}\varphi_{X_i}(t)
    = r_i\operatorname{Log}\varphi_S(t)
    \qquad
    \text{for all }t\in I_0.
\end{equation}
Conversely, \eqref{eq:global_log_cf} implies \eqref{eq:global_cf} if $I_0=\RR$, as for infinitely divisible laws. It also does so if all $X_j$ possess moment generating functions that are finite on a neighborhood $(-\delta,\delta)$ of the origin; in that case it suffices that $K_{X_i}(t)=r_iK_S(t)$ for $t\in(-\delta,\delta)$, where $K$ denotes the cumulant generating function. Moreover, $r_i\in[0,1]$ whenever $S$ is non-degenerate, and $r_i=\Var(X_i)/\Var(S)$ whenever the components have finite variances.
\end{enumerate}
\end{corollary}

Tail constancy is a knife-edge property; generically, $h$ varies along the tail and the covariance in \eqref{eq:cov_representation} is non-zero. We now show that its sign is, under a mild shape restriction, completely determined by the direction in which the conditional share moves with the aggregate.

\begin{definition}
\label{def:quasi_concave}
For $u^\ast\in[0,1)$ with tail domain $\D$, we say that $g_i$ is $\PP_S$-almost surely quasi-concave (quasi-convex) on $\D$ if there exists a mode $s_m\in[s^\ast,\infty]$ such that $g_i$ is $\PP_S$-almost surely non-decreasing (non-increasing) on $\D\cap[s^\ast,s_m]$ and $\PP_S$-almost surely non-increasing (non-decreasing) on $\D\cap[s_m,\infty)$.
\end{definition}

The class of quasi-concave functions contains all non-decreasing ($s_m=\infty$) and all non-increasing ($s_m=s^\ast$) functions, as well as every unimodal function, flat stretches included.

\begin{theorem}\label{thm:allocations_dominance_tail}
Fix $u^\ast \in [0, 1)$ and let $\D$ be its tail domain. Assume $\int_{u^\ast}^1 q(v)\,dv > 0$, and suppose $g_i$ is $\PP_S$-almost surely quasi-concave (quasi-convex) on $\D$. 

Then $\Phi_{i,u} \ge (\le)\, \Psi_{i,u}$ for all $u \in [u^\ast, 1)$ if and only if $g_i$ is $\PP_S$-almost surely non-decreasing (non-increasing) on $\D$. The \emph{if} part holds without the quasi-concavity (quasi-convexity) assumption.
\end{theorem}

Theorem~\ref{thm:allocations_dominance_tail} makes precise the intuition conveyed by \eqref{eq:intro_cov}: the severity-weighted allocation dominates the flat one at every tail level exactly when the conditional share grows with the aggregate, that is, when the component drives the tail. Quasi-concavity cannot simply be dropped from the necessity part, since a conditional share that increases along the tail except for a short dip can keep the covariance in \eqref{eq:cov_representation} positive at every level. When $S$ has an atom at the origin inside $\D$, monotonicity also involves the constant $c_i$.

\begin{remark}
\label{rem:ci_choice_monotonicity}
Let $0\in\D$ and $\PP(S=0)>0$, so that $g_i(0)=c_i$. Since $\Psi_{i,u}$, unlike $\Phi_{i,u}$, depends on $c_i$, the monotonicity of $g_i$ in Theorem~\ref{thm:allocations_dominance_tail} requires a compatible choice of $c_i$. With $g_i(0^{\pm})$ denoting the one-sided limits of a monotone regular version on either side of the origin (Remark~\ref{rem:regular_version_gi}), $g_i$ is non-decreasing across the origin if and only if
\[
g_i(0^{-})\le c_i\le g_i(0^{+}),
\]
with the inequalities reversed in the non-increasing case. Whenever the two one-sided limits coincide, the unique compatible choice is
\[
c_i=\lim_{s\to 0}g_i(s),
\]
and this continuous extension automatically satisfies $\sum_{i\in N}c_i=1$, since $\sum_{i\in N}g_i(s)=1$ for $\PP_S$-almost every $s\in\D\setminus\{0\}$.
\end{remark}

Theorems~\ref{thm:tail_allocations_coincide} and~\ref{thm:allocations_dominance_tail} are qualitative: they decide whether, and in which direction, the two allocations differ. The covariance identity \eqref{eq:cov_representation} also yields a quantitative, distribution-free control of the size of the discrepancy.

\begin{remark}
\label{rem:tail_discrepancy}
Suppose, in addition, that $\e[S^2\mathds{1}_{\{S\in\D\}}]<\infty$ and $h\in L^2(u^\ast,1)$, and write $\Var_u(\cdot):=\Var(\cdot\mid V>u)$. By \eqref{eq:cov_representation}, $\Phi_{i,u}-\Psi_{i,u}$ is the product of the tail correlation between $q(V)$ and $h(V)$, the tail coefficient of variation of the aggregate, and the tail volatility of the conditional share. Writing $\mathrm{CV}_u(S) := \sqrt{\Var_u(q(V))}\,/\,|\e_u[q(V)]|$ for the tail coefficient of variation, which equals $\sqrt{\Var(S\mid S>q(u))}/|\e[S\mid S>q(u)]|$ for continuous $S$, the Cauchy--Schwarz inequality yields
\begin{equation*}
    |\Phi_{i,u} - \Psi_{i,u}| \leq \mathrm{CV}_u(S) \,\sqrt{\mathrm{Var}_u\bigl(h(V)\bigr)}.
\end{equation*} 
If the components are non-negative and $c_i\in[0,1]$, then $h$ takes values in $[0,1]$, so $h^2\le h$ gives $\Var_u(h(V))\le\Psi_{i,u}(1-\Psi_{i,u})\le\frac14$. This yields the distribution-free bound
\begin{equation*}
    |\Phi_{i,u} - \Psi_{i,u}| \leq \mathrm{CV}_u(S)\,\sqrt{\Psi_{i,u}\bigl(1-\Psi_{i,u}\bigr)}\le\frac{1}{2} \mathrm{CV}_u(S),
\end{equation*}
governed by the tail dispersion of the aggregate and sharpest for components with extreme expected shares. For an exponential aggregate with rate $\lambda$, $\mathrm{CV}_u(S)=1/(1+\lambda q(u))\to0$ as $u\to1^-$, so the two allocations merge in the extreme tail; for a Pareto (Type~I) aggregate with tail index $\alpha>2$, $\mathrm{CV}_u(S)=1/\sqrt{\alpha(\alpha-2)}$ for every $u$, and the limiting behavior is governed by the finer structure of $g_i$ studied in Section~\ref{sec:allocation_convergence}.
\end{remark}
\section{Independent EDMs}
\label{sec:Independent_EDMs}

Independent additive exponential dispersion models are the natural first testing ground for the theory of Section~\ref{sec:setup}. The conditional risk share is tractable there, since the aggregate is again a convolution within the class and $g_i$ is obtained from its density by a single differentiation with respect to a natural parameter, exactly in Proposition~\ref{prop:E[X_i|S]_formula} and in closed form through the saddle-point expansion in Corollary~\ref{cor:EDM_E[R_i|S]}. The class also spans the support structures the general framework was built to accommodate --- the whole real line for the Normal, the generalized hyperbolic secant and the tilted extreme stable laws, the positive half-line for the gamma and the inverse Gaussian, the non-negative integers for the Poisson and its compound relatives --- so atoms, signed risks and the role of a positive threshold can all be examined within one family. And the equality question of Section~\ref{sec:setup} admits a complete answer inside it, Theorem~\ref{thm:EDM_coincide} below.

Throughout this section $\XX$ has independent components and, invoking Remark~\ref{rem:regular_version_gi}, we work with a regular version of $g_i$ defined pointwise on $\D$. We begin with the integrability that the framework requires. All results of Section~\ref{sec:setup} rest on $\e[|R_i|\mathds{1}_{\{S\in\D\}}]<\infty$, and their global counterparts on the stronger $\e[|R_i|]<\infty$. The next proposition shows that this stronger condition fails, for every component at once, as soon as a single component has a density positive on the whole line: for such models the global allocation $\Psi_{i,0}$ is not merely hard to compute but undefined, and localization to a tail domain is unavoidable rather than a matter of convenience.

\begin{proposition}
\label{prop:R_not_L1}
Suppose there exists an index $i \in N$ such that $X_i$ has an absolutely continuous distribution with a density $f_{X_i}$ that is bounded away from zero on every compact subset of $\RR$, as is the case when $f_{X_i}$ is continuous and strictly positive on $\RR$. Write $Z_i = S - X_i = \sum_{j \neq i} X_j$. Then:
\begin{enumerate}
    \item[(a)]\label{item:a}
        For every $j \neq i$ with $X_j \not\equiv 0$\,: $\;\e\left[|R_j|\right]=\infty$.
    \item[(b)]\label{item:b}
        If\/ $Z_i \not\equiv 0$\,: $\;\e\left[|R_i|\right]=\infty$.
\end{enumerate}
\end{proposition}

\begin{remark}
\label{rem:sharp_conditions_R_i_not_L1}
Both conditions are sharp. If $Z_i \equiv 0$ then $S = X_i$ almost surely and $R_i = 1\in L^\infty$; by independence, $Z_i\equiv0$ forces each $X_j$, $j \neq i$, to be almost surely constant with the constants summing to zero, so part~(b) applies in every non-degenerate situation. Likewise $X_j \equiv 0$, which gives $R_j=0$, is the only way part~(a) can fail.
\end{remark}

\begin{corollary}
\label{cor:all_full_support}
If $X_1,\ldots,X_n$ are independent, each with a continuous and strictly positive density on $\RR$, then $\e\left[|R_i|\right]=\infty$ for every $i\in N$.
\end{corollary}

\begin{remark}
\label{rem:tail_integrability_scope}
Proposition~\ref{prop:R_not_L1} and Corollary~\ref{cor:all_full_support} rule out the \emph{global} allocation ($u^\ast = 0$) for components with full support. For \emph{tail} allocations with $s^{\ast} = q(u^\ast) > 0$ they are, however, immaterial: as noted below \eqref{eq:standing_integrability}, $S \ge s^\ast > 0$ on $\{S\in\D\}$, so
\[
    \e\left[|R_i| \,\mathds{1}_{\{S\in\D\}}\right]
    \le \frac{1}{s^\ast}\,\e\left[|X_i| \,\mathds{1}_{\{S\in\D\}}\right]
    <\infty,
\]
and only $\e[|X_i|\mathds{1}_{\{S\in\D\}}]<\infty$ is required. Distributions excluded from the global framework, such as the Normal and the tilted extreme stable laws, are therefore fully admissible on any tail with $s^\ast>0$.
\end{remark}

We now add distributional structure. Assume that each $X_i$ follows an additive Exponential Dispersion Model (EDM) \citep{Jorgensen1987} with unit cumulant function $A_i$, index parameter $\lambda_i>0$ and natural parameter $\theta_i$ lying in the interior of the natural parameter space $\Theta_i\subseteq\RR$; that is,
\begin{equation*}
\label{eq:EDM_MGF}
M_{X_i}(t)=\e\left[e^{tX_i}\right]=\exp\left\{\lambda_i(A_i(\theta_i+t)-A_i(\theta_i))\right\},
\end{equation*}
equivalently
\begin{equation}
\label{eq:EDM_CGF}
K_{X_i}(t)=\log M_{X_i}(t)=\lambda_i(A_i(\theta_i+t)-A_i(\theta_i)),
\end{equation}
for $t$ in a neighborhood of the origin, whence
\[
\e[X_i]=K_{X_i}^{\prime}(0)=\lambda_i A_i^{\prime}(\theta_i),
\qquad
\Var(X_i)=K_{X_i}^{\prime\prime}(0)=\lambda_i A_i^{\prime\prime}(\theta_i)>0 .
\]
Because $\theta_i$ is interior, $A_i$ is real-analytic there and $X_i$ possesses moments of all orders; in particular $\e[|X_i|]<\infty$ always holds, so the first condition in \eqref{eq:standing_integrability} is automatic. The consistency condition \eqref{eq:standing_consistency} is void unless $0\in\D$, and holds trivially when the components are non-negative, since then $\{S=0\}\subseteq\{X_i=0\}$, and when $S$ is continuous. The second condition in \eqref{eq:standing_integrability} is automatic on any tail with $s^\ast>0$ by Remark~\ref{rem:tail_integrability_scope}, whereas for global allocations it excludes the EDMs with full real support, such as the Normal and those generated by extreme stable laws \citep{Eaton1971}, by Proposition~\ref{prop:R_not_L1}.

Two results below, Theorems~\ref{thm:EDM_coincide} and~\ref{thm:EDM_monotone_IR}, use in addition that the model generated by $A_i$ is infinitely divisible, which for an additive EDM is equivalent to its index set being all of $(0,\infty)$ \citep{Jorgensen1997}. This covers every family we consider --- Normal, Poisson, gamma, inverse Gaussian and the remaining Tweedie models, as well as the generalized hyperbolic secant family --- and excludes only the models with index set $\lambda_0\NN$, such as the binomial.

\begin{theorem}
\label{thm:EDM_coincide}
Suppose that $A_i = A$ for all $i \in N$, that the family generated by $A$ is either infinitely divisible or supported on a half-line, that $\e[|R_i|] < \infty$ for all $i \in N$, that \eqref{eq:standing_consistency} holds with $u^\ast=0$, and that $\e[S]>0$.
\begin{itemize}
\item[(1)] If $\theta_i=\theta$ for all $i\in N$, then, with $c_i=r_i$, $\Phi_{i,u}=\Psi_{i,u}=r_i:=\lambda_i/\sum_{j\in N}\lambda_j$ for all $u\in[0,1)$ and all $i\in N$.
\item[(2)] If $\theta_i\neq \theta_j$ for at least one pair, then $\Phi_{i,u}=\Psi_{i,u}$ for all $u\in[0,1)$ and all $i\in N$ if and only if $A$ generates a scaled Poisson family, that is,
\begin{equation}
\label{eq:scaled_Poisson_cumulant}
A(\theta)=\frac{a}{b}\,e^{b\theta}+\textup{const},\qquad ab>0,\ b\neq 0 ,
\end{equation}
in which case $X_i$ takes values in $b\,\NN_0$ with mean $\lambda_i a e^{b\theta_i}$, and
\[
c_i=r_i=\frac{\lambda_ie^{b\theta_i}}{\sum_{j\in N}\lambda_je^{b\theta_j}}.
\] The requirement $\e[S]>0$ forces $a,b>0$, so that the components are non-negative.
\end{itemize}
\end{theorem}

Within a common additive EDM family, Theorem~\ref{thm:EDM_coincide} shows that the coincidence of the two allocations across all levels is governed by the natural parameters. Under homogeneity the allocations coincide for every EDM, whereas under heterogeneity the scaled Poisson family is the only model for which equality can hold. The mechanism behind both statements is the constancy of $g_i$ on $\mathfrak{S}_0$ guaranteed by Theorem~\ref{thm:tail_allocations_coincide}: in the homogeneous case
\[
g_i(s)=\frac{\lambda_i}{\sum_{j\in N}\lambda_j},
\]
while in the scaled Poisson case
\[
g_i(s)=\frac{\lambda_i e^{b\theta_i}}{\sum_{j \in N} \lambda_j e^{b\theta_j}},
\]
irrespective of the heterogeneity of the $\theta_j$. Two comments on the hypotheses are in order. The requirement $\e[S]>0$ is what makes $\Phi_{i,u}$ well defined at every level, since $\mathrm{ES}_u(S)\ge\e[S]>0$; it is automatic for non-negatively supported EDMs. The two structural hypotheses enter in different ways. Infinite divisibility converts the functional equation satisfied by $A^{\prime}$ into a statement about the canonical measure of the L\'{e}vy--Khintchine representation; a support bounded on one side achieves the same end by an entirely elementary route, through the behaviour of the tilted mean at the unbounded edge of the natural parameter space. Together they cover every additive EDM considered here: the non-negatively supported families fall under the second condition, and the families supported on all of $\RR$ that we use --- Normal, generalized hyperbolic secant, and the tilted extreme stable laws --- are infinitely divisible. We know of no additive EDM outside their union, and of none violating the conclusion.

Theorem~\ref{thm:EDM_coincide} is a global statement. On a tail domain, by contrast, constancy of $g_i$ on $\D$ alone suffices for coincidence, and this can occur for any EDM. If, for instance, each $X_i$ is Normal with zero mean and variance $\sigma_i^2$, corresponding to $A(\theta)=\theta^2/2$, $\lambda_i=\sigma_i^2$ and $\theta_i=0$, then
\[
g_i(s)=\frac{\lambda_i}{\sum_{j\in N}\lambda_j}=\frac{\sigma_i^2}{\sum_{j\in N}\sigma_j^2}
\]
for every $s$, so that $\Phi_{i,u}=\Psi_{i,u}$ for all $u\in[u^\ast,1)$ and any $u^\ast>1/2$, for which $s^\ast>0$; the global allocation, on the other hand, is undefined by Proposition~\ref{prop:R_not_L1}. To investigate equality, and more generally dominance, for an arbitrary collection of additive EDMs, in which $A_i$, $\theta_i$ and $\lambda_i$ may all differ, Theorems~\ref{thm:tail_allocations_coincide} and~\ref{thm:allocations_dominance_tail} instruct us to determine the constancy, respectively the monotonicity, of $g_i$. We therefore derive first an exact and then an approximate expression for it.

\begin{proposition}
\label{prop:E[X_i|S]_formula}
Let each $X_i$, $i\in N$, follow an arbitrary additive EDM and let $f_S(\cdot;\bm{\theta})$, $\bm{\theta}=(\theta_1,\ldots,\theta_n)$, denote the density or mass function of $S$ with respect to the dominating measure. Then, for $\PP_S$-almost every $s$,
\begin{equation}
\label{eq:EDM_E[X_i|S]_exact}
\e[X_i|S=s]=\lambda_iA^{\prime}_i(\theta_i)+\frac{\partial \log f_S(s;\bm{\theta})}{\partial \theta_i}
=\e[X_i]+\frac{\partial \log f_S(s;\bm{\theta})}{\partial \theta_i}.
\end{equation}
Moreover, for $s$ in the interior of the convex hull of the support of $S$, the saddle-point approximation gives
\begin{equation}
\label{eq:EDM_E[X_i|S]_approx}
\e[X_i|S=s]= K_{X_i}^{\prime}\left(\th(s)\right)\bigl(1+\mathcal{O}(\Lambda^{-1})\bigr)
=\lambda_iA^{\prime}_i\left(\theta_i+\th(s)\right)\bigl(1+\mathcal{O}(\Lambda^{-1})\bigr),
\end{equation}
where $\th(s)$ is the unique solution of $K_S^{\prime}(\th)=s$, $\Lambda:=\sum_{j\in N}\lambda_j$, and the error is uniform on compact subsets of the interior of the mean domain as $\Lambda\to\infty$ with $s/\Lambda$, $\lambda_j/\Lambda$ and $\bm{\theta}$ fixed.
\end{proposition}

\begin{corollary}
\label{cor:EDM_E[R_i|S]}
Following Proposition~\ref{prop:E[X_i|S]_formula}, the exact conditional risk share is
\begin{equation}
\label{eq:EDM_E[R_i|S]_exact}
g_i(s)=\frac{1}{s}\left(\lambda_iA^{\prime}_i(\theta_i)+\frac{\partial \log f_S(s;\bm{\theta})}{\partial \theta_i}\right)\mathds{1}_{\{s\neq 0\}}+c_i\mathds{1}_{\{s= 0\}},
\end{equation}
and, for $s \neq 0$, its saddle-point proxy is
\begin{equation}
\label{eq:EDM_E[R_i|S]_approx}
\gh_i(s)= \frac{K_{X_i}^{\prime}\left(\th(s)\right)}{s}=\frac{K_{X_i}^{\prime}\left(\th(s)\right)}{K_{S}^{\prime}\left(\th(s)\right)}=\frac{\lambda_iA^{\prime}_i\left(\theta_i+\th(s)\right)}{\sum_{j\in N}\lambda_jA^{\prime}_j\left(\theta_j+\th(s)\right)},
\end{equation}
so that $g_i(s)=\gh_i(s)+\mathcal{O}(\Lambda^{-1})$ in the regime of Proposition~\ref{prop:E[X_i|S]_formula}.
\end{corollary}

Expression \eqref{eq:EDM_E[R_i|S]_approx} has a transparent reading. Let $\PP_{\th}$ denote the measure obtained by tilting every component exponentially by $\th$, which within an additive EDM simply shifts each natural parameter from $\theta_j$ to $\theta_j+\th$; the saddle-point $\th(s)$ is precisely the tilt under which the aggregate has mean $s$. Then
\[
\gh_i(s)=\frac{\e_{\th(s)}[X_i]}{\e_{\th(s)}[S]},
\]
the share of component $i$ in the mean of the tilted collection of risks. Two structural features follow at once: $\gh_i$ inherits the full-allocation property $\sum_{i\in N}\gh_i(s)=1$ from the exact $g_i$, and $\gh_i$ depends on the model only through the ratios of the tilted mean functions. We record next the one case in which the approximation is superfluous because the exact computation is immediate.

\begin{corollary}
\label{cor:EDM_homogeneous}
If $A_i=A$ and $\theta_i=\theta$ for all $i\in N$, and if the arbitrary constant is set to $c_i=\lambda_i/\sum_{j\in N}\lambda_j$, then the exact conditional share \eqref{eq:EDM_E[R_i|S]_exact} reduces to
\[
g_i(s)=\frac{\lambda_i}{\sum_{j\in N}\lambda_j},
\]
which is the constancy condition underlying Theorem~\ref{thm:EDM_coincide}--(1).
\end{corollary}

The saddle-point approximation is particularly well suited to EDMs. The aggregate is itself a convolution of independent additive EDMs, so its cumulant generating function is explicit and smooth, and the exponential tilting on which the saddle-point is built leaves the family invariant. The approximation is known to be accurate not only in the centre but also far in the tails, typically outperforming normal or moment-based expansions \citep{BarndorffNielsen1994,Reid1988}, and by \eqref{eq:EDM_E[X_i|S]_exact} the same order of accuracy propagates to the conditional mean, which is obtained from $\log f_S$ by a single differentiation. It is convenient to give the induced allocations a name.

\begin{definition}
\label{def:saddlepoint_allocations}
For $u\in[u^\ast,1)$ with $Q(u)\neq0$, the \emph{saddle-point allocations} $\phih_{i,u}$ and $\psih_{i,u}$ are obtained from Definition~\ref{def:allocations_paradigms} upon replacing $g_i$ by its proxy $\gh_i$ of \eqref{eq:EDM_E[R_i|S]_approx}.
\end{definition}

\begin{remark}
\label{rem:saddlepoint_transfer}
Since $\phih_{i,u}$ and $\psih_{i,u}$ are weighted averages of $\gh_i\circ q$ over $(u,1)$ with weights $q\ge s^\ast>0$ and $1$ respectively, Corollary~\ref{cor:EDM_E[R_i|S]} gives
\[
\sup_{u\in[u^\ast,1)}\bigl|\Phi_{i,u}-\phih_{i,u}\bigr|\ \vee\ \sup_{u\in[u^\ast,1)}\bigl|\Psi_{i,u}-\psih_{i,u}\bigr|
\ \le\ \sup_{s\in\D}\bigl|g_i(s)-\gh_i(s)\bigr| .
\]
The ordering criteria derived below from $\gh_i$ therefore determine the ordering of the true allocations up to an error of the same order as the saddle-point error itself, uniformly in the threshold. They determine it exactly whenever $\gh_i=g_i$, which happens in the Gaussian case of Example~\ref{ex:Normal_tail} and in the scaled Poisson case of Theorem~\ref{thm:EDM_coincide}--(2), and, as Proposition~\ref{prop:gamma_exact} shows, the criteria also reproduce the exact ordering for gamma risks, for which the saddle-point density is not exact.
\end{remark}

\begin{definition}
\label{def:EDM_tail_domain}
The saddle-point $\th\equiv\th(s)$ solves $K_S^{\prime}(\th)=s$ and is unique for any $s$ in the interior of the range of $K_S^{\prime}$, which, under steepness, is the interior of the convex hull of the support of $S$. For $u^\ast\in[0,1)$ with tail domain $\D$, the \emph{cumulant tail domain} is
\begin{equation*}
\label{eq:T_effective_domain_def}
 \T:=
  \th\left(\D\setminus \{s_{\min},s_{\max}\}\right).
\end{equation*}
Since
\begin{equation*}
\label{eq:T_effective_domain_0}
 \mathfrak{T}_{0}=\dom\left(K_S^{\prime}\right)=\bigcap_{j\in N}\bigl(\inf\Theta_j-\theta_j,\ \sup\Theta_j-\theta_j\bigr)
 =\Bigl(\max_{j\in N}\bigl(\inf\Theta_j-\theta_j\bigr),\ \min_{j\in N}\bigl(\sup\Theta_j-\theta_j\bigr)\Bigr),
\end{equation*} 
it holds more generally that
\begin{equation*}
\label{eq:T_effective_domain_general}
 \T\subseteq\Bigl(\max_{j\in N}\bigl(\inf\Theta_j-\theta_j\bigr),\ \min_{j\in N}\bigl(\sup\Theta_j-\theta_j\bigr)\Bigr).
\end{equation*} 
\end{definition}

\begin{remark}
\label{rem:T_D_bijection}
Implicit differentiation of $K_S^{\prime}(\th(s))=s$ gives
\[
\th^{\prime}(s) = \frac{1}{K_S^{\prime\prime}\left(\th(s)\right)}>0 ,
\] 
so $\th$ is an increasing bijection between $\D\setminus\{s_{\min},s_{\max}\}$ and $\T$; in particular
\begin{equation*}
\label{eq:T_effective_domain}
 \T=
  \begin{cases}
    \left(\th\left(s^{\ast}\right),\th\left(s_{\max}\right)\right) & \text{if }F_S(s^{\ast})=u^{\ast},\\[4pt]
    \left[\th\left(s^{\ast}\right),\th\left(s_{\max}\right)\right) & \text{if }F_S(s^{\ast})>u^{\ast},
  \end{cases}
\end{equation*}
where
\[
\th\left(s_{\min}\right)=\max_{j\in N}\bigl(\inf\Theta_j-\theta_j\bigr)\quad\text{and}\quad\th\left(s_{\max}\right)=\min_{j\in N}\bigl(\sup\Theta_j-\theta_j\bigr).
\]
\end{remark}

\begin{proposition}
\label{prop:EDM_gi_monotonicity}
Fix $u^\ast\in[0,1)$ with $s^\ast>0$. Then $\gh_i$ is non-decreasing (non-increasing) on $\D\setminus \{s_{\min},s_{\max}\}$ if and only if 
\begin{equation}
\label{eq:EDM_gi_monotonicity_K}
\frac{K_S^{\prime}(t)}{K_S^{\prime\prime}(t)} \geq(\leq) \frac{K_{X_i}^{\prime}(t)}{K_{X_i}^{\prime\prime}(t)},
\end{equation}
or equivalently
\begin{equation}
\label{eq:EDM_gi_monotonicity_A}
\frac{\sum_{j \in N} \lambda_j A^{\prime}_j(\theta_j + t)}{\sum_{j \in N} \lambda_j A^{\prime\prime}_j(\theta_j + t)}\ge (\le)\frac{A^{\prime}_i(\theta_i + t)}{A^{\prime\prime}_i(\theta_i + t)},
\end{equation}
holds for all $t\in\T$.
\end{proposition}

\begin{remark}
\label{rm:EDM_gi_quasiconcave_convex}
If the difference of the two sides of \eqref{eq:EDM_gi_monotonicity_K}, equivalently of \eqref{eq:EDM_gi_monotonicity_A}, changes sign at most once on $\T$, from positive to negative (negative to positive), then $\gh_i$ is quasi-concave (quasi-convex) on $\D\setminus\{s_{\min},s_{\max}\}$ in the sense of Definition~\ref{def:quasi_concave}.
\end{remark}

\begin{remark}
\label{rem:ghat_boundary_treatment}
The restriction $s^\ast>0$ is not merely technical: if $0\in\D$ and $K_{X_i}^{\prime}(t_0)\neq0$ at the tilt $t_0$ solving $K_S^{\prime}(t_0)=0$, then $\gh_i$ has a pole at the origin, whereas the exact $g_i$ is bounded there. Similarly, condition \eqref{eq:EDM_gi_monotonicity_K} governs $\gh_i$ only on the open part of the tail domain, so that the endpoints $s_{\min}$ and $s_{\max}$ must be inspected separately. This matters for compactly supported models with an atom at an endpoint, where the exact value of $g_i$ has to be examined directly.
\end{remark}

\begin{corollary}
\label{cor:EDM_allocations_monotonicity_general}
Fix $u^\ast\in[0,1)$ with $s^\ast>0$ and let $\T$ be its cumulant tail domain. Suppose $\gh_i$ is quasi-concave (quasi-convex) on $\D\setminus\{s_{\min},s_{\max}\}$, as characterized in Remark~\ref{rm:EDM_gi_quasiconcave_convex}. Then $\phih_{i,u}\ge(\le)\psih_{i,u}$ for all $u\in[u^\ast,1)$ if and only if \eqref{eq:EDM_gi_monotonicity_K}, equivalently \eqref{eq:EDM_gi_monotonicity_A}, holds on $\T$. By Remark~\ref{rem:saddlepoint_transfer}, the same ordering holds for $\Phi_{i,u}$ and $\Psi_{i,u}$ up to the saddle-point error, uniformly in $u$.
\end{corollary}

Condition \eqref{eq:EDM_gi_monotonicity_A} compares the mean-to-variance ratio of a single component with that of the collection of risks, both evaluated at the common tilt $t$. When all components are drawn from one family, this comparison can be resolved by a single scalar function.

\begin{theorem}
\label{thm:EDM_allocations_monotonicity_Ai=A}
Fix $u^\ast\in[0,1)$ with $s^\ast>0$, let $\T$ be its cumulant tail domain, and suppose $A_j=A$ for all $j\in N$. Set $\theta_{\min}=\min_{j\in N}\theta_j$, $\theta_{\max}=\max_{j\in N}\theta_j$ and
\[
\iota(\theta):=\frac{A^{\prime}(\theta)}{A^{\prime\prime}(\theta)},\qquad I(t):=\iota\bigl(\theta_{\min}+t\bigr),
\]
and assume that $\iota$ is monotone on $\Theta^{\ast}:=\bigl(\theta_{\min}+\inf\T,\ \theta_{\max}+\sup\T\bigr)$, the range of natural parameters visited by the components as the tilt ranges over $\T$. Then:
\begin{enumerate}
\item[\textup{(i)}] if $\iota$ is non-decreasing on $\Theta^{\ast}$ and $\theta_i=\theta_{\min}$, then $\gh_i$ is non-decreasing and $\phih_{i,u}\ge \psih_{i,u}$ for all $u\in[u^\ast,1)$, while if $\theta_i=\theta_{\max}$ the reverse inequalities hold;
\item[\textup{(ii)}] if $\iota$ is non-increasing on $\Theta^{\ast}$, the two conclusions in \textup{(i)} are interchanged.
\end{enumerate}
\end{theorem}

Theorem~\ref{thm:EDM_allocations_monotonicity_Ai=A} is stated for the proxy $\gh_i$. In the leading non-trivial case of gamma risks, the criterion can be verified directly on the exact conditional share, and it turns out to be exactly right even though the saddle-point density itself is not.

\begin{proposition}
\label{prop:gamma_exact}
Let $X_j$, $j\in N$, be independent gamma risks with shape $\alpha_j>0$ and rate $\beta_j>0$; this is the Tweedie EDM with $p=2$, $A(\theta)=-\log(-\theta)$, $\lambda_j=\alpha_j$ and $\theta_j=-\beta_j$. Then the conditional law of $\bm{R}=(R_1,\ldots,R_n)$ given $S=s$ has density proportional to $\prod_{j\in N} r_j^{\alpha_j-1}\exp\{-s\sum_{j\in N}\beta_jr_j\}$ on the unit simplex, and, for every $s>0$,
\begin{equation}
\label{eq:gamma_gi_derivative}
g_i^{\prime}(s)=-\Cov\Bigl(R_i,\sum_{j\in N}\beta_jR_j\Bigm| S=s\Bigr)
=-\sum_{j\neq i}(\beta_j-\beta_i)\,\Cov\bigl(R_i,R_j\bigm| S=s\bigr).
\end{equation}
Consequently, if $n=2$ then $g_i^{\prime}(s)=(\beta_j-\beta_i)\Var(R_i\mid S=s)$ for $j\neq i$, so that $\Phi_{i,u}\ge\Psi_{i,u}$ for all $u\in[0,1)$ if and only if $\beta_i\le\beta_j$. If $n\ge3$ and the conditional shares are pairwise non-positively correlated, which holds when the component densities are log-concave, that is $\alpha_j\ge1$, by the negative dependence properties of multivariate reverse rule densities \citep{Karlin1980,JoagDev1983}, then $\beta_i=\min_{j\in N}\beta_j$ implies $\Phi_{i,u}\ge\Psi_{i,u}$ and $\beta_i=\max_{j\in N}\beta_j$ implies $\Phi_{i,u}\le\Psi_{i,u}$, for all $u\in[0,1)$.
\end{proposition}

Since $\theta_j=-\beta_j$, the component with the smallest rate is the one with the largest natural parameter, and Proposition~\ref{prop:gamma_exact} states that this component, which is the heaviest-tailed and therefore drives the aggregate, is charged more by the severity-weighted allocation. This is exactly the conclusion delivered by Theorem~\ref{thm:EDM_allocations_monotonicity_Ai=A} for the Tweedie family at $p=2$, as the next example records.

\begin{example}
\label{ex:Tweedie_I(t)}
For the Tweedie family,
\[
A(\theta)=\begin{cases}
\frac{\gamma-1}{\gamma}\left(\frac{\theta}{\gamma-1}\right)^{\gamma},\quad &\text{for } p\neq 1,2,
\\
-\log(-\theta),\quad &\text{for } p=2,
\\
e^{\theta},\quad&\text{for }p=1,
\end{cases}
\]
with $\gamma=(p-2)/(p-1)$, differentiation gives $A^{\prime}(\theta)=(\theta/(\gamma-1))^{\gamma-1}$ and $A^{\prime\prime}(\theta)=(\theta/(\gamma-1))^{\gamma-2}$ for $p\neq1$, so that
\[
\iota(\theta)=\begin{cases}
(1-p)\,\theta,\quad &\text{for } p\neq 1,
\\
1,\quad&\text{for }p=1,
\end{cases}
\qquad\text{and}\qquad
I(t)=\iota(\theta_{\min}+t).
\]
In every case $\iota$ is monotone on all of $\Theta$, so the hypothesis of Theorem~\ref{thm:EDM_allocations_monotonicity_Ai=A} is met for any tail domain.

For $p\le 0$, which comprises the Normal at $p=0$ and the tilted extreme stable laws for $p<0$, $\iota$ is non-decreasing; these models are excluded from the global framework by Proposition~\ref{prop:R_not_L1} but admissible on any tail with $s^\ast>0$ by Remark~\ref{rem:tail_integrability_scope}. There, $\phih_{i,u}\ge\psih_{i,u}$ for the component with $\theta_i=\theta_{\min}$ and the reverse for $\theta_i=\theta_{\max}$. The Normal case is treated exactly in Example~\ref{ex:Normal_tail}. No Tweedie models exist for $0<p<1$. For $p=1$ we recover the Poisson model, for which $\iota$ is constant and, by Theorem~\ref{thm:EDM_coincide}--(2), the two allocations coincide exactly at every level. For $p>1$, $\iota$ is non-increasing, so the component with the largest (smallest) natural parameter has $\phih_{i,u}\ge(\le)\psih_{i,u}$ for all $u\in[u^\ast,1)$; at $p=2$ this is the exact conclusion of Proposition~\ref{prop:gamma_exact}. For $1<p<2$ the aggregate has an atom at the origin, so a threshold with $s^\ast>0$ must be used, in line with Remark~\ref{rem:ghat_boundary_treatment}.
\end{example}

The Normal family illustrates every feature of the theory at once: the global allocation is undefined, the tail allocations are well defined, the saddle-point proxy is exact, and the resulting criterion is explicit.

\begin{example}
\label{ex:Normal_tail}
Let $X_i\sim N(\mu_i,\sigma_i^2)$ for all $i\in N$, which is the Tweedie model at $p=0$ with $A(\theta)=\theta^2/2$, $\lambda_i=\sigma_i^2$ and $\theta_i=\mu_i/\sigma_i^2$, so that $S\sim N(\mu_S,\sigma_S^2)$ with $\mu_S=\sum_{j\in N}\mu_j$ and $\sigma_S^2=\sum_{j\in N}\sigma_j^2$. Fix $u^\ast$ with $s^\ast=q(u^\ast)>0$, so that $\D=(s^\ast,\infty)$. Remark~\ref{rem:tail_integrability_scope} supplies $\e[|R_i|\mathds{1}_{\{S\in\D\}}]<\infty$, and the continuity of $S$ makes \eqref{eq:standing_consistency} vacuous.

\medskip\noindent\textit{Exact conditional share.}
Since $(X_i,S)$ is jointly Normal with $\Cov(X_i,S)=\sigma_i^2$, linear regression gives $\e[X_i\mid S=s]=\mu_i+\beta_i(s-\mu_S)$ with $\beta_i:=\sigma_i^2/\sigma_S^2$, so that for all $s\in\D\subset(0,\infty)$
\begin{equation}
\label{eq:Normal_gi}
    g_i(s)
    = \beta_i+\frac{\mu_i-\beta_i\mu_S}{s}.
\end{equation}
The proxy \eqref{eq:EDM_E[R_i|S]_approx} is exact here: the quadratic cumulant generating functions $K_{X_i}(t)=\mu_it+\sigma_i^2t^2/2$ and $K_S(t)=\mu_St+\sigma_S^2t^2/2$ give the explicit saddle-point $\th(s)=(s-\mu_S)/\sigma_S^2$ and
\[
    \gh_i(s)=\frac{K^\prime_{X_i}\left(\th(s)\right)}{s}
    = \frac{\mu_i+\sigma_i^2\,\th(s)}{s}
    = \frac{\mu_i+\beta_i(s-\mu_S)}{s}
    = g_i(s),
\]
because for a quadratic cumulant generating function the saddle-point density approximation is itself exact.

\medskip\noindent\textit{Tail coincidence.}
By Theorem~\ref{thm:tail_allocations_coincide}, $\Phi_{i,u}=\Psi_{i,u}$ for all $u\in[u^\ast,1)$ if and only if $g_i$ is constant on $\D$, which by \eqref{eq:Normal_gi} happens if and only if
\begin{equation}
\label{eq:Normal_coincidence_condition}
    \mu_i = \beta_i\mu_S
    = \frac{\sigma_i^2}{\sigma_S^2}\,\mu_S
    \iff \theta_i=\frac{\mu_S}{\sigma_S^2},
\end{equation}
that is, the mean and the variance of $X_i$ represent the same fraction of those of $S$; then $r_i=\beta_i$ and $\sum_{i\in N}r_i=1$ automatically. Two natural sub-cases arise: the zero-mean case $\mu_i=0$ for all $i$, in which \eqref{eq:Normal_coincidence_condition} holds for every $i$ and $r_i$ is the variance share of $X_i$; and the mean--variance proportional case $\mu_i/\mu_S=\sigma_i^2/\sigma_S^2$ with $\mu_S\neq0$.

\medskip\noindent\textit{Allocation ordering.}
When \eqref{eq:Normal_coincidence_condition} fails, $g_i$ is strictly monotone on $\D$, since
\[
    \gh_i^{\prime}(s)=g_i^{\prime}(s)
    = -\frac{\mu_i-\beta_i\mu_S}{s^2}
\]
has constant sign there, and a strictly monotone function is both quasi-concave and quasi-convex, so Theorem~\ref{thm:allocations_dominance_tail} applies in both directions: $\Phi_{i,u}\ge\Psi_{i,u}$ for all $u\in[u^\ast,1)$ if and only if $\mu_i\le\beta_i\mu_S$, and $\Phi_{i,u}\le\Psi_{i,u}$ for all such $u$ if and only if $\mu_i\ge\beta_i\mu_S$. In the first case component $i$ contributes less to the aggregate mean than its variance share $\beta_i$ would suggest; as the threshold rises, $g_i(s)$ increases towards $\beta_i$, and the severity-weighted allocation, which puts more weight on the larger tail outcomes, charges it more. Finally, since $\mu_j=\theta_j\sigma_j^2$, the conditions $\mu_i<(>)\beta_i\mu_S$ read
\[
\theta_i<(>)\sum_{j\in N}\frac{\sigma_j^2}{\sigma_S^2}\,\theta_j ,
\]
comparisons of $\theta_i$ with a variance-weighted average of the natural parameters. They are therefore weaker than, and implied by, the extremal conditions $\theta_i=\theta_{\min}$ and $\theta_i=\theta_{\max}$ of Theorem~\ref{thm:EDM_allocations_monotonicity_Ai=A}, which is consistent with $\iota$ being non-decreasing for $p=0$.
\end{example}

The Tweedie models have $\iota$ monotone on the whole of $\Theta$. The next example shows that the localization to a tail domain is not a formality: there are families supported on $\RR$ for which $\iota$ is monotone only beyond a threshold, and that threshold is itself located in the right tail.

\begin{example}
\label{ex:hyperbolic_secant}
Let each $X_j$ follow a Generalized Hyperbolic Secant distribution, supported on $\RR$, with cumulant function $A(\theta)=-\log(\cos\theta)$ and natural parameter space $\Theta=(-\pi/2,\pi/2)$. From $A^{\prime}(\theta)=\tan\theta$ and $A^{\prime\prime}(\theta)=\sec^{2}\theta$ we obtain
\[
    \iota(\theta)=\tfrac12\sin(2\theta),\qquad I(t)=\tfrac12\sin\bigl(2(\theta_{\min}+t)\bigr),\qquad I^{\prime}(t)=\cos\bigl(2(\theta_{\min}+t)\bigr),
\]
and by Definition~\ref{def:EDM_tail_domain} the cumulant generating function of $S$ is finite on $\mathfrak{T}_0=(-\pi/2-\theta_{\min},\,t_{\max})$ with $t_{\max}=\pi/2-\theta_{\max}$. Thus $\iota$ increases up to $\pi/4$ and decreases thereafter, and the switching tilt is $t^{\ast}=\pi/4-\theta_{\min}$.

Assume $\theta_{\max}-\theta_{\min}<\pi/4$, which is exactly the condition $t^{\ast}<t_{\max}$ and bounds the heterogeneity of the natural parameters. Then $t^\ast$ lies in the right tail of the aggregate. Indeed,
\[
    K_S^{\prime}(t^{\ast}) = \sum_{j\in N} \lambda_j \tan\Bigl(\tfrac{\pi}{4} + \theta_j - \theta_{\min}\Bigr)\ \ge\ \sum_{j\in N} \lambda_j\ >\ 0,
\]
because $\theta_j-\theta_{\min}\in[0,\pi/4)$ places each argument in $[\pi/4,\pi/2)$, where the tangent is at least one. Since $K_S^{\prime}$ is strictly increasing, $t^{\ast}>t_0$, the unique root of $K_S^{\prime}(t_0)=0$, so choosing $u^\ast$ with $q(u^\ast)=s^\ast:=K_S^{\prime}(t^{\ast})>0$ gives a legitimate tail threshold, with $\T=(t^\ast,t_{\max})$ and the integrability of Remark~\ref{rem:tail_integrability_scope} in force.

On this tail domain, $\Theta^{\ast}=(\pi/4,\pi/2)$: for $t\in\T$ and any $j\in N$ we have $\theta_j+t>\theta_{\min}+t^{\ast}=\pi/4$ and $\theta_j+t<\theta_{\max}+t_{\max}=\pi/2$. Since $\iota^{\prime}(\theta)=\cos(2\theta)<0$ on $(\pi/4,\pi/2)$, $\iota$ is strictly decreasing on $\Theta^{\ast}$, and Theorem~\ref{thm:EDM_allocations_monotonicity_Ai=A}--(ii) gives $\phih_{i,u}\ge\psih_{i,u}$ for all $u\in[u^\ast,1)$ when $\theta_i=\theta_{\max}$, with the reverse inequality when $\theta_i=\theta_{\min}$. Globally, by contrast, no such conclusion is available, because $\iota$ is not monotone on $\mathfrak{T}_0$.
\end{example}

Example~\ref{ex:Tweedie_I(t)} shows that $\iota$ is non-increasing for every Tweedie model with non-negative support, and Example~\ref{ex:hyperbolic_secant} that this may fail for models supported on the whole line. The next theorem explains the pattern: for non-negatively supported infinitely divisible EDMs the monotonicity is a structural consequence of the L\'{e}vy--Khintchine representation, and the scaled Poisson family is the unique boundary case.
\begin{theorem}
\label{thm:EDM_monotone_IR}
Suppose $A_j = A$ for all $j \in N$ and that the additive EDM generated by $A$ is infinitely divisible with support contained in $[0,\infty)$; then $s_{\max}=\infty$ and
\begin{equation}
\label{eq:logconvex}
  A^{\prime}(\theta)\,A^{\prime\prime\prime}(\theta) \geq A^{\prime\prime}(\theta)^2
  \qquad \textup{for all } \theta \in \operatorname{int}(\Theta),
\end{equation}
or equivalently $\iota=A^{\prime}/A^{\prime\prime}$ is non-increasing on $\operatorname{int}(\Theta)$; in particular $I$ is non-increasing on $\mathfrak{T}_0$. Equality holds throughout \eqref{eq:logconvex} if and only if
\begin{equation}
\label{eq:scaled_Poisson_unit_cumulant}
  A(\theta) = c\,\bigl(e^{\theta x_0}-1\bigr)
  \qquad\text{for some } c,\,x_0>0,
\end{equation}
that is, the EDM belongs to the positively scaled Poisson family with jump size $x_0$ i.e. Theorem~\ref{thm:EDM_coincide}--(2) with $b=x_0$ and $a=cx_0$. In this case $\iota\equiv 1/x_0$ is constant; every other member of the class satisfies \eqref{eq:logconvex} strictly, with $\iota$ strictly decreasing.
\end{theorem}
\begin{remark}
\label{rem:logconvex_equiv}
Inequality \eqref{eq:logconvex} admits two equivalent readings. First, since $(\log A^{\prime})^{\prime\prime}=\bigl(A^{\prime}A^{\prime\prime\prime}-(A^{\prime\prime})^2\bigr)/(A^{\prime})^2$, it says that the mean function $\mu(\theta)=A^{\prime}(\theta)$ is log-convex in the canonical parameter. Second, writing $\mu=A^{\prime}(\theta)$ and $A^{\prime\prime}(\theta)=V(\mu)$ for the unit variance function, the chain rule gives $A^{\prime\prime\prime}=V^{\prime}(\mu)V(\mu)$ and \eqref{eq:logconvex} becomes
\[
    \mu\,V^{\prime}(\mu)\geq V(\mu)\qquad\textup{for all admissible }\mu>0,
\]
that is, the index of dispersion $V(\mu)/\mu$ is non-decreasing. For the Tweedie family, $V(\mu)=\mu^p$, this reduces to $p\ge1$, which is precisely the non-negatively supported Tweedie subclass of Example~\ref{ex:Tweedie_I(t)}. In both readings equality characterizes the scaled Poisson: $\mu(\theta)$ is log-linear and $V(\mu)/\mu$ is constant, equal to the jump size $x_0$.
\end{remark}
\begin{corollary}
\label{cor:ordering_nonneg_ID}
Under the assumptions of Theorem~\ref{thm:EDM_monotone_IR}, fix $u^\ast\in[0,1)$ with $s^\ast>0$. Then, by Theorem~\ref{thm:EDM_allocations_monotonicity_Ai=A}--(ii):
\begin{enumerate}
  \item[(i)] if $\theta_i=\theta_{\max}$, then $\phih_{i,u}\geq\psih_{i,u}$ for all $u\in[u^\ast,1)$;
  \item[(ii)] if $\theta_i=\theta_{\min}$, then $\phih_{i,u}\leq\psih_{i,u}$ for all $u\in[u^\ast,1)$;
  \item[(iii)] if the $\theta_j$ are not all equal, then $\phih_{i,u}=\psih_{i,u}$ for all $u\in[u^\ast,1)$ and all $i\in N$ if and only if the EDM is the scaled Poisson family \eqref{eq:scaled_Poisson_unit_cumulant}, in which case the equality is exact, $\gh_i=g_i=r_i$, and Theorem~\ref{thm:EDM_coincide}--(2) applies at every level.
\end{enumerate}
\end{corollary}
\section{Comonotonicity}
\label{sec:comonotone_dependence}

We now move to the opposite end of the dependence spectrum and compare $\Phi_{i,u}$ with $\Psi_{i,u}$ when the components of $\XX$ are perfectly positively dependent. A random vector $\XX$ is \emph{comonotonic} if there exists $U \sim \mathcal{U}(0,1)$ such that $X_i = q_{X_i}(U)$ almost surely for each $i \in N$, where $q_{X_i}$ denotes the quantile function of $X_i$. The fundamental property of comonotonic sums \citep{Dhaene2002} is that the quantile function of $S$, denoted $q_S\equiv q$ in this section, satisfies
\begin{equation}
\label{eq:comonotone_quantile_sum}
  q_S(u) = \sum_{i \in N} q_{X_i}(u)
  \qquad \forall\, u \in (0,1).
\end{equation}
Throughout the section we retain the tail-integrability and consistency conditions \eqref{eq:standing_integrability} and \eqref{eq:standing_consistency} of Section~\ref{sec:setup}, and fix $u^\ast\in[0,1)$ with tail domain $\D$.

In contrast with Section~\ref{sec:Independent_EDMs}, where the conditional risk share had to be reached through a saddle-point expansion, comonotonicity delivers it in closed form. The reason is structural: perfect dependence collapses the conditional distribution of $X_i$ given $S$ to a point mass.

\begin{proposition}
\label{prop:comonotone_gi}
Suppose $\XX$ is comonotonic. Then $X_i=q_{X_i}\bigl(F_S(S)\bigr)$ almost surely, so that $X_i$ is an almost surely non-decreasing function of the aggregate and
\begin{equation}
\label{eq:comonotone_gi_exact}
  g_i(s)
  = \frac{q_{X_i}(F_S(s))}{s}\,\mathds{1}_{\{s \neq 0\}}
        + c_i\,\mathds{1}_{\{s = 0\}}
\end{equation}
for $\PP_S$-almost every $s\in\D$. Equivalently, in terms of $h:=g_i\circ q_S$ on $(u^\ast,1)$,
\begin{equation}
\label{eq:h_formula}
  h(u)
  = \frac{q_{X_i}(u)}{q_S(u)}\,
        \mathds{1}_{\{q_S(u) \neq 0\}}
        + c_i\,\mathds{1}_{\{q_S(u) = 0\}} .
\end{equation}
Moreover, if $\PP(S=0)>0$, the consistency condition \eqref{eq:standing_consistency} holds if and only if $q_{X_i}(F_S(0))=0$ for every $i\in N$.
\end{proposition}

Two comments are in order. First, \eqref{eq:comonotone_gi_exact} is exact and requires no smoothness: the conditional risk share is the ratio of the marginal to the aggregate quantile, read at the level $F_S(s)$. Second, in line with Remark~\ref{rem:regular_version_gi} we adopt \eqref{eq:comonotone_gi_exact} as the version of $g_i$ on $\D$; since $q_S$ is non-decreasing and maps $(u^\ast,1)$ into $\D$, Lemma~\ref{lem:tail_transfer}(b)--(c) transfers constancy and monotonicity between $g_i$ on $\D$ and the quantile ratio $q_{X_i}/q_S$ on $(u^\ast,1)$, both understood almost everywhere. The last statement of Proposition~\ref{prop:comonotone_gi} shows that the consistency condition is not vacuous under comonotonicity: when the aggregate has an atom at the origin, it forces every component to vanish there, which is exactly what makes the quantile ratio in \eqref{eq:h_formula} an indeterminate rather than an infinite form. The exact representation also gives both allocations in closed form. Writing $\mathrm{ES}_u(X_i):=\frac{1}{1-u}\int_u^1q_{X_i}(v)\,dv$ for the Expected Shortfall of the $i$-th component, \eqref{eq:comonotone_quantile_sum} yields at once the classical comonotonic additivity $\sum_{i\in N}\mathrm{ES}_u(X_i)=\mathrm{ES}_u(S)$.

\begin{proposition}
\label{prop:comonotone_allocations}
Suppose $\XX$ is comonotonic and let $u\in[u^\ast,1)$ with $\mathrm{ES}_u(S)\neq0$. Then
\begin{equation}
\label{eq:comonotone_allocations}
  \Phi_{i,u}=\frac{\mathrm{ES}_u(X_i)}{\mathrm{ES}_u(S)}
  =\frac{\int_u^1 q_{X_i}(v)\,dv}{\int_u^1 q_S(v)\,dv}
  \qquad\text{and}\qquad
  \Psi_{i,u}=\frac{1}{1-u}\int_u^1\frac{q_{X_i}(v)}{q_S(v)}\,dv ,
\end{equation}
and consequently, by \eqref{eq:cov_representation},
\begin{equation}
\label{eq:comonotone_cov}
  \Phi_{i,u}-\Psi_{i,u}
  =\frac{\Cov_u\bigl(q_S(V),\,q_{X_i}(V)/q_S(V)\bigr)}{\mathrm{ES}_u(S)},
  \qquad V\sim\mathcal{U}(0,1).
\end{equation}
\end{proposition}

Identity \eqref{eq:comonotone_allocations} displays the thesis of the paper in its purest form. Under comonotonicity the Expected Shortfall is additive, so $\Phi_{i,u}$ is the \emph{ratio of the tail averages} of the two quantile curves, whereas $\Psi_{i,u}$ is the \emph{tail average of their ratio}. The first identity has an interpretation of independent interest: the marginal, Euler-type allocation of the Expected Shortfall collapses to the \emph{stand-alone} proportional allocation, in which every component is charged its own Expected Shortfall normalized by the aggregate one. No such reduction is available for $\Psi_{i,u}$, which remains a genuine average of realized risk shares. The two agree only when the ratio is flat, and \eqref{eq:comonotone_cov} shows that their difference is governed by the covariance between the aggregate quantile and the quantile share along the tail. One structural consequence is immediate from \eqref{eq:comonotone_quantile_sum}: the quantile shares satisfy $\sum_{j\in N}q_{X_j}(u)/q_S(u)=1$, so they cannot all be non-decreasing on $(u^\ast,1)$ unless every one of them is constant. Dominance is therefore intrinsically relative: whatever a component gains under the severity-weighted allocation is lost by the others.

\begin{corollary}
\label{cor:comonotone_coincidence}
Suppose $\XX$ is comonotonic and $\mathrm{ES}_u(S)\neq0$ for all $u\in[u^\ast,1)$. Then $\Phi_{i,u} = \Psi_{i,u}$ for all $u \in [u^\ast, 1)$ if and only if there is a constant $r_i\in\RR$ with
\begin{equation}
\label{eq:comonotone_constancy}
  q_{X_i}(u) = r_i\,q_S(u)
  \qquad \textup{for all } u \in (u^\ast, 1),
\end{equation}
in which case $\Phi_{i,u}=\Psi_{i,u}=r_i$ and, if in addition $0\in\D$ with $\PP(S=0)>0$, necessarily $c_i=r_i$. Condition \eqref{eq:comonotone_constancy} says that on the tail the component is a fixed fraction of the aggregate; it holds in particular when $X_j \stackrel{d}{=} \lambda_j X$ for a common random variable $X$ and constants $\lambda_j > 0$, and then $r_i = \lambda_i / \sum_{j \in N} \lambda_j$.
\end{corollary}

Two consequences deserve mention. First, $r_i\ge0$ unless $q_S$ is constant on $(u^\ast,1)$: a negative $r_i$ would make $q_{X_i}=r_iq_S$ both non-decreasing and non-increasing there. Second, taking $u^\ast=0$ in \eqref{eq:comonotone_constancy} gives $q_{X_i}=r_iq_S$ on all of $(0,1)$ and hence $X_i=r_iS$ almost surely; the two allocations therefore agree globally, for every component, precisely when $\XX$ is a comonotonic scale family, the sharpest analogue of the homogeneous case of Theorem~\ref{thm:EDM_coincide}--(1). On a tail domain the requirement is strictly weaker: only the tail quantiles need be proportional, and the marginals are unrestricted below $u^\ast$.

\begin{corollary}
\label{cor:comonotone_monotonicity}
Suppose $\XX$ is comonotonic, $\int_{u^\ast}^1 q_S(v)\,dv > 0$, and $g_i$ is $\PP_S$-almost surely quasi-concave (quasi-convex) on $\D$. Then $\Phi_{i,u} \geq (\leq) \Psi_{i,u}$ for all $u \in [u^\ast, 1)$ if and only if
\begin{equation}
\label{eq:comonotone_monotonicity}
  u \mapsto \frac{q_{X_i}(u)}{q_S(u)}
  \quad \textup{is non-decreasing (non-increasing) on }
  (u^\ast, 1).
\end{equation}
At every $u$ at which $q_{X_i}$ and $q_S$ are differentiable with $q_{X_i}^{\prime}(u) > 0$ and $q_S^{\prime}(u) > 0$, the monotonicity in \eqref{eq:comonotone_monotonicity} is equivalent to
\begin{equation}
\label{eq:comonotone_logquantile}
  \frac{q_S(u)}{q_S^{\prime}(u)}
  \geq (\leq)
  \frac{q_{X_i}(u)}{q_{X_i}^{\prime}(u)} ,
\end{equation}
that is, at every such $u$ with $q_{X_i}(u)>0$, to $\bigl(\log q_{X_i}\bigr)^{\prime}(u)\geq(\leq)\bigl(\log q_S\bigr)^{\prime}(u)$: the component's quantile grows proportionally faster (slower) than the aggregate's throughout the tail. If in addition $0\in\D$ and $\PP(S=0)>0$, the unique choice of $c_i$ compatible with the monotonicity at $s=0$ is
\[
  c_i
  = \lim_{s \to 0}
        \frac{q_{X_i}(F_S(s))}{s}
  = \lim_{u \to F_S(0)}
        \frac{q_{X_i}(u)}{q_S(u)},
\]
whenever this limit exists, and $\sum_{i \in N} c_i = 1$ holds automatically by \eqref{eq:comonotone_quantile_sum}.
\end{corollary}

Condition \eqref{eq:comonotone_logquantile} is a comparison of tail heaviness at matched probability levels: a component whose quantile function accelerates faster than the aggregate's captures an increasing share of the pool as the threshold rises, and is therefore charged more by the severity-weighted allocation. The next two examples make this explicit, first for a location-scale family, where the quantile ratio is a M\"{o}bius transform and the discrepancy can be computed in closed form, then for heterogeneous Pareto tails, where the full classification requires the analysis of a log-partition function.

\begin{example}
\label{ex:comonotone_location_scale}
Let each $X_j$ belong to the same location-scale family, that is $X_j \stackrel{d}{=} \mu_j + \sigma_j X$ for constants $\mu_j \in \RR$, $\sigma_j > 0$ and a common integrable random variable $X$. Their comonotonic sum is $S = \mu_S + \sigma_S X$ with $\mu_S = \sum_{j \in N} \mu_j$ and $\sigma_S = \sum_{j \in N} \sigma_j$, the comonotonic aggregate scale, which is distinct from the independent-sum scale $(\sum_{j\in N} \sigma_j^2)^{1/2}$. Since $q_{X_j}(u) = \mu_j + \sigma_j\,q_X(u)$, \eqref{eq:comonotone_quantile_sum} gives $q_S(u) = \mu_S + \sigma_S\,q_X(u)$. Fix $u^\ast$ with $s^\ast = q_S(u^\ast) > 0$. By Proposition~\ref{prop:comonotone_gi}, for $s\in\D$,
\begin{equation*}
\label{eq:locscale_gi}
  g_i(s)
  =
  \frac{\mu_i + \sigma_i\,q_X(F_S(s))}
       {\mu_S + \sigma_S\,q_X(F_S(s))},
  \qquad\text{that is}\qquad
  h(u)=\frac{\mu_i + \sigma_i\,x}{\mu_S + \sigma_S\,x}\Big|_{x=q_X(u)},
\end{equation*}
a M\"{o}bius transformation in $x$ whose derivative,
\[
  \frac{d}{dx}\!\left(\frac{\mu_i + \sigma_i\,x}{\mu_S + \sigma_S\,x}\right)
  =
  \frac{\sigma_i\,\mu_S - \mu_i\,\sigma_S}{(\mu_S + \sigma_S\,x)^2},
\]
has the denominator $q_S(u)^2>0$ on $\D$. Hence $h$ is monotone on $(u^\ast,1)$, with the direction fixed once and for all by the sign of $\sigma_i\mu_S - \mu_i\sigma_S$, and $g_i$ is simultaneously quasi-concave and quasi-convex, so that Corollary~\ref{cor:comonotone_monotonicity} applies in both directions.

The discrepancy is in fact available in closed form. Decomposing $q_{X_i}=(\sigma_i/\sigma_S)q_S+(\mu_i-\sigma_i\mu_S/\sigma_S)$ and inserting this in \eqref{eq:comonotone_allocations},
\[
  \Phi_{i,u}=\frac{\sigma_i}{\sigma_S}+\frac{\sigma_S\mu_i-\sigma_i\mu_S}{\sigma_S}\cdot\frac{1}{\mathrm{ES}_u(S)},
  \qquad
  \Psi_{i,u}=\frac{\sigma_i}{\sigma_S}+\frac{\sigma_S\mu_i-\sigma_i\mu_S}{\sigma_S}\cdot\frac{1}{\mathrm{HS}_u(S)},
\]
where $\mathrm{HS}_u(S):=\bigl(\frac{1}{1-u}\int_u^1 q_S(v)^{-1}dv\bigr)^{-1}$ denotes the tail harmonic mean of the aggregate. Subtracting,
\begin{equation}
\label{eq:locscale_discrepancy}
  \Phi_{i,u}-\Psi_{i,u}
  =\frac{\sigma_i\mu_S-\sigma_S\mu_i}{\sigma_S}
   \left(\frac{1}{\mathrm{HS}_u(S)}-\frac{1}{\mathrm{ES}_u(S)}\right),
\end{equation}
and the bracket is non-negative by the harmonic--arithmetic mean inequality, strictly positive unless $q_S$ is almost everywhere constant on $(u,1)$. Since the harmonic, geometric and arithmetic tail means are ordered, $\mathrm{HS}_u(S)\le\mathrm{GTE}_u(S)\le\mathrm{ES}_u(S)$, the discrepancy is driven by the very spread between tail means that separates the two risk measures in the first place. Three cases follow at once, for every $u\in[u^\ast,1)$:
\begin{enumerate}
  \item[(i)] \emph{Coincidence.} $\Phi_{i,u} = \Psi_{i,u}$ if and only if $\mu_i/\sigma_i = \mu_S/\sigma_S$, in which case $g_i\equiv r_i=\sigma_i/\sigma_S$, the scale share of $X_i$.
  \item[(ii)] \emph{Dominance.} $\Phi_{i,u} \geq \Psi_{i,u}$ if and only if $\mu_i/\sigma_i \le \mu_S/\sigma_S$, strictly whenever the inequality between the ratios is strict and $q_S$ is non-degenerate on $(u,1)$.
  \item[(iii)] \emph{Reverse dominance.} $\Phi_{i,u} \leq \Psi_{i,u}$ if and only if $\mu_i/\sigma_i \ge \mu_S/\sigma_S$.
\end{enumerate}
The criterion is a comparison of location-to-scale ratios, and it is worth contrasting it with the independent Gaussian case of Example~\ref{ex:Normal_tail}, where the corresponding criterion compared $\mu_i/\sigma_i^2$ with $\mu_S/\sigma_S^2$. Both compare the component's mean with what its dispersion share alone would predict; perfect dependence aggregates dispersion linearly, $\sigma_S=\sum_{j\in N}\sigma_j$, whereas independence aggregates it in quadratic mean. The family covers, among others, the Normal, Laplace, logistic and Student-$t$ with more than one degree of freedom.
\end{example}

\begin{example}
\label{ex:comonotone_Pareto}
Let $X_j \sim \mathrm{Pareto}(\alpha_j, \sigma_j)$ with quantile function $q_{X_j}(u) = \sigma_j(1-u)^{-1/\alpha_j}$, $\sigma_j > 0$ and $\alpha_j > 1$, joined comonotonically for $j \in N$. By \eqref{eq:comonotone_quantile_sum},
\[
  q_S(u) = \sum_{j \in N} \sigma_j(1-u)^{-1/\alpha_j}.
\]
Since $q_S(u) \geq \sum_{j\in N}\sigma_j > 0$ for all $u \in (0,1)$, the requirement $s^\ast > 0$ holds for every $u^\ast\in[0,1)$, integrability follows from $\alpha_j > 1$, and $c_i$ plays no role. By Proposition~\ref{prop:comonotone_gi},
\[
  g_i(s)
  =
  \frac{\sigma_i\,(1 - F_S(s))^{-1/\alpha_i}}
       {\sum_{j \in N}
        \sigma_j\,(1 - F_S(s))^{-1/\alpha_j}},
  \qquad s \in \D.
\]

\medskip\noindent\textit{Monotonicity analysis.}
Substituting $y = -\log(1-u)>0$, an increasing reparametrization of $(0,1)$ onto $(0,\infty)$, the quantile share becomes
\[
  h = \frac{\sigma_i\,e^{y/\alpha_i}}
             {\sum_{j \in N}\sigma_j\,e^{y/\alpha_j}} ,
\]
so that $h$ is non-decreasing in $u$ if and only if $\frac{d}{dy}\log h=1/\alpha_i-w(y)\ge0$, where
\begin{equation*}
\label{eq:Pareto_weighted_mean}
  w(y)
  :=
  \frac{d}{dy}\log\Bigl(\sum_{j\in N}\sigma_{j}e^{y/\alpha_j}\Bigr)
  =\sum_{j \in N} p_j(y)\,\frac{1}{\alpha_j},
  \qquad
  p_j(y) = \frac{\sigma_j\,e^{y/\alpha_j}}
                 {\sum_{k \in N}\sigma_k\,e^{y/\alpha_k}} .
\end{equation*}
Thus $w(y)$ is the $p(y)$-weighted average of the reciprocal tail indices, and it is the derivative of the log-partition function of the exponential family $\{p(y)\}_{y>0}$; its own derivative is therefore the corresponding weighted variance,
\begin{equation}
\label{eq:w_derivative}
  w^{\prime}(y)
  = \sum_{j\in N}p_j(y)\Bigl(\frac{1}{\alpha_j}-w(y)\Bigr)^{2}
  \geq 0,
\end{equation}
with equality if and only if all $\alpha_j$ are equal. Hence $w$ is non-decreasing on $[0,\infty)$, strictly so when the $\alpha_j$ are not all equal, and
\[
  w(0)
  =
  \frac{\sum_{j \in N}\sigma_j/\alpha_j}
       {\sum_{j \in N}\sigma_j}
  \leq w(y) \leq w(\infty)
  = \max_{j \in N}\frac{1}{\alpha_j}
  = \frac{1}{\alpha_{\min}},
  \qquad \alpha_{\min}:=\min_{j\in N}\alpha_j .
\]
Writing $y^\ast = -\log(1-u^\ast)\ge0$, the sign of $\frac{d}{dy}\log h$ on $(y^\ast,\infty)$ is therefore governed by the position of $1/\alpha_i$ relative to the non-decreasing function $w$, and exactly one of the following four mutually exclusive cases occurs.

\begin{enumerate}
  \item \emph{Coincidence.} $h$ is constant on $(u^\ast,1)$, and by Corollary~\ref{cor:comonotone_coincidence} $\Phi_{i,u}=\Psi_{i,u}$ for all $u\in[u^\ast,1)$, if and only if $w\equiv1/\alpha_i$ there, which by \eqref{eq:w_derivative} happens if and only if all tail indices are equal, $\alpha_j=\alpha$ for all $j\in N$. The common value is then the scale share $r_i = \sigma_i / \sum_{j \in N}\sigma_j$.

  \item \emph{Dominance} ($\Phi_{i,u} \geq \Psi_{i,u}$ for all $u \in [u^\ast, 1)$). Since $w\le w(\infty)=1/\alpha_{\min}$ with $w(y)\to1/\alpha_{\min}$, the inequality $1/\alpha_i \geq w(y)$ holds for all $y \geq y^\ast$ if and only if $\alpha_i = \alpha_{\min}$, that is, if and only if $X_i$ is among the heaviest-tailed components. Notably, this criterion does not depend on the threshold $u^\ast$. In this case the quantile share increases to
  \[
  \lim_{u\to1^-}h(u)=\frac{\sigma_i}{\sum_{j:\,\alpha_j=\alpha_{\min}}\sigma_j},
  \]
  the limit that Section~\ref{sec:allocation_convergence} shows both allocations inherit.

  \item \emph{Reverse dominance} ($\Phi_{i,u} \leq \Psi_{i,u}$ for all $u \in [u^\ast, 1)$). Since $w(y^\ast)$ is the minimum of $w$ on $(y^\ast,\infty)$, the inequality $1/\alpha_i \leq w(y)$ holds for all $y\ge y^\ast$ if and only if
  \[
    \frac{1}{\alpha_i}
    \leq w(y^\ast)
    =
    \frac{\sum_{j \in N}
          \tfrac{\sigma_j}{\alpha_j}\,
          (1-u^\ast)^{-1/\alpha_j}}
         {\sum_{j \in N}
          \sigma_j\,(1-u^\ast)^{-1/\alpha_j}} ,
  \]
  which for $u^\ast = 0$ reduces to $1/\alpha_i \leq w(0)=\sum_{j\in N}(\sigma_j/\alpha_j)/\sum_{j\in N}\sigma_j$. Unlike case 2, this criterion tightens as the threshold rises, since $w(y^\ast)$ increases with $u^\ast$.

  \item \emph{Quasi-concavity with an interior mode.} If $w(y^\ast) < 1/\alpha_i < 1/\alpha_{\min}$, then $w$ is continuous and strictly increasing, so there is a unique $y_0 > y^\ast$ with $w(y_0) = 1/\alpha_i$, equivalently
  \[
    \sum_{j \in N}\sigma_j\!\left(
      \frac{1}{\alpha_j} - \frac{1}{\alpha_i}
    \right)\!e^{y_0/\alpha_j} = 0 .
  \]
  With $u_0 = 1 - e^{-y_0} \in (u^\ast, 1)$ and $s_0 = q_S(u_0)$, the share $h$ is strictly increasing on $(u^\ast, u_0)$ and strictly decreasing on $(u_0, 1)$, so $g_i$ is quasi-concave on $\D$ with interior mode $s_0$. Consequently $\Phi_{i,u} \leq \Psi_{i,u}$ for all $u \in [u_0, 1)$, by the sufficiency part of Theorem~\ref{thm:allocations_dominance_tail} applied at the threshold $u_0$; whereas for every threshold $u^{\prime}<u_0$ the necessity part rules out $\Phi_{i,u}\ge\Psi_{i,u}$ holding simultaneously for all $u\in[u^{\prime},1)$. Below $u_0$ the shape of $g_i$ alone no longer decides: the sign of $\Phi_{i,u}-\Psi_{i,u}$ is that of the covariance in \eqref{eq:comonotone_cov}, which weighs the increasing stretch of $h$ against the decreasing one and may perfectly well remain negative at every such level.
\end{enumerate}
In cases 3 and 4 the component is not among the heaviest-tailed ones, so $p_i(y)\to0$ and $h(u)\to0$ as $u\to1^-$: the tail allocation is asymptotically absorbed by the components with the smallest tail index, a limit that Section~\ref{sec:allocation_convergence} shows both $\Phi_{i,u}$ and $\Psi_{i,u}$ inherit. The classification is exhaustive: case 1 covers homogeneous tail indices, and when the $\alpha_j$ differ, exactly one of cases 2, 3 and 4 applies according to whether $1/\alpha_i$ lies above, below, or strictly between the two endpoints $w(y^\ast)$ and $w(\infty)=1/\alpha_{\min}$. The engine behind the trichotomy is the weighted variance formula \eqref{eq:w_derivative}: the log-partition structure makes $w$ monotone, hence it crosses any fixed level $1/\alpha_i$ at most once, which is precisely the unimodality that Theorem~\ref{thm:allocations_dominance_tail} requires.
\end{example}
\section{Limiting Behavior}\label{sec:allocation_convergence}

Sections~\ref{sec:Independent_EDMs} and~\ref{sec:comonotone_dependence} compared the two allocations at a fixed threshold. We now let the threshold recede to the right endpoint of the support and ask whether $\Phi_{i,u}$ and $\Psi_{i,u}$ ultimately agree on the component's asymptotic share of the aggregate risk. The question is not a technicality: capital rules are calibrated at levels such as $u=0.99$ or beyond, where the two paradigms are often used interchangeably.

The answer rests on a single structural asymmetry, already visible in the weights of Section~\ref{sec:setup}. Because $\Phi_{i,u}$ weights tail scenarios by severity while $\Psi_{i,u}$ treats them alike, $\Phi_{i,u}$ turns out to be an \emph{average of the family} $\{\Psi_{i,v}:v\in[u,1)\}$ against a probability measure, whereas the inverse relation expresses $\Psi_{i,u}$ through $\{\Phi_{i,v}:v\in[u,1)\}$ against a \emph{signed} measure. Averaging smooths, so the first relation transmits convergence unconditionally; the second may amplify oscillations, and convergence travels back only under an additional condition. This is the classical Abelian--Tauberian pattern \citep{Bingham1987}, and the following lemma makes both relations exact. Throughout, $\mathrm{ES}_u(S)=\frac{1}{1-u}\int_u^1q(v)\,dv$ and $\VaR_u[S]=q(u)$ are as in Section~\ref{sec:setup}, and we write
\[
  W_u:=\frac{\mathrm{ES}_u(S)}{\VaR_u[S]}
      =\frac{\int_u^1 q(v)\,dv}{(1-u)\,q(u)}\ \ge\ 1
\]
for the tail-to-quantile multiplier of the aggregate, defined whenever $q(u)>0$.

\begin{lemma}
\label{lem:two_representations}
Fix $u^\ast\in[0,1)$ with $s_{\max}>0$, let the standing conditions \eqref{eq:standing_integrability}--\eqref{eq:standing_consistency} hold on $\D$, and set $h:=g_i\circ q$. Let $u_1\in[u^\ast,1)$ be such that $q(u)>0$ for all $u\in(u_1,1)$. Then, for every $u\in(u_1,1)$:
\begin{enumerate}
\item[\textup{(a)}] \emph{(Abelian representation.)} With $H(v):=\int_v^1h(t)\,dt=(1-v)\Psi_{i,v}$,
\begin{equation}
\label{eq:IBP_identity}
  \int_u^1 q(v)h(v)\,dv
  = q(u)(1-u)\,\Psi_{i,u}
        + \int_u^1(1-v)\Psi_{i,v}\,dq(v),
\end{equation}
and consequently
\begin{equation}
\label{eq:Phi_weighted_avg}
  \Phi_{i,u}
  =
  \frac{q(u)(1-u)\,\Psi_{i,u}
        +\int_u^1(1-v)\Psi_{i,v}\,dq(v)}
       {\int_u^1 q(v)\,dv}
  =\int_{[u,1)}\Psi_{i,v}\,\mu_u(\mathrm{d}v),
\end{equation}
where $\mu_u$ is the probability measure on $[u,1)$ placing mass $q(u)(1-u)/\int_u^1q$ at the point $u$ and mass $(1-v)\,dq(v)/\int_u^1q$ on $(u,1)$.
\item[\textup{(b)}] \emph{(Tauberian representation.)} With $\nu_u(\mathrm{d}v):=\frac{1}{1-u}\bigl(\int_v^1q(t)\,dt\bigr)\,\mathrm{d}\bigl(-1/q(v)\bigr)$, a non-negative measure on $(u,1)$ of total mass $W_u-1$,
\begin{equation}
\label{eq:Psi_nonconvex_avg}
    \Psi_{i,u} = W_u\,\Phi_{i,u} - \int_u^1 \Phi_{i,v}\,\nu_u(\mathrm{d}v).
\end{equation}
\end{enumerate}
\end{lemma}

The two representations say that $\Phi_{i,u}$ is obtained from the family $\{\Psi_{i,v}\}_{v\ge u}$ by genuine averaging, whereas recovering $\Psi_{i,u}$ from $\{\Phi_{i,v}\}_{v\ge u}$ requires a coefficient $W_u\ge1$ on the current value and a compensating negative mass $W_u-1$ on the future ones. The size of that negative mass is exactly the excess of the Expected Shortfall over the Value-at-Risk,
\[
  W_u-1=\frac{\mathrm{ES}_u(S)-\VaR_u[S]}{\VaR_u[S]}=\frac{\Delta(u)}{(1-u)\,q(u)}
  =\frac{\int_{q(u)}^{\infty}\Fbar_S(t)\,dt}{(1-u)\,q(u)},
\]
where $\Fbar_S=1-F_S$, the last equality following from Lemma~\ref{lem:tail_transfer}(a) applied to $\phi\equiv1$ and $\phi(s)=s$ together with the identity $\e[S\mathds{1}_{\{S>x\}}]=x\Fbar_S(x)+\int_x^{\infty}\Fbar_S(t)\,dt$, and
with $\Delta$ the dispersion functional that already decided the equality dichotomy in the proof of Theorem~\ref{thm:tail_allocations_coincide}. Everything in this section follows from this one quantity.

\begin{theorem}\label{thm:allocation_limits}
Fix $u^\ast \in [0,1)$ with $s_{\max} > 0$, suppose the standing conditions of Section~\ref{sec:setup} hold on $\D$, and let $\ell \in \RR$. Then, as $u \to 1^-$:
\begin{enumerate}
  \item[\textup{(1)}] \emph{(Abelian.)} $\bigl|\Phi_{i,u}-\ell\bigr|\le\sup_{v\in[u,1)}\bigl|\Psi_{i,v}-\ell\bigr|$; in particular, if $\lim_{u\to1^-}\Psi_{i,u}=\ell$ then $\lim_{u\to1^-}\Phi_{i,u}=\ell$.

  \item[\textup{(2)}] \emph{(Tauberian, aggregate condition.)} For $u$ close enough to $1$,
  \begin{equation}
  \label{eq:tauberian_bounds}
    \bigl|\Psi_{i,u}-\ell\bigr|\le(2W_u-1)\sup_{v\in[u,1)}\bigl|\Phi_{i,v}-\ell\bigr|,
    \qquad
    \bigl|\Psi_{i,u}-\Phi_{i,u}\bigr|\le 2\,(W_u-1)\sup_{v\in[u,1)}\bigl|\Phi_{i,v}\bigr| .
  \end{equation}
  Consequently:
  \begin{enumerate}
  \item[\textup{(2a)}] if the multiplier stays bounded, $\limsup_{u\to1^-}W_u<\infty$, then $\lim_{u\to1^-}\Phi_{i,u}=\ell$ if and only if $\lim_{u\to1^-}\Psi_{i,u}=\ell$;
  \item[\textup{(2b)}] if the multiplier flattens, $W_u\to1$, and $\Phi_{i,\cdot}$ is bounded near $1$, as is automatic for non-negative components, then $\Phi_{i,u}-\Psi_{i,u}\to0$ whether or not either allocation converges.
  \end{enumerate}
    The multiplier is bounded whenever $s_{\max}<\infty$, whenever $S$ belongs to the maximum domain of attraction of the Gumbel law, and whenever it belongs to that of the Fr\'{e}chet law with index $\alpha>1$, where $\limsup_{u\to1^-}W_u\le\alpha/(\alpha-1)$, with convergence to $\alpha/(\alpha-1)$ when $F_S$ is continuous; it flattens to one in the first two cases but not in the third.

  \item[\textup{(3)}] Some condition of this kind is needed: Example~\ref{ex:oscillating_counterexample} exhibits a non-negative vector $\XX$, with $W_u\to\infty$, for which $\Phi_{i,u}\to\ell$ while $\Psi_{i,u}$ oscillates.
\end{enumerate}
\end{theorem}

The bounds \eqref{eq:tauberian_bounds} are distribution-free, and they separate two questions easily conflated. Transmitting a limit backwards, from $\Phi_{i,u}$ to $\Psi_{i,u}$, asks only that the amplification factor stay finite, which every standard tail model provides; making the two allocations merge before any limit is reached asks for the stronger flattening $W_u\to1$, and then happens at the explicit rate $2(W_u-1)$, the asymptotic counterpart of Remark~\ref{rem:tail_discrepancy}. What is decisive is thus the behavior of the tail-to-quantile multiplier rather than the boundedness of $S$, and both conclusions are sharp. In Example~\ref{ex:oscillating_counterexample}, deferred to Appendix~\ref{app:proofs}, the multiplier diverges and the two allocations decouple completely; there $\sup_{v\ge u}|\Phi_{i,v}-\ell|$ is of exact order $W_u^{-1}$, so \textup{(2a)} cannot be relaxed beyond $W_u\sup_{v\ge u}|\Phi_{i,v}-\ell|\to0$. Example~\ref{ex:oscillating_merge}, also deferred, shows conversely that \textup{(2b)} says nothing about convergence: for an exponential aggregate the two allocations merge while both oscillate forever.

\begin{remark}
\label{rem:Phi_Psi_structure}
Lemma~\ref{lem:two_representations} isolates the asymmetry between the two paradigms. The map $\Psi\mapsto\Phi$ is an averaging, hence non-expansive in the supremum norm by Theorem~\ref{thm:allocation_limits}--(1): severity weighting can only smooth the profile of realized shares along the tail. The inverse map is not an averaging, since it assigns weight $W_u\ge1$ to the present level and the compensating negative mass $W_u-1$ to the levels beyond; it is expansive, with modulus $2W_u-1$.

This clarifies the role of the distributional assumptions used in the recent asymptotic literature. \citet{Owada2025} obtain the asymptotic equivalence of the two allocations when the risk vector lies in the Fr\'{e}chet maximum domain of attraction under multivariate regular variation, or in the Gumbel one. Theorem~\ref{thm:allocation_limits}--(2) locates where such assumptions are needed. The link between the two allocations costs little: in either domain of attraction the multiplier is bounded, and Part~\textup{(2a)} ties the two limits together with no assumption on the dependence structure. What the tail model supplies is the existence and the value of the limit, the province of multivariate regular variation, and, in the Gumbel case, the stronger conclusion of Part~\textup{(2b)}. The amplification of Example~\ref{ex:oscillating_counterexample} is excluded as soon as the multiplier is bounded and, failing that, by the monotonicity of the conditional risk share, as we show next.
\end{remark}

\begin{corollary}
\label{cor:tauberian}
Under the assumptions of Theorem~\ref{thm:allocation_limits}, suppose in addition that $g_i$ is $\PP_S$-almost surely non-decreasing, or $\PP_S$-almost surely non-increasing, on $\D$. Then
\[
  \lim_{u \to 1^-}\Phi_{i,u} = \ell
  \iff
  \lim_{u \to 1^-}\Psi_{i,u} = \ell ,
\]
without any condition on the aggregate.
\end{corollary}

\begin{remark}
\label{rem:tauberian_connection}
Corollary~\ref{cor:tauberian} has the shape of a classical Tauberian theorem. In the language of summability, $\Psi_{i,u}$ is a Ces\`{a}ro-type average of the conditional share along the quantile curve and $\Phi_{i,u}$ a further weighted mean of $\{\Psi_{i,v}\}_{v\ge u}$; convergence of the inner average always transmits to the outer one, whereas the converse needs a Tauberian side condition, here monotonicity of $g_i$, much as slow oscillation serves in the classical theory \citep{Bingham1987}. The two side conditions available to us are of different natures: Theorem~\ref{thm:allocation_limits}--\textup{(2a)} restricts the aggregate, through the boundedness of $\mathrm{ES}_u(S)/\VaR_u[S]$, and is indifferent to the allocation, whereas Corollary~\ref{cor:tauberian} restricts the allocation and is indifferent to the aggregate. The latter is precisely the hypothesis under which Theorem~\ref{thm:allocations_dominance_tail} orders the two allocations at every finite threshold: the monotonicity that fixes the sign of $\Phi_{i,u}-\Psi_{i,u}$ also forces it to vanish in the limit.
\end{remark}

With conditions for asymptotic equivalence in place, it remains to identify the common limit. The next corollary shows that whenever the conditional risk share itself stabilizes at the top of the support, both allocations inherit its limit.

\begin{corollary}
\label{cor:limit_identification}
Under the assumptions of Theorem~\ref{thm:allocation_limits}, suppose that $h(v)=g_i(q(v))$ converges to $\ell$ as $v\to1^-$. Then $\Psi_{i,u}\to\ell$ and $\Phi_{i,u}\to\ell$ as $u\to1^-$. The hypothesis holds if $\lim_{s\to s_{\max}^-}g_i(s)=\ell$ for the regular version of Remark~\ref{rem:regular_version_gi} and $\PP(S=s_{\max})=0$, and it holds with $\ell=g_i(s_{\max})$ if instead $S$ has an atom at $s_{\max}$.
\end{corollary}

This criterion bridges back to the dependence structures of the two preceding sections: it requires only the endpoint behavior of the conditional risk share, which both were designed to deliver in closed form.

\begin{remark}
\label{rem:limit_examples}
Corollary~\ref{cor:limit_identification} identifies $\ell$ in the two model classes of this paper.
\begin{enumerate}
  \item[\textup{(i)}] \emph{Comonotonic components} (Section~\ref{sec:comonotone_dependence}). Here $h$ is the quantile ratio $q_{X_i}/q_S$ of Proposition~\ref{prop:comonotone_gi}, so that $\ell$ is the limit of that ratio at the top of the unit interval, whenever it exists. For the collection of Pareto risks of Example~\ref{ex:comonotone_Pareto} the entire tail allocation concentrates on the heaviest-tailed components,
  \[
  \ell = \frac{\sigma_i\,\mathds{1}_{\{\alpha_i=\alpha_{\min}\}}}{\sum_{j:\,\alpha_j=\alpha_{\min}}\sigma_j},
  \]
  and both $\Phi_{i,u}$ and $\Psi_{i,u}$ converge to it. Here $W_u\to\alpha_{\min}/(\alpha_{\min}-1)$ is bounded but does not flatten, so Theorem~\ref{thm:allocation_limits}--\textup{(2a)} already ties the two limits together; the endpoint behavior of $g_i$ delivers both directly.
  \item[\textup{(ii)}] \emph{Independent EDM components} (Section~\ref{sec:Independent_EDMs}). Applying the corollary to the saddle-point proxy $\gh_i$ and its allocations of Definition~\ref{def:saddlepoint_allocations},
  \[
  \ell=\lim_{t\to t_{\max}}\frac{K_{X_i}^{\prime}(t)}{K_S^{\prime}(t)},
  \qquad t_{\max}=\sup\Theta-\theta_{\max},
  \]
  by Remark~\ref{rem:T_D_bijection}. For the non-negatively supported Tweedie family with $p>1$ and a common cumulant function, $A^{\prime}(\theta_j+t)$ diverges as $t\to t_{\max}$ only for the indices attaining $\theta_{\max}$, so
  \[
  \ell=\frac{\lambda_i\,\mathds{1}_{\{\theta_i=\theta_{\max}\}}}{\sum_{j:\,\theta_j=\theta_{\max}}\lambda_j},
  \]
  the allocation again concentrating on the components with the heaviest tilted tails. For the extremal components, $\theta_i\in\{\theta_{\min},\theta_{\max}\}$, Theorems~\ref{thm:EDM_allocations_monotonicity_Ai=A} and~\ref{thm:EDM_monotone_IR} make $\gh_i$ monotone, so Corollary~\ref{cor:tauberian} applies to $\phih_{i,u}$ and $\psih_{i,u}$ as well; in all cases Remark~\ref{rem:saddlepoint_transfer} transfers the conclusion to $\Phi_{i,u}$ and $\Psi_{i,u}$ up to the saddle-point error.
\end{enumerate}
In both cases the limit is a degenerate allocation whenever the tail is dominated by a strict subset of the collection of risks: asymptotically, the components that do not drive the aggregate receive nothing, under either paradigm.
\end{remark}

\section{Conclusions}
\label{sec:conclusions}

Choosing between the two allocations looks at first like a matter of taste between two risk measures. It is not. One of them answers how much of the aggregate tail capital a risk component absorbs, the other how large a share of the risk pool the component is expected to represent, and these are different questions.

They are, however, questions about one and the same object. Both allocations are weighted averages of the conditional risk share $g_i(S)=\e[R_i\mid S]$, $\Phi_{i,u}$ weighting a tail scenario by its severity through $s\,\pi_u(s)$ and $\Psi_{i,u}$ treating all of them alike through $\pi_u(s)$. Their difference is therefore a covariance along the tail, and the whole qualitative theory follows from the shape of that one function: the allocations agree throughout a tail precisely when $g_i$ is flat on it and, once $g_i$ is unimodal, one dominates throughout precisely when $g_i$ is monotone there, the gap never exceeding half the tail coefficient of variation of the aggregate for non-negative risks. Normalizing and tail averaging do not commute, and how badly they fail to commute is a property of the aggregate, not of the functional applied to it.

Stated on $g_i$, the criterion becomes explicit wherever the conditional share can be computed, and our two model classes sit at opposite ends of the dependence spectrum. Among independent additive exponential dispersion models sharing a cumulant function, agreement at all levels is a knife edge: it holds when the natural parameters coincide, and otherwise only within the scaled Poisson family. Away from agreement, the ordering follows the mean-to-variance ratio of the tilted family, which is non-increasing for every non-negatively supported infinitely divisible model and exact for gamma risks. Under comonotonicity the conditional share is a ratio of quantile functions, making $\Phi_{i,u}$ a ratio of Expected Shortfalls and $\Psi_{i,u}$ a tail average of quantile ratios; the two then agree on a tail exactly when the tail quantiles are proportional, and globally exactly for a comonotonic scale family.

Deep in the tail the relation becomes Abelian--Tauberian, and a single number decides it. Since $\Phi_{i,u}$ averages $\{\Psi_{i,v}\}_{v\ge u}$ against a probability measure, convergence of the expected share always transmits to the capital contribution. The inverse relation carries a signed measure whose negative mass is the relative excess of the Expected Shortfall over the Value-at-Risk, so convergence travels back while that multiplier stays bounded, as under every standard tail model, and the two allocations merge outright, convergent or not, once it flattens to one. Only when it diverges can they part company altogether, and monotonicity of $g_i$ then restores the equivalence, the same monotonicity that orders them at every finite threshold.

The choice between the two paradigms therefore matters least where it is most often debated. Under a light-tailed aggregate they agree asymptotically; under a heavy-tailed one the choice is real, and its direction is settled by whether the component's conditional share grows with the aggregate. What deserves to be estimated, in either case, is not an allocation but $g_i$, which determines both and serves any other tail quasi-linear mean equally well.

\section*{Acknowledgements}
This work received financial support from the NORM program through NSERC grant ALLRP 580632 and Mitacs
grant IT33381. We are grateful to Harry Joe for valuable discussions and for detailed comments on an earlier draft of this paper.
\bibliographystyle{apalike}
\bibliography{References.bib}
\begin{appendices}
\section{Proofs}
\label{app:proofs}

This appendix collects the proofs of all results stated in Sections~\ref{sec:setup}--\ref{sec:allocation_convergence}, in order of appearance, together with the two examples that establish the sharpness of Theorem~\ref{thm:allocation_limits}.

\begin{proof}[\textup{\textbf{Proof of Lemma~\ref{lem:tail_transfer}}}]
Write $s_u:=q(u)$ and let $V\sim\mathcal{U}(0,1)$. By \eqref{eq:quantile_galois}, $F_S(s)<u$ for all $s<s_u$, whence $F_S(s_u)-u\le F_S(s_u)-F_S(s_u^-)=\PP(S=s_u)$ and $\theta_u\in[0,1]$. If $\pi_u(s)>0$, then either $s>s_u\ge s^\ast$, or $s=s_u$ with $F_S(s_u)>u\ge u^\ast$; in both cases $s\in\D$ up to the $\PP_S$-null set $(s_{\max},\infty)$. For (a), \eqref{eq:quantile_galois} gives $\{q(V)>s_u\}=\{V>F_S(s_u)\}$, while $q(v)=s_u$ for $v\in(u,F_S(s_u)]$. Hence
\begin{align*}
\int_u^1\phi\bigl(q(v)\bigr)\,dv
&=\bigl(F_S(s_u)-u\bigr)\phi(s_u)+\e\bigl[\phi(q(V))\,\mathds{1}_{\{q(V)>s_u\}}\bigr]\\
&=\theta_u\,\PP(S=s_u)\,\phi(s_u)+\e\bigl[\phi(S)\,\mathds{1}_{\{S>s_u\}}\bigr],
\end{align*}
which is \eqref{eq:tail_transfer}; the choice $\phi\equiv1$ gives $\e[\pi_u(S)]=1-u$. For (b), note that $\pi_{u^\ast}=\mathds{1}_{\D}$ $\PP_S$-almost everywhere, except at a possible atom $s^\ast\in\D$, where $\pi_{u^\ast}(s^\ast)=\theta_{u^\ast}>0$ because $F_S(s^\ast)>u^\ast$; hence (a) applied to $|\phi|$ gives $\underline{\theta}\,\e[|\phi(S)|\mathds{1}_{\{S\in\D\}}]\le\int_{u^\ast}^1|\phi(q(v))|\,dv\le\e[|\phi(S)|\mathds{1}_{\{S\in\D\}}]$ with $\underline{\theta}:=\theta_{u^\ast}$ if $s^\ast\in\D$ and $\underline{\theta}:=1$ otherwise, and both assertions follow, the second upon replacing $\phi$ by $\mathds{1}_{\{\phi\neq0\}}$. For (c), if $\phi$ is non-decreasing on $\D\setminus\mathcal{N}$ with $\PP_S(\mathcal{N})=0$, then $\phi\circ q$ is non-decreasing off $q^{-1}(\mathcal{N})\cap(u^\ast,1)$, which is Lebesgue-null by (b). Conversely, if $\phi\circ q$ is non-decreasing off a Lebesgue-null set $L$, then $\phi$ is non-decreasing on $q\bigl((u^\ast,1)\setminus L\bigr)$, because $q$ is non-decreasing; the complement $\mathcal{N}$ of this set in $\D$ is $\PP_S$-null, since $\{q(V)\in\mathcal{N},V>u^\ast\}\subset\{V\in L\}$ and a possible atom $s^\ast\in\D$ has the non-degenerate preimage $(u^\ast,F_S(s^\ast)]$, which is not contained in $L$.
\end{proof}


\begin{proof}[\textup{\textbf{Proof of Proposition~\ref{prop:allocation_representation}}}]
Apply \eqref{eq:tail_transfer} with $\phi(s)=s$, $\phi=g_i$ and $\phi(s)=s\,g_i(s)$, which is legitimate by \eqref{eq:L1_triplet}. Since $\pi_u(S)$ is $\sigma(S)$-measurable and vanishes almost surely outside $\{S\in\D\}$, the tower property and \eqref{eq:consistency_identity} give $\e[g_i(S)\pi_u(S)]=\e[R_i\pi_u(S)]$ and $\e[S\,g_i(S)\,\pi_u(S)]=\e[X_i\,\pi_u(S)]$. The summation property follows from $\sum_{i\in N}X_i=S$ and $\sum_{i\in N}R_i=1$, and the remaining statements from \eqref{eq:tail_weight}.
\end{proof}


\begin{proof}[\textup{\textbf{Proof of Lemma~\ref{lem:cov_representation}}}]
By Definition~\ref{def:allocations_paradigms}, $\Phi_{i,u}=\e_u[q(V)h(V)]/\e_u[q(V)]$ and $\Psi_{i,u}=\e_u[h(V)]$, all expectations being finite by \eqref{eq:L1_triplet}.
\end{proof}


\begin{proof}[\textup{\textbf{Proof of Theorem~\ref{thm:tail_allocations_coincide}}}]
By Lemma~\ref{lem:tail_transfer}(b), \eqref{eq:tail_constancy} is equivalent to $h=r_i$ almost everywhere on $(u^\ast,1)$.

\medskip
\noindent\textit{Sufficiency.} If $h=r_i$ almost everywhere on $(u^\ast,1)$, then $\Cov_u(q(V),h(V))=0$ and $\Psi_{i,u}=r_i$ for all $u\in[u^\ast,1)$, and the claim follows from \eqref{eq:cov_representation}.

\medskip
\noindent\textit{Necessity.} Introduce the primitives
\[
Q(u)=\int_{u}^1 q(v)\,dv, \qquad G(u):=\int_{u}^1 h(v)\,dv, \qquad H(u):=\int_{u}^1 q(v)h(v)\,dv,\qquad u\in[u^\ast,1].
\]
By \eqref{eq:L1_triplet}, they are absolutely continuous on $[u^\ast,1]$ with $Q'=-q$, $G'=-h$ and $H'=-qh$ almost everywhere. The function $\Psi_{i,u}=G(u)/(1-u)$ is absolutely continuous on every compact subinterval of $[u^\ast,1)$, with
\begin{equation}
\label{eq:Psi_derivative}
\frac{d}{du}\Psi_{i,u}=\frac{\Psi_{i,u}-h(u)}{1-u}\qquad\text{for almost every }u\in(u^\ast,1).
\end{equation}
Clearing denominators, $\Phi_{i,u}=\Psi_{i,u}$ for all $u\in[u^\ast,1)$ is equivalent to
\begin{equation}
\label{eq:integral_identity_tail}
(1-u)\int_{u}^1 q(v)h(v)\,dv = \left(\int_{u}^1 q(v)\,dv\right)\left(\int_{u}^1 h(v)\,dv\right), \qquad u \in [u^\ast, 1),
\end{equation}
that is, to $H(u)=Q(u)\,\Psi_{i,u}$ on $[u^\ast,1)$. Differentiating \eqref{eq:integral_identity_tail} and using \eqref{eq:Psi_derivative} gives $-q(u)h(u)=-q(u)\Psi_{i,u}+Q(u)\bigl(\Psi_{i,u}-h(u)\bigr)/(1-u)$ almost everywhere. With $\Delta(u):=Q(u)-(1-u)q(u)=\int_u^1\bigl(q(v)-q(u)\bigr)\,dv$, this rearranges to the factorization
\begin{equation}
\label{eq:factored_tail}
\bigl[\Psi_{i,u}-h(u)\bigr]\,\Delta(u)=0\qquad\text{for almost every }u\in(u^\ast,1).
\end{equation}
The factor $\Delta(u)=(1-u)\bigl[\mathrm{ES}_u(S)-\VaR_u[S]\bigr]$ measures the dispersion of the aggregate beyond its Value-at-Risk. Since $q$ is non-decreasing, $\Delta\ge0$, and $\Delta$ is non-increasing, because for $u_1<u_2$ in $[u^\ast,1)$ the bound $\int_{u_1}^{u_2}q\ge(u_2-u_1)q(u_1)$ gives $\Delta(u_1)-\Delta(u_2)\ge(1-u_2)\bigl(q(u_2)-q(u_1)\bigr)\ge0$. Moreover, $\Delta(u)=0$ if and only if $q$ is constant on $[u,1)$: if $q(v_0)>q(u)$ for some $v_0\in(u,1)$, then $q(v)-q(u)\ge q(v_0)-q(u)>0$ for all $v\in[v_0,1)$, so that $\Delta(u)>0$. Let
\[
u_{0} := \inf\bigl\{u \in [u^\ast, 1) : \Delta(u) = 0\bigr\} \in [u^\ast, 1],
\]
with $\inf\varnothing:=1$. Then $\Delta>0$ on $(u^\ast,u_0)$ and, if $u_0<1$, $q$ equals some constant $b$ on $(u_0,1)$.

On $(u^\ast,u_0)$, \eqref{eq:factored_tail} gives $h(u)=\Psi_{i,u}$ almost everywhere, so that \eqref{eq:Psi_derivative} yields $\frac{d}{du}\Psi_{i,u}=0$ almost everywhere. By absolute continuity on compact subintervals, $\Psi_{i,u}$ equals a constant $r_i$ on $[u^\ast,u_0)$, and hence $h=r_i$ almost everywhere on $(u^\ast,u_0)$. On $(u_0,1)$, $h=g_i(b)=:\medtilde{r}_i$ and consequently $\Psi_{i,u}=\medtilde{r}_i$ for $u\in[u_0,1)$. If $u^\ast<u_0<1$, the continuity of $u\mapsto\Psi_{i,u}$ at $u_0$ forces $r_i=\medtilde{r}_i$. In all cases, therefore, $h$ is almost everywhere constant on $(u^\ast,1)$, which proves the necessity. The final assertion follows from the sufficiency part.
\end{proof}


\begin{proof}[\textup{\textbf{Proof of Corollary~\ref{cor:tail_weighted}}}]
If $g_{i}(S)=r_{i}$ almost surely on $\D$, the tower property and condition (i) give $\e[R_{i}\,w(S)]=\e[g_{i}(S)\,w(S)]=r_{i}\,\e[w(S)]$ for every admissible $w$. Conversely, every indicator $w=\mathds{1}_{B}$ with $B\subset\D$ Borel and $\PP(S\in B)>0$ is admissible by \eqref{eq:standing_integrability}, and \eqref{eq:tail_weighted} yields $\e[(g_{i}(S)-r_{i})\,\mathds{1}_{B}(S)]=0$. Choosing $B=\D\cap\{g_i>r_i\}$ and $B=\D\cap\{g_i<r_i\}$ whenever these sets have positive $\PP_S$-probability shows that $g_{i}(S)=r_{i}$ almost surely on $\D$. If $0\in\D$ and $\PP(S=0)>0$, then $g_i(0)=c_i$, because $R_i=c_i$ on $\{S=0\}$, so constancy forces $c_i=r_i$. Substituting the definition of $R_i$ into \eqref{eq:tail_weighted} gives
\begin{equation}
\label{eq:tail_weighted_expanded}
  r_{i}\,\e[w(S)]
  =\e\!\left[\frac{X_{i}}{S}\,w(S)\,\mathds{1}_{\{S\neq0\}}\right]
         +c_{i}\,\e\bigl[w(S)\,\mathds{1}_{\{S=0\}}\bigr],
\end{equation}
and $c_i=r_i$ yields \eqref{eq:ci_formula}. If $0\notin\D$ or $\PP(S=0)=0$, the last term in \eqref{eq:tail_weighted_expanded} vanishes for every admissible $w$.
\end{proof}


\begin{proof}[\textup{\textbf{Proof of Proposition~\ref{prop:tail_allocations_coincide_phi}}}]
By \eqref{eq:consistency_identity} and the tower property,
\[
\e\bigl[X_ie^{\i tS}\mathds{1}_{\{S\in\D\}}\bigr]=\e\bigl[S\,g_i(S)\,e^{\i tS}\mathds{1}_{\{S\in\D\}}\bigr],
\]
so \eqref{eq:tail_cf_condition} states that the Fourier--Stieltjes transform of the signed measure
\[
\nu(B):=\e\bigl[(g_i(S)-r_i)\,S\,\mathds{1}_{\{S\in B\cap\D\}}\bigr],\qquad B\subset\RR\text{ Borel},
\]
vanishes identically. The measure $\nu$ is finite, since $\e[|S\,g_i(S)|\mathds{1}_{\{S\in\D\}}]$ and $\e[|S|\mathds{1}_{\{S\in\D\}}]$ are finite by \eqref{eq:L1_triplet}. By the uniqueness theorem for Fourier--Stieltjes transforms, $\nu=0$, that is, $(g_i(S)-r_i)\,S=0$ almost surely on $\{S\in\D\}$, which is equivalent to $g_i(S)=r_i$ almost surely on $\{S\in\D\setminus\{0\}\}$. The converse is immediate. The final statement follows from Theorem~\ref{thm:tail_allocations_coincide}, because $g_i(0)=c_i$ when $\PP(S=0)>0$.
\end{proof}


\begin{proof}[\textup{\textbf{Proof of Corollary~\ref{cor:global_version_0}}}]
(i) For $u^\ast=0$, we have $s^\ast=s_{\min}$ and $\PP(S\in\mathfrak{S}_0)=1$ in both cases of Definition~\ref{def:Tail_event_SDomain}. Since $\e[S]>0$, $Q(u)>0$ for all $u\in[0,1)$, and (i) follows from Theorem~\ref{thm:tail_allocations_coincide} and Proposition~\ref{prop:tail_allocations_coincide_phi}.

(ii) Write $Z_i=S-X_i$. By independence, $\e[X_ie^{\i tS}]=-\i\,\varphi_{X_i}'(t)\,\varphi_{Z_i}(t)$ and $\e[Se^{\i tS}]=-\i\,\varphi_S'(t)$, so \eqref{eq:global_cf} reads $\varphi_{X_i}'\varphi_{Z_i}=r_i\,\varphi_S'$. Dividing by $\varphi_S=\varphi_{X_i}\varphi_{Z_i}$ on $I_0$ and integrating from the origin yields \eqref{eq:global_log_cf}; reversing these steps gives \eqref{eq:global_cf} on $I_0$, hence on $\RR$ when $I_0=\RR$. Under the moment generating function condition, the same computation with real arguments shows that $K_{X_i}=r_iK_S$ on $(-\delta,\delta)$ implies $\e[X_ie^{tS}]=r_i\,\e[Se^{tS}]$ there; both sides are analytic in the strip $\{|\mathrm{Re}\,z|<\delta\}$ and thus agree on the imaginary axis, which is \eqref{eq:global_cf}. If $r_i<0$, then \eqref{eq:global_log_cf} gives $|\varphi_{X_i}|=|\varphi_S|^{r_i}\ge1$, hence $|\varphi_{X_i}|=1$ on $I_0$, so that $X_i$ is degenerate; then $|\varphi_S|=1$ on $I_0$ and $S$ is degenerate. Hence $r_i\ge0$ for non-degenerate $S$, and $\sum_{j\in N}r_j=1$ gives $r_i\in[0,1]$. Finally, under finite variances, the covariance weight of Corollary~\ref{cor:tail_weighted} with $u^\ast=0$ gives $r_i=\Cov(X_i,S)/\Var(S)=\Var(X_i)/\Var(S)$ \citep[see also][]{Mohammed2021}.
\end{proof}


\begin{proof}[\textup{\textbf{Proof of Theorem~\ref{thm:allocations_dominance_tail}}}]
The map $u\mapsto\mathrm{ES}_u(S)$ is non-decreasing, so that $\e_u[q(V)]=\mathrm{ES}_u(S)$ is at least $\mathrm{ES}_{u^\ast}(S)>0$ for every $u\in[u^\ast,1)$. By \eqref{eq:cov_representation}, $\Phi_{i,u}\ge\Psi_{i,u}$ is thus equivalent to
\begin{equation}\label{eq:chebyshev_ineq_tail}
\frac{1}{1-u}\int_{u}^1 q(v)h(v)\,dv \ge \left(\frac{1}{1-u}\int_{u}^1 q(v)\,dv\right)\left(\frac{1}{1-u}\int_{u}^1 h(v)\,dv\right),
\end{equation}
that is, to $\Cov_u\bigl(q(V),h(V)\bigr)\ge0$.
We shall use Chebyshev's integral inequality \citep{Hardy1952} in the form
\begin{equation}
\label{eq:chebyshev_double}
\Cov_u\bigl(q(V),h(V)\bigr)=\frac{1}{2(1-u)^2}\int_u^1\!\!\int_u^1\bigl(q(v)-q(w)\bigr)\bigl(h(v)-h(w)\bigr)\,dv\,dw,
\end{equation}
where the double integral converges absolutely by \eqref{eq:L1_triplet}. Its integrand is almost everywhere non-negative (non-positive) whenever $h$ is almost everywhere non-decreasing (non-increasing) on $(u,1)$.

\medskip
\noindent\textit{Sufficiency.} If $g_i$ is $\PP_S$-almost surely non-decreasing on $\D$, then $h$ is almost everywhere non-decreasing on $(u^\ast,1)$ by Lemma~\ref{lem:tail_transfer}(c), and \eqref{eq:chebyshev_double} gives \eqref{eq:chebyshev_ineq_tail} for every $u\in[u^\ast,1)$. Quasi-concavity was not used.

\medskip
\noindent\textit{Necessity.} Assume \eqref{eq:chebyshev_ineq_tail} for all $u\in[u^\ast,1)$, and let $s_m$ be a mode as in Definition~\ref{def:quasi_concave}. Put $u_m:=F_S(s_m)$, with $u_m:=1$ if $s_m=\infty$; then $u_m\ge F_S(s^\ast)\ge u^\ast$. By \eqref{eq:quantile_galois}, $q(v)\in\D\cap[s^\ast,s_m]$ for $v\in(u^\ast,u_m]$ and $q(v)\in\D\cap(s_m,\infty)$ for $v\in(u_m,1)$. Arguing as in Lemma~\ref{lem:tail_transfer}(c), $h$ is almost everywhere non-decreasing on $(u^\ast,u_m)$ and almost everywhere non-increasing on $(u_m,1)$.

If $u_m=1$, then $h$ is almost everywhere non-decreasing on $(u^\ast,1)$, and Lemma~\ref{lem:tail_transfer}(c) concludes. Let $u_m<1$. For $u\in[u_m,1)$, the functions $q$ and $h$ are oppositely ordered on $(u,1)$, so \eqref{eq:chebyshev_double} gives $\Cov_u(q(V),h(V))\le0$. Together with \eqref{eq:chebyshev_ineq_tail}, this yields $\Phi_{i,u}=\Psi_{i,u}$ for all $u\in[u_m,1)$. The argument of Theorem~\ref{thm:tail_allocations_coincide}, applied with threshold $u_m$, provides a constant $r$ with $h=r$ almost everywhere on $(u_m,1)$. If $u_m=u^\ast$, the proof is complete.

Let $u_m>u^\ast$ and $M_h:=\operatorname*{ess\,sup}_{v\in(u^\ast,u_m)}h(v)$; we claim that $M_h\le r$. Suppose instead that $M_h>r$. Since $h$ is almost everywhere non-decreasing on $(u^\ast,u_m)$, there is $\epsilon\in(0,u_m-u^\ast)$ with $h>r$ almost everywhere on $(u_m-\epsilon,u_m)$. Write $\bar{q}_u:=\e_u[q(V)]$. Because $q>s_m$ on $(u_m,1)$, we have $\bar{q}_{u_m}>s_m$, and by continuity of $u\mapsto\bar{q}_u$ we may shrink $\epsilon$ so that $\bar{q}_{u}>s_m$ for $u:=u_m-\epsilon$. Since a covariance is unaffected by shifting $h$ by the constant $r$, and $h=r$ almost everywhere on $(u_m,1)$,
\[
(1-u)\,\Cov_u\bigl(q(V),h(V)\bigr)=\int_{u}^{u_m}\bigl(q(v)-\bar{q}_u\bigr)\bigl(h(v)-r\bigr)\,dv<0,
\]
because $q(v)\le s_m<\bar{q}_u$ and $h(v)>r$ for almost every $v\in(u,u_m)$. This contradicts \eqref{eq:chebyshev_ineq_tail}. Hence $h\le r$ almost everywhere on $(u^\ast,u_m)$. Being almost everywhere non-decreasing on $(u^\ast,u_m)$, bounded above by $r$, and equal to $r$ almost everywhere on $(u_m,1)$, the function $h$ is almost everywhere non-decreasing on $(u^\ast,1)$. Lemma~\ref{lem:tail_transfer}(c) concludes the proof.

\medskip
\noindent\textit{The quasi-convex case.} The preceding argument never used that $h$ is a share. Since $g_i$ is quasi-convex if and only if $-g_i$ is quasi-concave, and $\Cov_u(q(V),-h(V))=-\Cov_u(q(V),h(V))$, the quasi-convex case follows by applying the quasi-concave case to $-h$.
\end{proof}


\begin{proof}[\textup{\textbf{Proof of Proposition~\ref{prop:R_not_L1}}}]
Because $X_i$ is absolutely continuous and independent of $Z_i$, the convolution formula gives $S$ an absolutely continuous distribution, so $\PP(S = 0) = 0$; the term $c_j\mathds{1}_{\{S=0\}}$ in $R_j$ therefore vanishes almost surely and $\e[|R_j|]=\e[|X_j|/|S|]$ for every $j\in N$. Both statements rest on the non-integrability of $1/|w|$ at the origin: the aggregate can come arbitrarily close to zero while a component stays away from it.

\smallskip
\noindent{\emph{(a)}} Write $S = X_i + X_j + Z_{ij}$ with $Z_{ij} = \sum_{k \neq i, j} X_k$ independent of $(X_i,X_j)$. Since the integrand is non-negative, Tonelli's theorem gives
\begin{equation}
\label{eq:Tonelli_a}
    \e\!\left[\frac{|X_j|}{|S|}\right]
    =
    \iiint_{\RR^3}
        \frac{|y|\,f_{X_i}(x)}{\,|x + y + z|\,}
        \;dx\;d\PP_{X_j}(y)\;d\PP_{Z_{ij}}(z).
\end{equation}
Fix $y \neq 0$ and $z \in \RR$, and substitute $w = x + y + z$ in the inner integral. With $a := \inf_{|w|\le 1} f_{X_i}(w - (y+z)) > 0$,
\[
    \int_{\RR} \frac{f_{X_i}(x)}{|x + y + z|}\,dx
    = \int_{\RR} \frac{f_{X_i}(w - (y+z))}{|w|}\,dw
    \geq a \int_{0 < |w| \le 1} \frac{dw}{|w|}
    = \infty .
\]
Integrating the remaining variables in~\eqref{eq:Tonelli_a} over the set $\{y\neq 0\}\times\RR$, which has positive $\PP_{X_j}\otimes\PP_{Z_{ij}}$-measure because $X_j\not\equiv0$, yields $\e[|X_j|/|S|]=\infty$.

\smallskip
\noindent{\emph{(b)}} Write $S = X_i + Z_i$ and apply Tonelli's theorem in the form
\begin{equation}
\label{eq:Tonelli_b}
    \e\!\left[\frac{|X_i|}{|S|}\right]
    =
    \int_{\RR}\!\left(
        \int_{\RR} \frac{|x|}{|x+z|}\,f_{X_i}(x)\,dx
    \right)d\PP_{Z_i}(z).
\end{equation}
For $z=0$ the inner integral equals $1$. For $z \neq 0$, pick $0 < \delta < |z|$ and set $a' := \inf_{|x + z|\le\delta} f_{X_i}(x) > 0$; since $|x|\ge|z|-\delta>0$ on $\{|x+z|\le\delta\}$,
\[
    \int_{\RR} \frac{|x|}{|x+z|}\,f_{X_i}(x)\,dx
    \geq (|z|-\delta)\,a'
        \int_{0 < |x+z|\le\delta} \frac{dx}{|x+z|}
    = \infty.
\]
As $\PP(Z_i \neq 0) > 0$ by hypothesis, \eqref{eq:Tonelli_b} is infinite.
\end{proof}


\begin{proof}[\textup{\textbf{Proof of Theorem~\ref{thm:EDM_coincide}}}]
Since the EDMs possess moment generating functions in a neighborhood of the origin, Corollary~\ref{cor:global_version_0} applies in its cumulant form: $\Phi_{i,u}=\Psi_{i,u}$ for all $u\in[0,1)$ if and only if $K_{X_i}=r_iK_S$ near the origin for some $r_i$, with $c_i=r_i$ when $\PP(S=0)>0$. Moreover $r_i=\Var(X_i)/\Var(S)\in(0,1)$, because $\Var(X_i)=\lambda_iA''(\theta_i)>0$.

\smallskip
\noindent{\emph{(1)}} Here $K_{X_i}(t)=\lambda_i(A(\theta+t)-A(\theta))$ and, by independence, $K_{S}(t)=(\sum_{j\in N}\lambda_j)(A(\theta+t)-A(\theta))$, so that $K_{X_i}=r_iK_S$ with $r_i=\lambda_i/\sum_{j\in N}\lambda_j$.

\smallskip
\noindent{\emph{(2), sufficiency.}} If $A$ is as in \eqref{eq:scaled_Poisson_cumulant}, then
\[
    K_{X_i}(t) = \lambda_i \frac{a}{b}\, e^{b\theta_i} \left( e^{bt} - 1 \right),
    \qquad
    K_S(t) = \Bigl( \sum_{j \in N} \lambda_j \frac{a}{b} e^{b\theta_j} \Bigr) \left( e^{bt} - 1 \right),
\]
so the factor $e^{bt}-1$ cancels in the ratio and $K_{X_i}=r_iK_S$ with $r_i$ as stated, irrespective of the heterogeneity of the $\theta_j$.

\smallskip
\noindent{\emph{(2), necessity.}} Choose a pair $(i,j)$ with $\delta:=\theta_i-\theta_j\neq0$. Since $K_{X_i}=r_iK_S$ and $K_{X_j}=r_jK_S$, we get $r_jK_{X_i}=r_iK_{X_j}$ on $\mathfrak{T}_{0}$, the common domain of the cumulant generating functions. Differentiating twice in $t$ and writing $x=\theta_j+t$,
\begin{equation}
\label{eq:functional_equation}
    A^{\prime}(x + \delta) = C\,A^{\prime}(x)
    \quad\text{and}\quad
    A^{\prime\prime}(x + \delta) = C\,A^{\prime\prime}(x),
    \qquad x\in J:=\theta_j+\mathfrak{T}_{0},
\end{equation}
where $C = r_i \lambda_j/(r_j \lambda_i) > 0$ and both $J$ and $J+\delta=\theta_i+\mathfrak{T}_0$ are non-empty open subintervals of $\operatorname{int}(\Theta)$.

Infinite divisibility turns the second equation in \eqref{eq:functional_equation} into a statement about measures. By the L\'{e}vy--Khintchine representation \citep{Sato2013} applied to $X_j$, with Gaussian variance $\sigma^2\ge0$ and L\'{e}vy measure $\nu_j$,
\[
\lambda_jA^{\prime\prime}(\theta_j+t)=K_{X_j}^{\prime\prime}(t)=\sigma^{2}+\int_{\RR} x^{2}e^{tx}\,\nu_j(\mathrm{d}x),
\]
so that, setting $\rho(\mathrm{d}x):=\lambda_j^{-1}e^{-\theta_jx}\bigl(\sigma^{2}\delta_{\{0\}}(\mathrm{d}x)+x^{2}\nu_j(\mathrm{d}x)\bigr)$,
\begin{equation}
\label{eq:Aprimeprime_transform}
A^{\prime\prime}(\theta)=\int_{\RR}e^{\theta x}\,\rho(\mathrm{d}x),\qquad \theta\in\operatorname{int}(\Theta),
\end{equation}
for a non-null positive measure $\rho$; non-nullity follows from $A^{\prime\prime}>0$. Substituting \eqref{eq:Aprimeprime_transform} into \eqref{eq:functional_equation} gives
\[
\int_{\RR}e^{\theta x}\bigl(e^{\delta x}-C\bigr)\rho(\mathrm{d}x)=0\qquad\text{for all }\theta\in J .
\]
Fix $\theta_0\in J$ and let $\varsigma(\mathrm{d}x):=e^{\theta_0x}(e^{\delta x}-C)\rho(\mathrm{d}x)$, a finite signed measure because the integrals $\int e^{\theta x}e^{\delta x}\rho$ and $\int e^{\theta x}\rho$ both converge for $\theta\in J$. Its two-sided Laplace transform is analytic on the interior of its convergence strip and vanishes on a non-degenerate real interval; by the identity theorem it vanishes on the whole strip, in particular on the imaginary axis, and the uniqueness theorem for Fourier transforms gives $\varsigma=0$. Hence $\rho$ is concentrated on $\{x\in\RR:e^{\delta x}=C\}=\{b\}$, where $b:=\delta^{-1}\log C$, so that $\rho=c\,\delta_{\{b\}}$ with $c>0$ and
\[
A^{\prime\prime}(\theta)=c\,e^{b\theta},\qquad \theta\in\operatorname{int}(\Theta).
\]
If $b=0$, then $C=e^{b\delta}=1$ and $A^{\prime}(\theta)=c\theta+\varkappa$, which upon substitution in \eqref{eq:functional_equation} yields $c\,\delta=0$, contradicting $c>0$ and $\delta\neq0$. Therefore $b\neq0$ and $C=e^{b\delta}\neq1$. Integrating once more, $A^{\prime}(\theta)=(c/b)e^{b\theta}+\varkappa$, and \eqref{eq:functional_equation} forces $\varkappa=C\varkappa$, hence $\varkappa=0$. Consequently $A^{\prime}(\theta)=ae^{b\theta}$ with $a=c/b$ and $ab=c>0$, which is \eqref{eq:scaled_Poisson_cumulant} up to an additive constant; such a constant is immaterial, as only the increments $A(\theta+t)-A(\theta)$ enter \eqref{eq:EDM_CGF}. The resulting cumulant generating function $K_{X_i}(t)=\lambda_i(a/b)e^{b\theta_i}(e^{bt}-1)$ identifies $X_i$ as $b$ times a Poisson variable with mean $\lambda_i(a/b)e^{b\theta_i}$. Finally, $\e[S]=a\sum_{j\in N}\lambda_je^{b\theta_j}>0$ forces $a>0$ and then $b>0$.

\smallskip
\noindent\emph{(2), necessity without infinite divisibility.} Replacing the equation $A^{\prime}(x+\delta)=CA^{\prime}(x)$ by the equivalent $A^{\prime}(x-\delta)=C^{-1}A^{\prime}(x)$ if necessary, we may assume $\delta>0$. Suppose first that the generating measure is bounded below; the case in which it is bounded above reduces to this one upon replacing $A$ by $\check{A}(\vartheta):=A(-\vartheta)$, again a unit cumulant function with $\check{A}^{\prime\prime}>0$, whose generating measure is the reflection of $\nu$ and which satisfies the same equation with $\delta$ replaced by $-\delta$; the argument below then returns $A^{\prime}(\theta)=ae^{b\theta}$ with $a,b<0$, which is again of the form \eqref{eq:scaled_Poisson_cumulant}. Fix $\lambda_0\in\Lambda$ and let $\nu$ be the corresponding generating measure, so that $L(\theta):=\int_{\RR}e^{\theta x}\,\nu(\mathrm{d}x)=e^{\lambda_0A(\theta)}$ on $\Theta$; write $\widetilde{A}:=\lambda_0A=\log L$, so that $\widetilde{A}^{\prime}(\theta)=\e_{\nu_\theta}[X]$ is the mean of the tilted measure $\nu_\theta(\mathrm{d}x)\propto e^{\theta x}\nu(\mathrm{d}x)$ and the first equation in \eqref{eq:functional_equation} holds for $\widetilde{A}^{\prime}$ with the same $C$.

Both sides of that equation are real-analytic on the interval $\operatorname{int}(\Theta)\cap\bigl(\operatorname{int}(\Theta)-\delta\bigr)$, which contains $J$, so the identity theorem extends it to the whole of that interval. Since $\nu$ is bounded below, $L$ is finite at every $\theta$ below any point of $\Theta$, so $\Theta$ is unbounded below and the equation may be iterated:
\[
\widetilde{A}^{\prime}(\theta-n\delta)=C^{-n}\widetilde{A}^{\prime}(\theta),\qquad n\in\NN .
\]
As $\widetilde{A}^{\prime}$ is strictly increasing, $C=1$ would give $\delta=0$, and a zero of $\widetilde{A}^{\prime}$ at $x_0$ would give another at $x_0+\delta$; hence $C\neq1$ and $\widetilde{A}^{\prime}$ has constant sign. Writing $x_-:=\inf\operatorname{supp}(\nu)>-\infty$, we have $\widetilde{A}^{\prime}\ge x_-$, while monotonicity gives $\widetilde{A}^{\prime}(\theta-n\delta)\le\widetilde{A}^{\prime}(\theta)$. If $\widetilde{A}^{\prime}>0$ the latter forces $C^{-n}\le1$ for all $n$, so $C>1$; if $\widetilde{A}^{\prime}<0$ the former bounds $C^{-n}$, so again $C>1$, and then $\widetilde{A}^{\prime}(\theta-n\delta)\to0$, contradicting $\widetilde{A}^{\prime}\le\widetilde{A}^{\prime}(\theta)<0$. Hence $C>1$ and $\widetilde{A}^{\prime}>0$.

Put $\beta:=\delta^{-1}\log C>0$. Then $P(\theta):=\log\widetilde{A}^{\prime}(\theta)-\beta\theta$ is continuous and $\delta$-periodic, so $e^{P}$ is bounded above and below by positive constants and $\widetilde{A}^{\prime}(\theta)=e^{\beta\theta}e^{P(\theta)}$ is integrable at $-\infty$. Consequently $\log L(\theta)$ has a finite limit as $\theta\to-\infty$, so $L(\theta)\to c\in(0,\infty)$. Were $x_-<0$, the mass of $\nu$ near $x_-$ would make $L(\theta)\to\infty$; and $\widetilde{A}^{\prime}(\theta-n\delta)\to0$ together with $\widetilde{A}^{\prime}\ge x_-$ gives $x_-\le0$. Hence $x_-=0$, $\nu$ is carried by $[0,\infty)$, and monotone convergence identifies $c=\nu(\{0\})>0$.

Finally, $L^{\prime}(\theta)=\int_{(0,\infty)}e^{\theta x}\rho(\mathrm{d}x)$ with $\rho(\mathrm{d}x):=x\,\nu(\mathrm{d}x)$, and
\[
e^{-\beta\theta}L^{\prime}(\theta)=e^{P(\theta)}L(\theta)
\]
is bounded above and below by positive constants as $\theta\to-\infty$. If $\rho$ charged $(0,\beta)$ this quantity would diverge, and if $\rho(\{\beta\})$ vanished it would tend to zero; hence $\rho((0,\beta))=0<\rho(\{\beta\})$ and, by dominated convergence, $e^{-\beta\theta}L^{\prime}(\theta)\to\rho(\{\beta\})$. Therefore $e^{P(\theta)}\to\rho(\{\beta\})/c$, and a continuous periodic function with a limit is constant, so that $\widetilde{A}^{\prime}(\theta)=\widetilde{a}e^{\beta\theta}$ for some $\widetilde{a}>0$ and $A^{\prime}(\theta)=ae^{b\theta}$ with $a=\widetilde{a}/\lambda_0>0$ and $b=\beta>0$, as required.
\end{proof}


\begin{proof}[\textup{\textbf{Proof of Proposition~\ref{prop:E[X_i|S]_formula}}}]
Write $S=X_i+Z_i$ with $Z_i=\sum_{j\neq i}X_j$ independent of $X_i$, so that
\[
f_{S}(s;\bm{\theta})=\int f_{X_i}(x;\theta_i)\,f_{Z_i}(s-x;\bm{\theta}_{-i})\,dx .
\]
Since $f_{X_i}(x;\theta_i)=c_i(x;\lambda_i)e^{\theta_ix-\lambda_iA_i(\theta_i)}$, we have $\partial_{\theta_i}f_{X_i}(x;\theta_i)=(x-\lambda_iA^{\prime}_i(\theta_i))f_{X_i}(x;\theta_i)$. Differentiating under the integral sign, which is legitimate by the analyticity of exponential families in the interior of the natural parameter space,
\[
\frac{\partial f_S(s;\bm{\theta})}{\partial \theta_i}
=\int \bigl(x-\lambda_iA^{\prime}_i(\theta_i)\bigr)f_{X_i}(x;\theta_i)f_{Z_i}(s-x;\bm{\theta}_{-i})\,dx
=f_S(s;\bm\theta)\bigl(\e[X_i|S=s]-\lambda_iA^{\prime}_i(\theta_i)\bigr),
\]
which is \eqref{eq:EDM_E[X_i|S]_exact}. This is the conditional score identity: the sensitivity of the aggregate log-likelihood to the $i$-th natural parameter is exactly the centered conditional mean of $X_i$.

For the second statement, recall the saddle-point expansion \citep{Daniels1954}
\[
f_S(s; \bm{\theta}) = \frac{\exp\left\{K_S\left(\th; \bm{\theta}\right) - \th s\right\}}{\sqrt{2\pi K_S^{\prime\prime}\left(\th;\bm\theta\right)}} \bigl(1+\mathcal{O}(\Lambda^{-1})\bigr),
\]
where the saddle-point $\th\equiv\th(s;\bm\theta)$ solves $K_S^{\prime}(\th)=s$; it exists and is unique for $s$ in the interior of the range of $K_S^{\prime}$, which under distributional steepness coincides with the interior of the convex hull of the support of $S$. Differentiating the logarithm of the leading factor with respect to $\theta_i$, and noting that the term generated by the implicit dependence of $\th$ on $\theta_i$ is multiplied by $K_S^{\prime}(\th)-s=0$,
\[
\frac{\partial}{\partial \theta_i}\bigl(K_S(\th;\bm\theta)-\th s\bigr)
=\left.\frac{\partial K_S(t;\bm\theta)}{\partial\theta_i}\right|_{t=\th}
=\lambda_i A^{\prime}_i\left(\theta_i + \th\right) - \lambda_i A^{\prime}_i(\theta_i).
\]
The remaining contributions are of smaller order. Indeed, $K_S^{\prime\prime}(\th)$ and $\lambda_iA^{\prime}_i(\theta_i+\th)$ are both of order $\Lambda$ in the stated regime, whereas $\partial_{\theta_i}\th=-\lambda_i(A^{\prime\prime}_i(\theta_i+\th)-A^{\prime\prime}_i(\theta_i))/K_S^{\prime\prime}(\th)$ is of order one, so that $\partial_{\theta_i}\bigl[-\tfrac12\log(2\pi K_S^{\prime\prime}(\th))\bigr]=\mathcal{O}(1)$; and the relative error $1+\mathcal{O}(\Lambda^{-1})$ is analytic in $\theta_i$ on a complex neighborhood on which the bound is uniform, so Cauchy's estimate shows that its logarithmic derivative is again $\mathcal{O}(\Lambda^{-1})$. Substituting into \eqref{eq:EDM_E[X_i|S]_exact}, the terms $\lambda_iA^{\prime}_i(\theta_i)$ cancel and
\[
\e[X_i|S=s]=\lambda_i A^{\prime}_i\left(\theta_i + \th(s)\right)+\mathcal{O}(1)=K^{\prime}_{X_i}(\th(s))\bigl(1+\mathcal{O}(\Lambda^{-1})\bigr) . \qedhere
\]
\end{proof}


\begin{proof}[\textup{\textbf{Proof of Corollary~\ref{cor:EDM_homogeneous}}}]
Under homogeneity, $S$ follows the additive EDM generated by $A$ with index $\Lambda=\sum_{j\in N}\lambda_j$ and natural parameter $\theta$, and the conditional law of $X_i$ given $S=s$ does not depend on $\theta$, because the joint density of $(X_i,S)$ factorizes as $c_i(x;\lambda_i)c_{-i}(s-x;\Lambda-\lambda_i)e^{\theta s-\Lambda A(\theta)}$. Write $\psi(s):=\e[X_i\mid S=s]$, a $\theta$-free function. Taking expectations gives
\[
\e_\theta\Bigl[\psi(S)-\frac{\lambda_i}{\Lambda}S\Bigr]=\lambda_iA^{\prime}(\theta)-\frac{\lambda_i}{\Lambda}\Lambda A^{\prime}(\theta)=0
\qquad\text{for every }\theta\in\operatorname{int}(\Theta).
\]
The family of laws of $S$ is a full natural exponential family and hence complete, so $\psi(s)=\lambda_i s/\Lambda$ for $\PP_S$-almost every $s$. Dividing by $s\neq0$ and using $c_i=\lambda_i/\Lambda$ at $s=0$ gives the claim.
\end{proof}


\begin{proof}[\textup{\textbf{Proof of Proposition~\ref{prop:EDM_gi_monotonicity}}}]
Since $s^\ast>0$ we have $\D\subset(0,\infty)$, so $\gh_i(s)=K_{X_i}^{\prime}(\th(s))/s$ is continuously differentiable on $\D\setminus\{s_{\min},s_{\max}\}$. By the quotient rule and the chain rule, together with $\th^{\prime}=1/K_S^{\prime\prime}(\th)$ from Remark~\ref{rem:T_D_bijection},
\[
\gh^{\prime}_i(s)
= \frac{s\,K_{X_i}^{\prime\prime}\left(\th(s)\right)/K_S^{\prime\prime}\left(\th(s)\right) - K_{X_i}^{\prime}\left(\th(s)\right)}{s^2}
= \frac{s\,K_{X_i}^{\prime\prime}\left(\th(s)\right) - K_{X_i}^{\prime}\left(\th(s)\right)K_S^{\prime\prime}\left(\th(s)\right)}{s^2 K_S^{\prime\prime}\left(\th(s)\right)} .
\]
The denominator is strictly positive, because $K_S^{\prime\prime}(\th)$ is the variance of $S$ under the tilt $\th$ and the model is non-degenerate. Substituting $s=K_S^{\prime}(\th(s))$, the sign of $\gh^{\prime}_i(s)$ is thus that of
\[
K_S^{\prime}\left(\th(s)\right) K_{X_i}^{\prime\prime}\left(\th(s)\right) - K_{X_i}^{\prime}\left(\th(s)\right)K_S^{\prime\prime}\left(\th(s)\right),
\]
and dividing by the strictly positive product $K_S^{\prime\prime}(\th(s))K_{X_i}^{\prime\prime}(\th(s))$ turns the sign condition into \eqref{eq:EDM_gi_monotonicity_K} evaluated at $t=\th(s)$. As $\th$ is an increasing bijection onto $\T$ by Remark~\ref{rem:T_D_bijection}, requiring the condition for all $s\in\D\setminus\{s_{\min},s_{\max}\}$ is the same as requiring it for all $t\in\T$. Finally, \eqref{eq:EDM_gi_monotonicity_A} follows from $K_{X_i}^{\prime}(t)=\lambda_iA^{\prime}_i(\theta_i+t)$ and $K_{X_i}^{\prime\prime}(t)=\lambda_iA^{\prime\prime}_i(\theta_i+t)$, with the analogous sums over $N$ for $S$.
\end{proof}


\begin{proof}[\textup{\textbf{Proof of Corollary~\ref{cor:EDM_allocations_monotonicity_general}}}]
Apply Theorem~\ref{thm:allocations_dominance_tail} with $\gh_i$ in place of $g_i$, which is legitimate because that theorem uses only the integrability \eqref{eq:L1_triplet} of the composed function and its shape, and combine it with Proposition~\ref{prop:EDM_gi_monotonicity}.
\end{proof}


\begin{proof}[\textup{\textbf{Proof of Theorem~\ref{thm:EDM_allocations_monotonicity_Ai=A}}}]
Assume $\iota$ is non-decreasing on $\Theta^{\ast}$ and $\theta_i=\theta_{\min}$. For every $t\in\T$ and every $j\in N$ both $\theta_i+t$ and $\theta_j+t$ lie in $\Theta^{\ast}$ and $\theta_i+t\le\theta_j+t$, so $\iota(\theta_j+t)\ge\iota(\theta_i+t)$. Writing the left-hand side of \eqref{eq:EDM_gi_monotonicity_A} as an average of the $\iota(\theta_j+t)$ with the strictly positive weights $\lambda_jA^{\prime\prime}(\theta_j+t)$,
\[
\frac{K_S^{\prime}(t)}{K_S^{\prime\prime}(t)}=\frac{\sum_{j\in N} \lambda_j A^{\prime\prime}(\theta_j+t)\,\iota(\theta_j+t)}{\sum_{j\in N} \lambda_j A^{\prime\prime}(\theta_j+t)} \ge \iota(\theta_i+t)=\frac{K_{X_i}^{\prime}(t)}{K_{X_i}^{\prime\prime}(t)} ,
\]
which is \eqref{eq:EDM_gi_monotonicity_K}. Proposition~\ref{prop:EDM_gi_monotonicity} then makes $\gh_i$ non-decreasing, and the sufficiency part of Theorem~\ref{thm:allocations_dominance_tail}, applied to $\gh_i$, gives $\phih_{i,u}\ge\psih_{i,u}$ on $[u^\ast,1)$; note that no quasi-concavity assumption is needed for this direction. The remaining three cases are proved in the same way, exchanging the roles of $\theta_{\min}$ and $\theta_{\max}$ and reversing the inequalities.
\end{proof}


\begin{proof}[\textup{\textbf{Proof of Proposition~\ref{prop:gamma_exact}}}]
On $(0,\infty)^n$ the joint density of $\XX$ is proportional to
\[
\prod_{j\in N}x_j^{\alpha_j-1}e^{-\beta_jx_j},
\]
and substituting $x_j=sr_j$, then restricting to the unit simplex, produces the factor
\[
s^{\sum_{j\in N}\alpha_j-1}\prod_{j\in N}r_j^{\alpha_j-1}\exp\Bigl\{-s\sum_{j\in N}\beta_jr_j\Bigr\},
\]
in which the power of $s$ does not depend on $\bm r$; normalizing gives the stated conditional density $\pi_s$. Thus $\{\pi_s\}_{s>0}$ is a one-parameter exponential family on the simplex with natural parameter $-s$ and sufficient statistic $T(\bm r)=\sum_{j\in N}\beta_jr_j$. Since the simplex is compact, differentiation under the integral sign is immediate and
\[
g_i^{\prime}(s)=\frac{d}{ds}\,\e_{\pi_s}[R_i]=-\Cov_{\pi_s}\bigl(R_i,T\bigr),
\]
which is the first equality in \eqref{eq:gamma_gi_derivative}. The second follows from $\sum_{j\in N}R_j=1$, which gives $\sum_{j\in N}\Cov(R_i,R_j\mid S=s)=0$ and hence allows $\beta_j$ to be replaced by $\beta_j-\beta_i$, the term $j=i$ then vanishing. For $n=2$, $\Cov(R_1,R_2\mid S=s)=-\Var(R_1\mid S=s)<0$, so $g_i$ is strictly increasing when $\beta_i<\beta_j$, strictly decreasing when $\beta_i>\beta_j$ and constant when $\beta_i=\beta_j$; since a monotone function is both quasi-concave and quasi-convex, Theorem~\ref{thm:allocations_dominance_tail} applies in both directions and yields the stated equivalence. For $n\ge3$ the sign of \eqref{eq:gamma_gi_derivative} is determined by the signs of $\beta_j-\beta_i$ once all conditional covariances $\Cov(R_i,R_j\mid S=s)$, $j\neq i$, are non-positive, and the sufficiency part of Theorem~\ref{thm:allocations_dominance_tail} concludes.
\end{proof}


\begin{proof}[\textup{\textbf{Proof of Theorem~\ref{thm:EDM_monotone_IR}}}]
By the L\'{e}vy--Khintchine theorem for non-negative infinitely divisible distributions \citep{Sato2013} there are a drift $\tau\ge0$ and a L\'{e}vy measure $\xi$ on $(0,\infty)$ with $\int_0^\infty\min(x,1)\,\xi(\mathrm{d}x)<\infty$ such that, for $\theta\in\operatorname{int}(\Theta)$ and up to an additive constant,
\begin{equation}
\label{eq:LK}
  A(\theta)
  = \tau\theta
    + \int_0^{\infty}\!\bigl(e^{\theta x}-1\bigr)\,\xi(\mathrm{d}x).
\end{equation}
Writing $m_k(\theta):=\int_0^{\infty} x^k e^{\theta x}\,\xi(\mathrm{d}x)$ for $k\in\{1,2,3\}$ and differentiating \eqref{eq:LK} under the integral sign, which is justified by the analyticity of $A$ on $\operatorname{int}(\Theta)$,
\begin{equation}
\label{eq:A_derivatives}
  A^{\prime}(\theta)=\tau+m_1(\theta),\qquad
  A^{\prime\prime}(\theta)=m_2(\theta),\qquad
  A^{\prime\prime\prime}(\theta)=m_3(\theta).
\end{equation}
Non-degeneracy, $A^{\prime\prime}>0$, forces $\xi\neq0$; hence $m_k(\theta)>0$ for every $k\ge1$, $A^{\prime}>0$, and $\iota$ is well defined and strictly positive. Moreover $\xi\neq0$ implies that the support of the member with index $\lambda$, which contains $\lambda\tau+x_0\NN_0$ for any $x_0$ in the support of $\xi$, is unbounded, so $s_{\max}=\infty$.

Let $\nu_\theta(\mathrm{d}x):=e^{\theta x}\xi(\mathrm{d}x)$. The Cauchy--Schwarz inequality in $L^2(\nu_\theta)$ applied to $x^{1/2}$ and $x^{3/2}$ gives
\begin{equation}
\label{eq:CS}
  m_2(\theta)^2 \leq m_1(\theta)\,m_3(\theta),
\end{equation}
whence, using $\tau\ge0$ and $m_3>0$,
\begin{equation}
\label{eq:chain}
  A^{\prime}(\theta)\,A^{\prime\prime\prime}(\theta)
  =\bigl(\tau+m_1(\theta)\bigr)m_3(\theta)
  \geq m_1(\theta)\,m_3(\theta)
  \geq m_2(\theta)^2
  =A^{\prime\prime}(\theta)^2 ,
\end{equation}
which is \eqref{eq:logconvex}. Since $A^{\prime\prime}>0$, the quotient rule gives
\[
  \iota^{\prime}(\theta)
  =
  \frac{A^{\prime\prime}(\theta)^2-A^{\prime}(\theta)\,A^{\prime\prime\prime}(\theta)}
       {A^{\prime\prime}(\theta)^2}
  \leq 0 ,
\]
so $\iota$ is non-increasing, and so is $I(t)=\iota(\theta_{\min}+t)$.

For the equality case, both inequalities in \eqref{eq:chain} are equalities at a given $\theta$ if and only if, first, $x^{1/2}$ and $x^{3/2}$ are proportional $\nu_\theta$-almost surely, that is $\xi=c\,\delta_{\{x_0\}}$ for some $c,x_0>0$, and second, $\tau m_3(\theta)=0$, that is $\tau=0$. Both conditions concern the L\'{e}vy pair $(\tau,\xi)$ alone and are therefore independent of $\theta$, so equality at one point propagates to all of $\operatorname{int}(\Theta)$. Substituting $\tau=0$ and $\xi=c\,\delta_{\{x_0\}}$ into \eqref{eq:LK} gives \eqref{eq:scaled_Poisson_unit_cumulant}, for which $A^{\prime}=cx_0e^{\theta x_0}$, $A^{\prime\prime}=cx_0^2e^{\theta x_0}$ and $\iota\equiv1/x_0$. Every other member satisfies \eqref{eq:logconvex} strictly and has $\iota$ strictly decreasing.
\end{proof}


\begin{proof}[\textup{\textbf{Proof of Corollary~\ref{cor:ordering_nonneg_ID}}}]
Parts (i) and (ii) are immediate from Theorems~\ref{thm:EDM_monotone_IR} and~\ref{thm:EDM_allocations_monotonicity_Ai=A}, since $\iota$ is non-increasing on all of $\operatorname{int}(\Theta)\supseteq\Theta^{\ast}$. For (iii), $\gh_i$ is constant on $\D\setminus\{s_{\min},s_{\max}\}$ if and only if \eqref{eq:EDM_gi_monotonicity_A} holds with equality on $\T$, that is, if and only if the average of the $\iota(\theta_j+t)$ with the strictly positive weights $\lambda_jA^{\prime\prime}(\theta_j+t)$ coincides with $\iota(\theta_i+t)$ for every $t\in\T$. Taking an index $i$ with $\theta_i=\theta_{\max}$, for which $\iota(\theta_i+t)=\min_{j\in N}\iota(\theta_j+t)$ because $\iota$ is non-increasing, forces all the $\iota(\theta_j+t)$ to coincide; letting $t$ range over $\T$ and using $\theta_{\min}<\theta_{\max}$, these overlapping equalities make $\iota$ constant on $\Theta^{\ast}$, which by Theorem~\ref{thm:EDM_monotone_IR} happens only for the scaled Poisson family. The converse is Theorem~\ref{thm:EDM_coincide}--(2). Applying Theorem~\ref{thm:tail_allocations_coincide} to $\gh_i$ converts constancy into the equality of $\phih_{i,u}$ and $\psih_{i,u}$.
\end{proof}


\begin{proof}[\textup{\textbf{Proof of Proposition~\ref{prop:comonotone_gi}}}]
Fix $u\in(0,1)$, put $s:=q_S(u)$ and let $I_s:=\{v\in(0,1):q_S(v)=s\}$, which is an interval because $q_S$ is non-decreasing, and which contains $u$. By \eqref{eq:comonotone_quantile_sum}, $\sum_{j\in N}q_{X_j}(v)=s$ for every $v\in I_s$. Each $q_{X_j}$ is non-decreasing, so a strict increase of one of them across two points of $I_s$ could not be offset by the others without one of them decreasing; hence every $q_{X_j}$ is constant on $I_s$. By the standard identity $q_S(F_S(q_S(u)))=q_S(u)$, the level $u^{\prime}:=F_S(s)$ also belongs to $I_s$ whenever $u^{\prime}<1$, and in the remaining case we use the convention $q_{X_i}(1):=\lim_{v\to1^-}q_{X_i}(v)$, which is again the common value of $q_{X_i}$ on $I_s$. In either case
\begin{equation}
\label{eq:comonotone_qXi_invariance}
q_{X_i}\bigl(F_S(q_S(u))\bigr)=q_{X_i}(u),\qquad u\in(0,1).
\end{equation}
Applying \eqref{eq:comonotone_qXi_invariance} at $u=U$ and using $S=q_S(U)$ and $X_i=q_{X_i}(U)$ almost surely gives $X_i=q_{X_i}(F_S(S))$ almost surely. In particular $X_i$ is $\sigma(S)$-measurable, so $\e[X_i\mid S]=X_i$ and, on $\{S\neq0\}$, $g_i(S)=\e[R_i\mid S]=X_i/S=q_{X_i}(F_S(S))/S$; this is \eqref{eq:comonotone_gi_exact}, and \eqref{eq:h_formula} follows upon substituting $s=q_S(u)$ and using \eqref{eq:comonotone_qXi_invariance} again. Finally, if $\PP(S=0)>0$, then $X_i=q_{X_i}(F_S(0))$ almost surely on $\{S=0\}$, so $\e[X_i\mathds{1}_{\{S=0\}}]=q_{X_i}(F_S(0))\,\PP(S=0)$, which vanishes if and only if $q_{X_i}(F_S(0))=0$.
\end{proof}


\begin{proof}[\textup{\textbf{Proof of Proposition~\ref{prop:comonotone_allocations}}}]
By \eqref{eq:h_formula} and the last part of Proposition~\ref{prop:comonotone_gi}, $q_S\,h=q_{X_i}$ almost everywhere on $(u^\ast,1)$, including on $\{q_S=0\}$, where both sides vanish. Definition~\ref{def:allocations_paradigms} then gives \eqref{eq:comonotone_allocations}, and \eqref{eq:comonotone_cov} is \eqref{eq:cov_representation}.
\end{proof}


\begin{proof}[\textup{\textbf{Proof of Corollary~\ref{cor:comonotone_coincidence}}}]
By Theorem~\ref{thm:tail_allocations_coincide} and Lemma~\ref{lem:tail_transfer}(b), $\Phi_{i,u}=\Psi_{i,u}$ for all $u\in[u^\ast,1)$ if and only if $h$ is almost everywhere equal to a constant $r_i$ on $(u^\ast,1)$. On $\{q_S\neq0\}$ this reads $q_{X_i}=r_iq_S$, and on $\{q_S=0\}$ the same identity holds trivially by Proposition~\ref{prop:comonotone_gi}; so \eqref{eq:comonotone_constancy} holds almost everywhere. Both $q_{X_i}$ and $q_S$ are left-continuous, hence so is $q_{X_i}-r_iq_S$, and a left-continuous function vanishing almost everywhere on an interval vanishes identically there, which upgrades \eqref{eq:comonotone_constancy} to every $u\in(u^\ast,1)$. Conversely, \eqref{eq:comonotone_constancy} makes $h\equiv r_i$ on $\{q_S\neq0\}$, and choosing $c_i=r_i$ extends this to all of $(u^\ast,1)$; the value $\Phi_{i,u}=\Psi_{i,u}=r_i$ follows from Theorem~\ref{thm:tail_allocations_coincide}, and $c_i=r_i$ is forced when $0\in\D$ and $\PP(S=0)>0$, because then $g_i(0)=c_i$. Finally, $X_j \stackrel{d}{=} \lambda_j X$ gives $q_{X_j}=\lambda_jq_X$ and $q_S=(\sum_{j\in N}\lambda_j)q_X$ by \eqref{eq:comonotone_quantile_sum}, whence $r_i=\lambda_i/\sum_{j\in N}\lambda_j$.
\end{proof}


\begin{proof}[\textup{\textbf{Proof of Corollary~\ref{cor:comonotone_monotonicity}}}]
By Theorem~\ref{thm:allocations_dominance_tail}, $\Phi_{i,u} \geq (\leq) \Psi_{i,u}$ for all $u \in [u^\ast, 1)$ if and only if $g_i$ is $\PP_S$-almost surely non-decreasing (non-increasing) on $\D$, which by Lemma~\ref{lem:tail_transfer}(c) and Proposition~\ref{prop:comonotone_gi} is equivalent to the monotonicity of $h=q_{X_i}/q_S$ on $(u^\ast,1)$; this is \eqref{eq:comonotone_monotonicity}. Where the two quantile functions are differentiable with positive derivatives,
\[
  h^{\prime}(u)
  = \frac{q_{X_i}^{\prime}(u)\,q_S(u) - q_{X_i}(u)\,q_S^{\prime}(u)}{q_S(u)^2},
\]
and dividing the numerator by the strictly positive product $q_S^{\prime}(u)q_{X_i}^{\prime}(u)$ turns $h^{\prime}(u)\ge0$ into \eqref{eq:comonotone_logquantile}; when moreover $q_{X_i}(u)>0$, dividing instead by the positive product $q_S(u)q_{X_i}(u)$ gives the logarithmic form. The statement about $c_i$ is Remark~\ref{rem:ci_choice_monotonicity} applied to \eqref{eq:comonotone_gi_exact}, and the normalization follows from $\sum_{i \in N} q_{X_i}(u) = q_S(u)$.
\end{proof}


\begin{proof}[\textup{\textbf{Proof of Lemma~\ref{lem:two_representations}}}]
By \eqref{eq:L1_triplet}, $q$, $h$ and $qh$ belong to $L^1(u^\ast,1)$. Hence $H$ and $G(v):=\int_v^1q(t)h(t)\,dt$, which equals $\Phi_{i,v}\int_v^1q$, are absolutely continuous on $[u^\ast,1]$, with $H(1)=G(1)=0$, $H^{\prime}=-h$ and $G^{\prime}=-qh$ almost everywhere.

(a) Since $q$ is non-decreasing and non-negative on $(u_1,1)$ and $H$ is continuous, Lebesgue--Stieltjes integration by parts gives
\[
  \int_u^1 q(v)h(v)\,dv=-\int_u^1q(v)\,dH(v)
  =q(u)H(u)-\lim_{v\to1^-}q(v)H(v)+\int_u^1 H(v)\,dq(v).
\]
The boundary term vanishes because $|q(v)H(v)|\le\int_v^1q(t)|h(t)|\,dt\to0$, the inequality using $q(v)\le q(t)$ for $t\ge v$, and \eqref{eq:IBP_identity} follows from $H(v)=(1-v)\Psi_{i,v}$. Applying the same identity with $h\equiv1$ gives $q(u)(1-u)+\int_u^1(1-v)\,dq(v)=\int_u^1q(v)\,dv$, so the two non-negative weights in \eqref{eq:IBP_identity} normalize to a probability measure after division by $\int_u^1q>0$, which is \eqref{eq:Phi_weighted_avg}.

(b) Since $h=-G^{\prime}/q$ almost everywhere and $v\mapsto-1/q(v)$ is non-decreasing on $(u_1,1)$,
\[
  H(u)=\int_u^1\frac{-G^{\prime}(v)}{q(v)}\,dv=-\int_u^1\frac{\mathrm{d}G(v)}{q(v)}
  =\frac{G(u)}{q(u)}-\lim_{v\to1^-}\frac{G(v)}{q(v)}-\int_u^1 G(v)\,\mathrm{d}\bigl(-1/q(v)\bigr),
\]
and the boundary term vanishes because $|G(v)|\le\int_v^1q|h|\to0$ while $q(v)\ge q(u)>0$. Dividing by $1-u$ and substituting $G(v)=\Phi_{i,v}\int_v^1q$ gives \eqref{eq:Psi_nonconvex_avg}. For the total mass, integrating by parts once more with $A(v):=\int_v^1q$ and $B(v):=-1/q(v)$, so that $\mathrm{d}A(v)=-q(v)\,dv$, and using $A(1^-)=0$ together with the finiteness of $B(1^-)=-1/s_{\max}$, read as $0$ when $s_{\max}=\infty$,
\[
  \int_u^1 A\,\mathrm{d}B=\bigl[A B\bigr]_u^{1^-}-\int_u^1 B(v)\,\mathrm{d}A(v)
  =\frac{A(u)}{q(u)}-\int_u^1\frac{q(v)}{q(v)}\,dv
  =\frac{A(u)}{q(u)}-(1-u),
\]
so that $\nu_u$ has total mass $A(u)/\bigl((1-u)q(u)\bigr)-1=W_u-1$.
\end{proof}


\begin{proof}[\textup{\textbf{Proof of Theorem~\ref{thm:allocation_limits}}}]
Let $u_1$ be as in Lemma~\ref{lem:two_representations}; since $q(u)\uparrow s_{\max}>0$, such a $u_1<1$ exists, and all statements concern $u\in(u_1,1)$.

(1) The measure $\mu_u$ of \eqref{eq:Phi_weighted_avg} is a probability measure on $[u,1)$, so $\Phi_{i,u}-\ell=\int_{[u,1)}(\Psi_{i,v}-\ell)\,\mu_u(\mathrm{d}v)$ and the bound follows.

(2) Since $\nu_u$ has total mass $W_u-1$, \eqref{eq:Psi_nonconvex_avg} gives
\[
  \Psi_{i,u}-\ell=W_u\bigl(\Phi_{i,u}-\ell\bigr)-\int_u^1\bigl(\Phi_{i,v}-\ell\bigr)\,\nu_u(\mathrm{d}v),
  \qquad
  \Psi_{i,u}-\Phi_{i,u}=(W_u-1)\Phi_{i,u}-\int_u^1\Phi_{i,v}\,\nu_u(\mathrm{d}v),
\]
and the two bounds in \eqref{eq:tauberian_bounds} follow by the triangle inequality, the mass of $\nu_u$ contributing the factor $W_u-1$ in each case. For (2a), write $\epsilon_u:=\sup_{v\in[u,1)}|\Phi_{i,v}-\ell|$. If $\Phi_{i,u}\to\ell$ then $\epsilon_u\to0$, and if moreover $W_u\le M<\infty$ near $1$, the first bound gives $|\Psi_{i,u}-\ell|\le(2M-1)\epsilon_u\to0$; the reverse implication is (1). For (2b), the second bound gives $\Phi_{i,u}-\Psi_{i,u}\to0$ as soon as $W_u\to1$ and $\sup_{v\ge u}|\Phi_{i,v}|$ stays bounded, which holds automatically when the components are non-negative, since then $\Phi_{i,v}\in[0,1]$.

For the sufficient conditions, recall from the identity following Lemma~\ref{lem:two_representations} that $W_u-1=\int_{q(u)}^{\infty}\Fbar_S(t)\,dt\,\big/\,\bigl((1-u)q(u)\bigr)$. Since $\Fbar_S(q(u))\le 1-u$,
\[
  1\le W_u\le 1+\frac{e\bigl(q(u)\bigr)}{q(u)},
  \qquad
  e(x):=\frac{1}{\Fbar_S(x)}\int_x^{\infty}\Fbar_S(t)\,dt ,
\]
with equality on the right when $F_S$ is continuous, $e$ being the mean excess function of $S$. If $s_{\max}<\infty$, then $q(u)\to s_{\max}>0$ and $\mathrm{ES}_u(S)\to s_{\max}$, so $W_u\to1$. Suppose next that $s_{\max}=\infty$ and that $S$ lies in the Gumbel maximum domain of attraction, so that $\Fbar_S$ is rapidly varying \citep[Section~3.3]{Embrechts1997}: $\Fbar_S(\lambda x)/\Fbar_S(x)\to0$ as $x\to\infty$, for every $\lambda>1$. Fix $\lambda>1$ and $x_0$ with $\Fbar_S(\lambda y)\le(2\lambda)^{-1}\Fbar_S(y)$ for all $y\ge x_0$; iterating gives $\Fbar_S(\lambda^kx)\le(2\lambda)^{-k}\Fbar_S(x)$ for $x\ge x_0$ and $k\ge0$, whence
\[
  \int_x^{\infty}\Fbar_S(t)\,dt
  \le\sum_{k\ge0}(\lambda-1)\lambda^k x\,\Fbar_S(\lambda^kx)
  \le 2(\lambda-1)\,x\,\Fbar_S(x),
\]
so that $e(x)\le2(\lambda-1)x$ for $x\ge x_0$. As $\lambda>1$ was arbitrary, $e(x)=o(x)$ and $W_u\to1$. Finally, if $S$ lies in the Fr\'{e}chet domain with index $\alpha>1$, then $\Fbar_S$ is regularly varying with index $-\alpha$, and Karamata's theorem \citep[Section~1.5]{Bingham1987} gives $\int_x^{\infty}\Fbar_S(t)\,dt\sim x\Fbar_S(x)/(\alpha-1)$, so that $e(x)/x\to1/(\alpha-1)$ and $\limsup_{u\to1^-}W_u\le\alpha/(\alpha-1)$, with convergence when $F_S$ is continuous.

(3) See Example~\ref{ex:oscillating_counterexample}.
\end{proof}


\begin{example}
\label{ex:oscillating_counterexample}
Let $n=2$, fix an integer $k_0\ge3$, and set
\[
  s_k := \frac{2^k}{k^2},\qquad p_k:=C\,2^{-k},\qquad C:=2^{k_0-1},\qquad k\ge k_0 .
\]
The sequence $s_k$ is strictly increasing for $k\ge k_0$, because $s_{k+1}/s_k=2k^2/(k+1)^2>1$ exactly when $k>1+\sqrt2$, and $\sum_{k\ge k_0}p_k=1$. Let $S$ have the distribution $\PP(S=s_k)=p_k$, $k\ge k_0$, so that $s_{\max}=\infty$, $\PP(S=0)=0$ and
\[
  \e[S]=\sum_{k\ge k_0}p_ks_k=C\sum_{k\ge k_0}k^{-2}<\infty .
\]
Define $g_1(s_k):=\tfrac12\bigl(1+(-1)^k\bigr)\in\{0,1\}$ and set $X_1:=S\,g_1(S)$ and $X_2:=S-X_1$. Then $0\le X_j\le S$, so all standing conditions hold with $u^\ast=0$, and $R_1=g_1(S)$ is $\sigma(S)$-measurable, so $g_1$ is indeed the conditional risk share.

Since $p_{k+1}=p_k/2$, the tail sums satisfy $\sum_{k>K}p_k=p_K$, so with $u_K:=\PP(S<s_K)$ we have $1-u_K=2p_K$ and $u_{K+1}=u_K+p_K$; hence $q(v)=s_K$ for $v\in(u_K,u_{K+1}]$. Write $u=u_K+\alpha p_K$ with $\alpha\in(0,1]$, so that $1-u=p_K(2-\alpha)$, and let $k(v)$ be the index with $v\in(u_{k(v)},u_{k(v)+1}]$. Because $g_1(s_k)=\tfrac12+\tfrac12(-1)^k$,
\begin{equation}
\label{eq:linear_decomposition_discrete}
  \Psi_{1,u}
  =\frac{1}{2}+\frac{1}{2}\,
       \frac{\int_u^1(-1)^{k(v)}\,dv}{1-u},
  \qquad
  \Phi_{1,u}
  =\frac{1}{2}+\frac{1}{2}\,
       \frac{\int_u^1 q(v)(-1)^{k(v)}\,dv}{\int_u^1 q(v)\,dv} .
\end{equation}

\emph{The flat allocation oscillates.} Summing the alternating geometric tail, $\sum_{k>K}p_k(-1)^k=C(-1)^{K+1}/(3\cdot2^{K})$, so
\[
  \int_u^1(-1)^{k(v)}\,dv
  =(1-\alpha)\,C2^{-K}(-1)^K+\frac{C(-1)^{K+1}}{3\cdot2^{K}}
  =C2^{-K}(-1)^K\Bigl(\tfrac23-\alpha\Bigr),
\]
and dividing by $1-u=C2^{-K}(2-\alpha)$,
\[
  \Psi_{1,u}=\frac12+\frac{(-1)^K}{2}\cdot\frac{\tfrac23-\alpha}{2-\alpha} .
\]
The fraction decreases from $\tfrac13$ to $-\tfrac13$ as $\alpha$ runs over $(0,1]$; at $\alpha=1$ the allocation equals $\tfrac13$ for even $K$ and $\tfrac23$ for odd $K$. Hence $\liminf_{u\to1^-}\Psi_{1,u}=\tfrac13$ and $\limsup_{u\to1^-}\Psi_{1,u}=\tfrac23$.

\emph{The severity-weighted allocation converges.} For the denominator in \eqref{eq:linear_decomposition_discrete},
\[
  \int_u^1 q(v)\,dv=\frac{C(1-\alpha)}{K^2}+C\sum_{k>K}\frac{1}{k^2}\ \ge\ \frac{C}{K+1},
\]
by the comparison $\sum_{k>K}k^{-2}\ge\int_{K+1}^\infty x^{-2}dx$. For the numerator, the alternating series estimate $\bigl|\sum_{k>K}(-1)^kk^{-2}\bigr|\le(K+1)^{-2}$ gives
\[
  \Bigl|\int_u^1 q(v)(-1)^{k(v)}\,dv\Bigr|
  \le\frac{C(1-\alpha)}{K^2}+\frac{C}{(K+1)^2}\le\frac{2C}{K^2},
\]
so that $\bigl|\Phi_{1,u}-\tfrac12\bigr|\le (K+1)/K^{2}=\mathcal{O}(K^{-1})$, uniformly in $\alpha\in(0,1]$. Thus $\Phi_{1,u}\to\tfrac12$ while $\Psi_{1,u}$ does not converge.

The mechanism is visible in the multiplier: here $\mathrm{ES}_u(S)\asymp C/\bigl(K\,(1-u)\bigr)$ while $q(u)=s_K=2^K/K^2$, so $W_u=K/(2-\alpha)+\mathcal{O}(1)\to\infty$, and the negative mass $W_u-1$ in \eqref{eq:Psi_nonconvex_avg} grows without bound. The alternation of $g_1$ along the quantile curve is then amplified rather than smoothed, which is exactly what Theorem~\ref{thm:allocation_limits}--(2) excludes when $W_u\to1$.
\end{example}

\begin{example}
\label{ex:oscillating_merge}
Part~\textup{(2b)} of Theorem~\ref{thm:allocation_limits} makes the two allocations merge, but it does not make either of them converge. Let $n=2$, let $S$ be standard exponential, and fix $\omega>0$. With $q(v)=-\log(1-v)$, set
\[
  g_1(s):=\tfrac12\bigl(1+\sin(\omega s)\bigr)\in[0,1],
  \qquad X_1:=S\,g_1(S),\qquad X_2:=S-X_1 .
\]
The components are non-negative with $0\le X_j\le S$ and $\e[S]=1$, so the standing conditions hold with $u^\ast=0$, and $R_1=g_1(S)$ is $\sigma(S)$-measurable, so $g_1$ is the conditional risk share. Here $\mathrm{ES}_u(S)=q(u)+1$, whence
\[
  W_u=1+\frac{1}{q(u)}\ \longrightarrow\ 1 ,
\]
and Part~\textup{(2b)} applies. Writing $y:=q(u)$ and substituting $v=1-e^{-y}$,
\[
  \Psi_{1,u}=\int_0^\infty g_1(y+t)\,e^{-t}\,dt
  =\frac12+\frac12\cdot\frac{\sin(\omega y)+\omega\cos(\omega y)}{1+\omega^{2}} ,
\]
which oscillates forever between $\tfrac12\pm\tfrac12(1+\omega^2)^{-1/2}$: the expected risk share has no limit. For the severity-weighted allocation,
\[
  \Phi_{1,u}=\frac{\int_y^\infty t\,g_1(t)e^{-t}\,dt}{\int_y^\infty t\,e^{-t}\,dt}
  =\frac12+\frac{1}{2(1+y)}\,
   \operatorname{Im}\left\{e^{\i\omega y}\left[\frac{y}{1-\i\omega}+\frac{1}{(1-\i\omega)^{2}}\right]\right\},
\]
so that $\Phi_{1,u}$ oscillates with the same asymptotic amplitude and phase as $\Psi_{1,u}$, and
\[
  \Phi_{1,u}-\Psi_{1,u}=\mathcal{O}\bigl(q(u)^{-1}\bigr),
\]
in accordance with the bound $2(W_u-1)=2/q(u)$. Both allocations therefore oscillate indefinitely, with identical $\liminf$ and $\limsup$, while their difference vanishes: flattening of the multiplier reconciles the two paradigms without stabilizing either. Taking $\omega\to0$ makes the common oscillation amplitude approach $\tfrac12$, its largest possible value for non-negative risks.
\end{example}

\begin{proof}[\textup{\textbf{Proof of Corollary~\ref{cor:tauberian}}}]
The implication $(\Leftarrow)$ is Theorem~\ref{thm:allocation_limits}--(1). For $(\Rightarrow)$, assume $g_i$ non-decreasing on $\D$; by Lemma~\ref{lem:tail_transfer}(c), $h$ is almost everywhere non-decreasing on $(u^\ast,1)$, so $\Psi_{i,u}=\frac{1}{1-u}\int_u^1h(v)\,dv\ge h(u)$ for almost every $u$, and \eqref{eq:Psi_derivative} gives
\[
  \frac{d}{du}\Psi_{i,u}=\frac{\Psi_{i,u}-h(u)}{1-u}\ \ge\ 0
  \qquad\text{for almost every }u\in(u^\ast,1).
\]
Being absolutely continuous on compact subintervals, $u\mapsto\Psi_{i,u}$ is therefore non-decreasing, so $\Psi_{i,v}\ge\Psi_{i,u}$ for $v\ge u$ and \eqref{eq:Phi_weighted_avg} yields $\Phi_{i,u}\ge\Psi_{i,u}$, in accordance with Theorem~\ref{thm:allocations_dominance_tail}. If now $\Phi_{i,u}\to\ell$, then $\Psi_{i,\cdot}$ is non-decreasing and bounded above by $\sup_{v\ge u}\Phi_{i,v}<\infty$, hence converges to some $\ell^{\prime}$; Part~(1) applied to $\Psi_{i,u}\to\ell^{\prime}$ forces $\ell^{\prime}=\ell$. The non-increasing case is symmetric.
\end{proof}


\begin{proof}[\textup{\textbf{Proof of Corollary~\ref{cor:limit_identification}}}]
The uniform average of a function converging at the endpoint inherits its limit: given $\varepsilon>0$, choose $u_\varepsilon$ with $|h(v)-\ell|\le\varepsilon$ for $v>u_\varepsilon$, so that $|\Psi_{i,u}-\ell|\le\frac{1}{1-u}\int_u^1|h(v)-\ell|\,dv\le\varepsilon$ for $u>u_\varepsilon$. Theorem~\ref{thm:allocation_limits}--(1) then gives $\Phi_{i,u}\to\ell$. For the two sufficient conditions, $q(v)\uparrow s_{\max}$ as $v\to1^-$, with $q(v)<s_{\max}$ for all $v<1$ when $\PP(S=s_{\max})=0$ and $q(v)=s_{\max}$ for $v$ beyond $\PP(S<s_{\max})$ otherwise.
\end{proof}

\end{appendices}
\end{document}